%% file: main.tex
\documentclass[a4paper,UKenglish,cleveref, autoref, thm-restate]{lipics-v2021}

\title{Connectivity at the crossroad of intuitionistic and classical polarizations in linear logic} 

\author{Raffaele {Di Donna}}{Université Paris Cité, France \and Università Roma Tre, Italy \and \url{https://www.irif.fr/~didonna} }{didonna@irif.fr}{}{}

\author{Giulio {Guerrieri}}{University of Sussex, UK\and \url{https://pageperso.lis-lab.fr/~giulio.guerrieri/}}{g.guerrieri@sussex.ac.uk}{}{}

\author{Lorenzo {Tortora de Falco}}{Università Roma Tre, Italy \and \url{http://logica.uniroma3.it/~tortora/}}{tortora@uniroma3.it}{}{}

\authorrunning{R.~{Di Donna}, G.~{Guerrieri} and L.~{Tortora de Falco}} 

\Copyright{Raffaele {Di Donna}, Giulio {Guerrieri} and Lorenzo {Tortora de Falco}} 

\ccsdesc[500]{Theory of computation~Linear logic}
\ccsdesc[500]{Theory of computation~Proof theory}

\keywords{linear logic, proof-net, sequentialization, bang calculus, polarization} 

\category{} 

\relatedversion{} 

\nolinenumbers 

\usepackage{amsmath}
\usepackage{amsthm}
\usepackage{amssymb}
\usepackage{mathtools}
\usepackage{anyfontsize}
\usepackage{ebproof}
\usepackage{stmaryrd}
\SetSymbolFont{stmry}{bold}{U}{stmry}{m}{n}
\usepackage{cmll}
\usepackage{graphicx}
\usepackage[table]{xcolor}
\definecolor{headcolor}{RGB}{234,236,240}
\definecolor{rowcolor}{RGB}{248,249,250}
\usepackage{tikz}
\usetikzlibrary{calc}
\usepackage{tikzit}
\def\infcstrut{\vrule width 0pt height 8.125pt depth 1.875pt}
\definecolor{myred}{RGB}{179,0,0}
\definecolor{myblue}{RGB}{0,0,179}
\definecolor{mygreen}{RGB}{0,128,0}
\definecolor{mypurple}{RGB}{222,0,174}
\definecolor{myyellow}{RGB}{230,170,0}
\definecolor{myorange}{RGB}{255,149,0}
\usepackage{hyperref}
\hypersetup{
    colorlinks,
    linkcolor={myred},
    citecolor={mygreen},
    urlcolor={myblue}
}
\usepackage{array,booktabs}
\usepackage{xspace}
\usepackage{adjustbox}

\input{style.tikzstyles}

\makeatletter
\tikzset{
    dot diameter/.store in=\dot@diameter,
    dot diameter=0.25ex,
    dot spacing/.store in=\dot@spacing,
    dot spacing=1ex,
    dots/.style={
        line width=\dot@diameter,
        line cap=round,
        dash pattern=on 0pt off \dot@spacing
    }
}
\makeatother

\tikzset{
    deduceLine/.style = {line width=1.1pt, dots}
}

\declaretheorem[sibling=theorem]{notation}

\makeatletter
\def\@textbottom{\vskip \z@ \@plus 2pt}
\let\@texttop\relax
\makeatother

\makeatletter
\g@addto@macro{\UrlBreaks}{%
    \do\/\do\a\do\b\do\c\do\d\do\e\do\f%
    \do\g\do\h\do\i\do\j\do\k\do\l\do\m%
    \do\n\do\o\do\p\do\q\do\r\do\s\do\t%
    \do\u\do\v\do\w\do\x\do\y\do\z%
    \do\A\do\B\do\C\do\D\do\E\do\F\do\G%
    \do\H\do\I\do\J\do\K\do\L\do\M\do\N%
    \do\O\do\P\do\Q\do\R\do\S\do\T\do\U%
    \do\V\do\W\do\X\do\Y\do\Z}
\makeatother

\makeatletter
\g@addto@macro\bfseries{\boldmath}
\def\thmhead@plain#1#2#3{%
    \thmname{#1}\thmnumber{\@ifnotempty{#1}{ }\@upn{#2}}%
    \thmnote{ {\the\thm@notefont\unboldmath(#3)}}}
\let\thmhead\thmhead@plain
\makeatother

\makeatletter
\let\save@mathaccent\mathaccent
\newcommand*\if@single[3]{%
    \setbox0\hbox{${\mathaccent"0362{#1}}^H$}%
    \setbox2\hbox{${\mathaccent"0362{\kern0pt#1}}^H$}%
    \ifdim\ht0=\ht2 #3\else #2\fi
}
\newcommand*\rel@kern[1]{\kern#1\dimexpr\macc@kerna}
\newcommand*\widebar[1]{\@ifnextchar^{{\wide@bar{#1}{0}}}{\wide@bar{#1}{1}}}
\newcommand*\wide@bar[2]{\if@single{#1}{\wide@bar@{#1}{#2}{1}}{\wide@bar@{#1}{#2}{2}}}
\newcommand*\wide@bar@[3]{%
    \begingroup
    \def\mathaccent##1##2{%
        \let\mathaccent\save@mathaccent
        \if#32 \let\macc@nucleus\first@char \fi
        \setbox\z@\hbox{$\macc@style{\macc@nucleus}_{}$}%
        \setbox\tw@\hbox{$\macc@style{\macc@nucleus}{}_{}$}%
        \dimen@\wd\tw@
        \advance\dimen@-\wd\z@
        \divide\dimen@ 3
        \@tempdima\wd\tw@
        \advance\@tempdima-\scriptspace
        \divide\@tempdima 10
        \advance\dimen@-\@tempdima
        \ifdim\dimen@>\z@ \dimen@0pt\fi
        \rel@kern{0.6}\kern-\dimen@
        \if#31
        \overline{\rel@kern{-0.6}\kern\dimen@\macc@nucleus\rel@kern{0.4}\kern\dimen@}%
        \advance\dimen@0.4\dimexpr\macc@kerna
        \let\final@kern#2%
        \ifdim\dimen@<\z@ \let\final@kern1\fi
        \if\final@kern1 \kern-\dimen@\fi
        \else
        \overline{\rel@kern{-0.6}\kern\dimen@#1}%
        \fi
    }%
    \macc@depth\@ne
    \let\math@bgroup\@empty \let\math@egroup\macc@set@skewchar
    \mathsurround\z@ \frozen@everymath{\mathgroup\macc@group\relax}%
    \macc@set@skewchar\relax
    \let\mathaccentV\macc@nested@a
    \if#31
    \macc@nested@a\relax111{#1}%
    \else
    \def\gobble@till@marker##1\endmarker{}%
    \futurelet\first@char\gobble@till@marker#1\endmarker
    \ifcat\noexpand\first@char A\else
    \def\first@char{}%
    \fi
    \macc@nested@a\relax111{\first@char}%
    \fi
    \endgroup
}
\makeatother

\DeclareFontFamily{U}{mathx}{}
\DeclareFontShape{U}{mathx}{m}{n}{<-> mathx10}{}
\DeclareSymbolFont{mathx}{U}{mathx}{m}{n}
\DeclareMathAccent{\widehat}{0}{mathx}{"70}
\DeclareMathAccent{\widecheck}{0}{mathx}{"71}

\def\rddots{\scriptstyle\cdot^{\scriptstyle\cdot^{\scriptstyle\cdot}}}

\newlength{\centralwd}

\newcommand{\centerwithright}[3][0.25em]{%
    \settowidth{\centralwd}{#2}%
    \makebox[0pt][c]{#2}%
    \makebox[0pt][l]{\hspace{\dimexpr .5\centralwd + #1\relax}#3}%
}

\newcommand{\centerwithleft}[3][0.25em]{%
    \settowidth{\centralwd}{#2}%
    \makebox[0pt][r]{#3\hspace{\dimexpr .5\centralwd + #1\relax}}%
    \makebox[0pt][c]{#2}%
}

\newcounter{subsubfigure}

\newcommand{\subsubcaptionbox}[2]{%
    \stepcounter{subsubfigure}%
    \begingroup
    \renewcommand{\thesubfigure}{\alph{subfigure}\arabic{subsubfigure}}%
    \subcaptionbox{#1}{#2}%
    \endgroup
    \addtocounter{subfigure}{-1}%
}

\renewcommand{\epsilon}{\varepsilon}
\newcommand{\emptysetAlt}{\varnothing}

\theoremstyle{definition}
\newtheorem*{outline}{Outline}

\EventEditors{John Q. Open and Joan R. Access}
\EventNoEds{2}
\EventLongTitle{42nd Conference on Very Important Topics (CVIT 2016)}
\EventShortTitle{CVIT 2016}
\EventAcronym{CVIT}
\EventYear{2016}
\EventDate{December 24--27, 2016}
\EventLocation{Little Whinging, United Kingdom}
\EventLogo{}
\SeriesVolume{42}
\ArticleNo{23}

 \input{macros}

\begin{document}

 \maketitle

 \begin{abstract}
  We investigate a property that extends the Danos-Regnier correctness criterion for linear logic \pss. The property applies to the correctness graphs of a \ps: it states that any such graph is acyclic and the number of its connected components is exactly one more than the number of nodes bottom or weakening. This is known to be necessary but not sufficient in multiplicative exponential linear logic (\MELL) to recover a \seqcp from a \ps. We present a geometric restriction on \pss allowing us to turn this necessary property into a sufficient one, computationally efficient: we can thus introduce the notable fragment \VMELL of \MELL for which the property is indeed a correctness criterion. The fragment \VMELL brings together the classical and intuitionistic polarizations. We translate the \bc terms into \pnst{\VMELL}, factorize the usual translations in linear logic of the \cbn and \cbv \lci, prove that \cutelim simulates bang reduction, and provide an explicit characterization of the \bc terms as \pns.
 \end{abstract}

 \section{Introduction}
  \input{Sections/Introduction}
  
 \section{Geometric preliminaries: graphs and untyped \texorpdfstring{\pss}{proof-structures}}
  \label{sec:geometric-preliminaries}
  \input{Sections/Geometric}
  
 \section{Pure polarized \texorpdfstring{\pss}{proof-structures}}
  \label{sec:pure-polarized}
  \input{Sections/Polarized}
  
 \section{The typed setting}
  \label{sec:typed-setting}
  \input{Sections/Typed}
 
 \section{\texorpdfstring{\Pns}{Proof-nets} for the \texorpdfstring{\bc}{bang calculus} and \texorpdfstring{\lci}{lambda-calculi}}
  \label{sec:bang-calculus}
  \input{Sections/Bang}

 \bibliography{Bibliography}
 
 \newpage
 
 \appendix
 
 \crefalias{subsection}{appendix}
 \crefalias{subsubsection}{appendix}
 
 \section{Technical appendix}
  \input{Sections/Appendix}
 
\end{document}

%% file: style.tikzstyles
\tikzstyle{form}=[inner sep=1.5pt, tikzit fill={rgb,255: red,255; green,224; blue,46}]
\tikzstyle{dot}=[font={$\bullet$}, tikzit fill={rgb,255: red,255; green,224; blue,46}]
\tikzstyle{ax}=[draw=black, anchor=north, trapezium, trapezium angle=110, rounded corners=2pt, inner sep=1.2pt, fill=white, font={$\scriptstyle \mathtt{ax}$}, tikzit fill={rgb,255: red,255; green,73; blue,76}, tikzit shape=rectangle]
\tikzstyle{bot}=[draw=black, anchor=north, trapezium, trapezium angle=110, rounded corners=2pt, inner sep=1.2pt, fill=white, font={$\scriptstyle \bot$}, tikzit fill={rgb,255: red,154; green,95; blue,230}]
\tikzstyle{bang}=[draw=black, anchor=north, trapezium, trapezium angle=110, rounded corners=2pt, inner sep=1.2pt, fill=white, font={$\scriptstyle !$}, tikzit fill={rgb,255: red,0; green,230; blue,205}]
\tikzstyle{cut}=[draw=black, anchor=north, trapezium, trapezium angle=110, rounded corners=2pt, inner sep=1.2pt, fill=white, font={$\scriptstyle \mathtt{cut}$}, tikzit fill={rgb,255: red,11; green,96; blue,255}, tikzit shape=rectangle]
\tikzstyle{der}=[draw=black, anchor=north, trapezium, trapezium angle=110, rounded corners=2pt, inner sep=1.2pt, fill=white, font={$\scriptstyle \mathtt{?d}$}, tikzit fill={rgb,255: red,109; green,90; blue,90}]
\tikzstyle{lab}=[draw=black, anchor=north, trapezium, trapezium angle=110, rounded corners=2pt, inner sep=1.2pt, fill=white, font={$\scriptstyle l$}, tikzit fill={rgb,255: red,0; green,180; blue,220}]
\tikzstyle{lab2}=[draw=black, anchor=north, trapezium, trapezium angle=110, rounded corners=2pt, inner sep=1.2pt, fill=white, font={$\scriptstyle\hskip 0.125em\smash{l'}\vphantom{l} \hskip -0.125em$}, tikzit fill={rgb,255: red,0; green,180; blue,220}]
\tikzstyle{one}=[draw=black, anchor=north, trapezium, trapezium angle=110, rounded corners=2pt, inner sep=1.2pt, fill=white, font={$\scriptstyle \mathtt{1}$}, tikzit fill={rgb,255: red,255; green,170; blue,30}]
\tikzstyle{par}=[draw=black, anchor=north, trapezium, trapezium angle=110, rounded corners=2pt, inner sep=1.2pt, fill=white, font={$\scriptstyle \parr$}, tikzit fill={rgb,255: red,0; green,255; blue,213}]
\tikzstyle{tens}=[draw=black, anchor=north, trapezium, trapezium angle=110, rounded corners=2pt, inner sep=1.2pt, fill=white, font={$\scriptstyle \varotimes$}, tikzit fill={rgb,255: red,12; green,255; blue,0}]
\tikzstyle{why}=[draw=black, anchor=north, trapezium, trapezium angle=110, rounded corners=2pt, inner sep=1.2pt, fill=white, font={$\scriptstyle ?$}, tikzit fill={rgb,255: red,200; green,255; blue,90}]
\tikzstyle{wtxt}=[text=white, tikzit fill={rgb,255: red,238; green,238; blue,238}]
\tikzstyle{lowL}=[yshift=-0.8425ex, tikzit fill={rgb,255: red,120; green,120; blue,120}]
\tikzstyle{lowR}=[yshift=-0.885ex, tikzit fill={rgb,255: red,80; green,80; blue,80}]
\tikzstyle{rect}=[draw=black, anchor=north, inner sep=1.5pt, shape=rectangle, tikzit fill={rgb,255: red,90; green,150; blue,255}]
\tikzstyle{both}=[draw=black, anchor=north, trapezium, trapezium angle=110, rounded corners=2pt, inner sep=1.2pt, fill=pink, font={$\scriptstyle \bot$}, tikzit fill={rgb,255: red,255; green,120; blue,180}]
\tikzstyle{parh}=[draw=black, anchor=north, trapezium, trapezium angle=110, rounded corners=2pt, inner sep=1.2pt, fill=pink, font={$\scriptstyle \parr$}, tikzit fill={rgb,255: red,255; green,120; blue,200}]
\tikzstyle{formr}=[inner sep=1.5pt, text=myred, tikzit fill={rgb,255: red,230; green,45; blue,55}]
\tikzstyle{dotr}=[font={$\textcolor{myred}{\bullet}$}, tikzit fill={rgb,255: red,230; green,45; blue,55}]
\tikzstyle{axr}=[draw=myred, anchor=north, trapezium, trapezium angle=110, rounded corners=2pt, inner sep=1.2pt, fill=white, font={$\scriptstyle \textcolor{myred}{\mathtt{ax}}$}, tikzit fill={rgb,255: red,230; green,45; blue,55}, tikzit shape=rectangle]
\tikzstyle{botr}=[draw=myred, anchor=north, trapezium, trapezium angle=110, rounded corners=2pt, inner sep=1.2pt, fill=white, font={$\scriptstyle \textcolor{myred}{\bot}$}, tikzit fill={rgb,255: red,230; green,45; blue,55}]
\tikzstyle{bangr}=[draw=myred, anchor=north, trapezium, trapezium angle=110, rounded corners=2pt, inner sep=1.2pt, fill=white, font={$\scriptstyle \textcolor{myred}{!}$}, tikzit fill={rgb,255: red,230; green,45; blue,55}]
\tikzstyle{cutr}=[draw=myred, anchor=north, trapezium, trapezium angle=110, rounded corners=2pt, inner sep=1.2pt, fill=white, font={$\scriptstyle \textcolor{myred}{\mathtt{cut}}$}, tikzit fill={rgb,255: red,230; green,45; blue,55}, tikzit shape=rectangle]
\tikzstyle{derr}=[draw=myred, anchor=north, trapezium, trapezium angle=110, rounded corners=2pt, inner sep=1.2pt, fill=white, font={$\scriptstyle \textcolor{myred}{\mathtt{?d}}$}, tikzit fill={rgb,255: red,230; green,45; blue,55}]
\tikzstyle{labr}=[draw=myred, anchor=north, trapezium, trapezium angle=110, rounded corners=2pt, inner sep=1.2pt, fill=white, font={$\scriptstyle \textcolor{myred}{l}$}, tikzit fill={rgb,255: red,230; green,45; blue,55}]
\tikzstyle{lab2r}=[draw=myred, anchor=north, trapezium, trapezium angle=110, rounded corners=2pt, inner sep=1.2pt, fill=white, font={$\scriptstyle \textcolor{myred}{\hskip 0.125em\smash{l'}\vphantom{l} \hskip -0.125em}$}, tikzit fill={rgb,255: red,230; green,45; blue,55}]
\tikzstyle{oner}=[draw=myred, anchor=north, trapezium, trapezium angle=110, rounded corners=2pt, inner sep=1.2pt, fill=white, font={$\scriptstyle \textcolor{myred}{\mathtt{1}}$}, tikzit fill={rgb,255: red,230; green,45; blue,55}]
\tikzstyle{parr}=[draw=myred, anchor=north, trapezium, trapezium angle=110, rounded corners=2pt, inner sep=1.2pt, fill=white, font={$\scriptstyle \textcolor{myred}{\parr}$}, tikzit fill={rgb,255: red,230; green,45; blue,55}]
\tikzstyle{tensr}=[draw=myred, anchor=north, trapezium, trapezium angle=110, rounded corners=2pt, inner sep=1.2pt, fill=white, font={$\scriptstyle \textcolor{myred}{\varotimes}$}, tikzit fill={rgb,255: red,230; green,45; blue,55}]
\tikzstyle{whyr}=[draw=myred, anchor=north, trapezium, trapezium angle=110, rounded corners=2pt, inner sep=1.2pt, fill=white, font={$\scriptstyle \textcolor{myred}{?}$}, tikzit fill={rgb,255: red,230; green,45; blue,55}]
\tikzstyle{wtxtr}=[text=white, text=myred, tikzit fill={rgb,255: red,230; green,45; blue,55}]
\tikzstyle{lowLr}=[yshift=-0.8425ex, text=myred, draw=myred, tikzit fill={rgb,255: red,230; green,45; blue,55}]
\tikzstyle{lowRr}=[yshift=-0.885ex, text=myred, draw=myred, tikzit fill={rgb,255: red,230; green,45; blue,55}]
\tikzstyle{rectr}=[draw=myred, text=myred, anchor=north, inner sep=1.5pt, shape=rectangle, tikzit fill={rgb,255: red,230; green,45; blue,55}]
\tikzstyle{bothr}=[draw=black, anchor=north, trapezium, trapezium angle=110, rounded corners=2pt, inner sep=1.2pt, fill=myred, font={$\scriptstyle \bot$}, tikzit fill={rgb,255: red,230; green,45; blue,55}]
\tikzstyle{parhr}=[draw=black, anchor=north, trapezium, trapezium angle=110, rounded corners=2pt, inner sep=1.2pt, fill=myred, font={$\scriptstyle \parr$}, tikzit fill={rgb,255: red,230; green,45; blue,55}]
\tikzstyle{formb}=[inner sep=1.5pt, text=myblue, tikzit fill={rgb,255: red,40; green,105; blue,230}]
\tikzstyle{dotb}=[font={$\textcolor{myblue}{\bullet}$}, tikzit fill={rgb,255: red,40; green,105; blue,230}]
\tikzstyle{axb}=[draw=myblue, anchor=north, trapezium, trapezium angle=110, rounded corners=2pt, inner sep=1.2pt, fill=white, font={$\scriptstyle \textcolor{myblue}{\mathtt{ax}}$}, tikzit fill={rgb,255: red,40; green,105; blue,230}, tikzit shape=rectangle]
\tikzstyle{botb}=[draw=myblue, anchor=north, trapezium, trapezium angle=110, rounded corners=2pt, inner sep=1.2pt, fill=white, font={$\scriptstyle \textcolor{myblue}{\bot}$}, tikzit fill={rgb,255: red,40; green,105; blue,230}]
\tikzstyle{bangb}=[draw=myblue, anchor=north, trapezium, trapezium angle=110, rounded corners=2pt, inner sep=1.2pt, fill=white, font={$\scriptstyle \textcolor{myblue}{!}$}, tikzit fill={rgb,255: red,40; green,105; blue,230}]
\tikzstyle{cutb}=[draw=myblue, anchor=north, trapezium, trapezium angle=110, rounded corners=2pt, inner sep=1.2pt, fill=white, font={$\scriptstyle \textcolor{myblue}{\mathtt{cut}}$}, tikzit fill={rgb,255: red,40; green,105; blue,230}, tikzit shape=rectangle]
\tikzstyle{derb}=[draw=myblue, anchor=north, trapezium, trapezium angle=110, rounded corners=2pt, inner sep=1.2pt, fill=white, font={$\scriptstyle \textcolor{myblue}{\mathtt{?d}}$}, tikzit fill={rgb,255: red,40; green,105; blue,230}]
\tikzstyle{labb}=[draw=myblue, anchor=north, trapezium, trapezium angle=110, rounded corners=2pt, inner sep=1.2pt, fill=white, font={$\scriptstyle \textcolor{myblue}{l}$}, tikzit fill={rgb,255: red,40; green,105; blue,230}]
\tikzstyle{lab2b}=[draw=myblue, anchor=north, trapezium, trapezium angle=110, rounded corners=2pt, inner sep=1.2pt, fill=white, font={$\scriptstyle \textcolor{myblue}{\hskip 0.125em\smash{l'}\vphantom{l} \hskip -0.125em}$}, tikzit fill={rgb,255: red,40; green,105; blue,230}]
\tikzstyle{oneb}=[draw=myblue, anchor=north, trapezium, trapezium angle=110, rounded corners=2pt, inner sep=1.2pt, fill=white, font={$\scriptstyle \textcolor{myblue}{\mathtt{1}}$}, tikzit fill={rgb,255: red,40; green,105; blue,230}]
\tikzstyle{parb}=[draw=myblue, anchor=north, trapezium, trapezium angle=110, rounded corners=2pt, inner sep=1.2pt, fill=white, font={$\scriptstyle \textcolor{myblue}{\parr}$}, tikzit fill={rgb,255: red,40; green,105; blue,230}]
\tikzstyle{tensb}=[draw=myblue, anchor=north, trapezium, trapezium angle=110, rounded corners=2pt, inner sep=1.2pt, fill=white, font={$\scriptstyle \textcolor{myblue}{\varotimes}$}, tikzit fill={rgb,255: red,40; green,105; blue,230}]
\tikzstyle{whyb}=[draw=myblue, anchor=north, trapezium, trapezium angle=110, rounded corners=2pt, inner sep=1.2pt, fill=white, font={$\scriptstyle \textcolor{myblue}{?}$}, tikzit fill={rgb,255: red,40; green,105; blue,230}]
\tikzstyle{wtxtb}=[text=white, text=myblue, tikzit fill={rgb,255: red,40; green,105; blue,230}]
\tikzstyle{lowLb}=[yshift=-0.8425ex, text=myblue, draw=myblue, tikzit fill={rgb,255: red,40; green,105; blue,230}]
\tikzstyle{lowRb}=[yshift=-0.885ex, text=myblue, draw=myblue, tikzit fill={rgb,255: red,40; green,105; blue,230}]
\tikzstyle{rectb}=[draw=myblue, text=myblue, anchor=north, inner sep=1.5pt, shape=rectangle, tikzit fill={rgb,255: red,40; green,105; blue,230}]
\tikzstyle{bothb}=[draw=black, anchor=north, trapezium, trapezium angle=110, rounded corners=2pt, inner sep=1.2pt, fill=myblue, font={$\scriptstyle \bot$}, tikzit fill={rgb,255: red,40; green,105; blue,230}]
\tikzstyle{parhb}=[draw=black, anchor=north, trapezium, trapezium angle=110, rounded corners=2pt, inner sep=1.2pt, fill=myblue, font={$\scriptstyle \parr$}, tikzit fill={rgb,255: red,40; green,105; blue,230}]
\tikzstyle{formg}=[inner sep=1.5pt, text=mygreen, tikzit fill={rgb,255: red,35; green,170; blue,85}]
\tikzstyle{dotg}=[font={$\textcolor{mygreen}{\bullet}$}, tikzit fill={rgb,255: red,35; green,170; blue,85}]
\tikzstyle{axg}=[draw=mygreen, anchor=north, trapezium, trapezium angle=110, rounded corners=2pt, inner sep=1.2pt, fill=white, font={$\scriptstyle \textcolor{mygreen}{\mathtt{ax}}$}, tikzit fill={rgb,255: red,35; green,170; blue,85}, tikzit shape=rectangle]
\tikzstyle{botg}=[draw=mygreen, anchor=north, trapezium, trapezium angle=110, rounded corners=2pt, inner sep=1.2pt, fill=white, font={$\scriptstyle \textcolor{mygreen}{\bot}$}, tikzit fill={rgb,255: red,35; green,170; blue,85}]
\tikzstyle{bangg}=[draw=mygreen, anchor=north, trapezium, trapezium angle=110, rounded corners=2pt, inner sep=1.2pt, fill=white, font={$\scriptstyle \textcolor{mygreen}{!}$}, tikzit fill={rgb,255: red,35; green,170; blue,85}]
\tikzstyle{cutg}=[draw=mygreen, anchor=north, trapezium, trapezium angle=110, rounded corners=2pt, inner sep=1.2pt, fill=white, font={$\scriptstyle \textcolor{mygreen}{\mathtt{cut}}$}, tikzit fill={rgb,255: red,35; green,170; blue,85}, tikzit shape=rectangle]
\tikzstyle{derg}=[draw=mygreen, anchor=north, trapezium, trapezium angle=110, rounded corners=2pt, inner sep=1.2pt, fill=white, font={$\scriptstyle \textcolor{mygreen}{\mathtt{?d}}$}, tikzit fill={rgb,255: red,35; green,170; blue,85}]
\tikzstyle{labg}=[draw=mygreen, anchor=north, trapezium, trapezium angle=110, rounded corners=2pt, inner sep=1.2pt, fill=white, font={$\scriptstyle \textcolor{mygreen}{l}$}, tikzit fill={rgb,255: red,35; green,170; blue,85}]
\tikzstyle{lab2g}=[draw=mygreen, anchor=north, trapezium, trapezium angle=110, rounded corners=2pt, inner sep=1.2pt, fill=white, font={$\scriptstyle \textcolor{mygreen}{\hskip 0.125em\smash{l'}\vphantom{l} \hskip -0.125em}$}, tikzit fill={rgb,255: red,35; green,170; blue,85}]
\tikzstyle{oneg}=[draw=mygreen, anchor=north, trapezium, trapezium angle=110, rounded corners=2pt, inner sep=1.2pt, fill=white, font={$\scriptstyle \textcolor{mygreen}{\mathtt{1}}$}, tikzit fill={rgb,255: red,35; green,170; blue,85}]
\tikzstyle{parg}=[draw=mygreen, anchor=north, trapezium, trapezium angle=110, rounded corners=2pt, inner sep=1.2pt, fill=white, font={$\scriptstyle \textcolor{mygreen}{\parr}$}, tikzit fill={rgb,255: red,35; green,170; blue,85}]
\tikzstyle{tensg}=[draw=mygreen, anchor=north, trapezium, trapezium angle=110, rounded corners=2pt, inner sep=1.2pt, fill=white, font={$\scriptstyle \textcolor{mygreen}{\varotimes}$}, tikzit fill={rgb,255: red,35; green,170; blue,85}]
\tikzstyle{whyg}=[draw=mygreen, anchor=north, trapezium, trapezium angle=110, rounded corners=2pt, inner sep=1.2pt, fill=white, font={$\scriptstyle \textcolor{mygreen}{?}$}, tikzit fill={rgb,255: red,35; green,170; blue,85}]
\tikzstyle{wtxtg}=[text=white, text=mygreen, tikzit fill={rgb,255: red,35; green,170; blue,85}]
\tikzstyle{lowLg}=[yshift=-0.8425ex, text=mygreen, draw=mygreen, tikzit fill={rgb,255: red,35; green,170; blue,85}]
\tikzstyle{lowRg}=[yshift=-0.885ex, text=mygreen, draw=mygreen, tikzit fill={rgb,255: red,35; green,170; blue,85}]
\tikzstyle{rectg}=[draw=mygreen, text=mygreen, anchor=north, inner sep=1.5pt, shape=rectangle, tikzit fill={rgb,255: red,35; green,170; blue,85}]
\tikzstyle{bothg}=[draw=black, anchor=north, trapezium, trapezium angle=110, rounded corners=2pt, inner sep=1.2pt, fill=mygreen, font={$\scriptstyle \bot$}, tikzit fill={rgb,255: red,35; green,170; blue,85}]
\tikzstyle{parhg}=[draw=black, anchor=north, trapezium, trapezium angle=110, rounded corners=2pt, inner sep=1.2pt, fill=mygreen, font={$\scriptstyle \parr$}, tikzit fill={rgb,255: red,35; green,170; blue,85}]
\tikzstyle{formp}=[inner sep=1.5pt, text=mypurple, tikzit fill={rgb,255: red,145; green,75; blue,210}]
\tikzstyle{dotp}=[font={$\textcolor{mypurple}{\bullet}$}, tikzit fill={rgb,255: red,145; green,75; blue,210}]
\tikzstyle{axp}=[draw=mypurple, anchor=north, trapezium, trapezium angle=110, rounded corners=2pt, inner sep=1.2pt, fill=white, font={$\scriptstyle \textcolor{mypurple}{\mathtt{ax}}$}, tikzit fill={rgb,255: red,145; green,75; blue,210}, tikzit shape=rectangle]
\tikzstyle{botp}=[draw=mypurple, anchor=north, trapezium, trapezium angle=110, rounded corners=2pt, inner sep=1.2pt, fill=white, font={$\scriptstyle \textcolor{mypurple}{\bot}$}, tikzit fill={rgb,255: red,145; green,75; blue,210}]
\tikzstyle{bangp}=[draw=mypurple, anchor=north, trapezium, trapezium angle=110, rounded corners=2pt, inner sep=1.2pt, fill=white, font={$\scriptstyle \textcolor{mypurple}{!}$}, tikzit fill={rgb,255: red,145; green,75; blue,210}]
\tikzstyle{cutp}=[draw=mypurple, anchor=north, trapezium, trapezium angle=110, rounded corners=2pt, inner sep=1.2pt, fill=white, font={$\scriptstyle \textcolor{mypurple}{\mathtt{cut}}$}, tikzit fill={rgb,255: red,145; green,75; blue,210}, tikzit shape=rectangle]
\tikzstyle{derp}=[draw=mypurple, anchor=north, trapezium, trapezium angle=110, rounded corners=2pt, inner sep=1.2pt, fill=white, font={$\scriptstyle \textcolor{mypurple}{\mathtt{?d}}$}, tikzit fill={rgb,255: red,145; green,75; blue,210}]
\tikzstyle{labp}=[draw=mypurple, anchor=north, trapezium, trapezium angle=110, rounded corners=2pt, inner sep=1.2pt, fill=white, font={$\scriptstyle \textcolor{mypurple}{l}$}, tikzit fill={rgb,255: red,145; green,75; blue,210}]
\tikzstyle{lab2p}=[draw=mypurple, anchor=north, trapezium, trapezium angle=110, rounded corners=2pt, inner sep=1.2pt, fill=white, font={$\scriptstyle \textcolor{mypurple}{\hskip 0.125em\smash{l'}\vphantom{l} \hskip -0.125em}$}, tikzit fill={rgb,255: red,145; green,75; blue,210}]
\tikzstyle{onep}=[draw=mypurple, anchor=north, trapezium, trapezium angle=110, rounded corners=2pt, inner sep=1.2pt, fill=white, font={$\scriptstyle \textcolor{mypurple}{\mathtt{1}}$}, tikzit fill={rgb,255: red,145; green,75; blue,210}]
\tikzstyle{parp}=[draw=mypurple, anchor=north, trapezium, trapezium angle=110, rounded corners=2pt, inner sep=1.2pt, fill=white, font={$\scriptstyle \textcolor{mypurple}{\parr}$}, tikzit fill={rgb,255: red,145; green,75; blue,210}]
\tikzstyle{tensp}=[draw=mypurple, anchor=north, trapezium, trapezium angle=110, rounded corners=2pt, inner sep=1.2pt, fill=white, font={$\scriptstyle \textcolor{mypurple}{\varotimes}$}, tikzit fill={rgb,255: red,145; green,75; blue,210}]
\tikzstyle{whyp}=[draw=mypurple, anchor=north, trapezium, trapezium angle=110, rounded corners=2pt, inner sep=1.2pt, fill=white, font={$\scriptstyle \textcolor{mypurple}{?}$}, tikzit fill={rgb,255: red,145; green,75; blue,210}]
\tikzstyle{wtxtp}=[text=white, text=mypurple, tikzit fill={rgb,255: red,145; green,75; blue,210}]
\tikzstyle{lowLp}=[yshift=-0.8425ex, text=mypurple, draw=mypurple, tikzit fill={rgb,255: red,145; green,75; blue,210}]
\tikzstyle{lowRp}=[yshift=-0.885ex, text=mypurple, draw=mypurple, tikzit fill={rgb,255: red,145; green,75; blue,210}]
\tikzstyle{rectp}=[draw=mypurple, text=mypurple, anchor=north, inner sep=1.5pt, shape=rectangle, tikzit fill={rgb,255: red,145; green,75; blue,210}]
\tikzstyle{bothp}=[draw=black, anchor=north, trapezium, trapezium angle=110, rounded corners=2pt, inner sep=1.2pt, fill=mypurple, font={$\scriptstyle \bot$}, tikzit fill={rgb,255: red,145; green,75; blue,210}]
\tikzstyle{parhp}=[draw=black, anchor=north, trapezium, trapezium angle=110, rounded corners=2pt, inner sep=1.2pt, fill=mypurple, font={$\scriptstyle \parr$}, tikzit fill={rgb,255: red,145; green,75; blue,210}]
\tikzstyle{formy}=[inner sep=1.5pt, text=myyellow, tikzit fill={rgb,255: red,235; green,190; blue,35}]
\tikzstyle{doty}=[font={$\textcolor{myyellow}{\bullet}$}, tikzit fill={rgb,255: red,235; green,190; blue,35}]
\tikzstyle{axy}=[draw=myyellow, anchor=north, trapezium, trapezium angle=110, rounded corners=2pt, inner sep=1.2pt, fill=white, font={$\scriptstyle \textcolor{myyellow}{\mathtt{ax}}$}, tikzit fill={rgb,255: red,235; green,190; blue,35}, tikzit shape=rectangle]
\tikzstyle{boty}=[draw=myyellow, anchor=north, trapezium, trapezium angle=110, rounded corners=2pt, inner sep=1.2pt, fill=white, font={$\scriptstyle \textcolor{myyellow}{\bot}$}, tikzit fill={rgb,255: red,235; green,190; blue,35}]
\tikzstyle{bangy}=[draw=myyellow, anchor=north, trapezium, trapezium angle=110, rounded corners=2pt, inner sep=1.2pt, fill=white, font={$\scriptstyle \textcolor{myyellow}{!}$}, tikzit fill={rgb,255: red,235; green,190; blue,35}]
\tikzstyle{cuty}=[draw=myyellow, anchor=north, trapezium, trapezium angle=110, rounded corners=2pt, inner sep=1.2pt, fill=white, font={$\scriptstyle \textcolor{myyellow}{\mathtt{cut}}$}, tikzit fill={rgb,255: red,235; green,190; blue,35}, tikzit shape=rectangle]
\tikzstyle{dery}=[draw=myyellow, anchor=north, trapezium, trapezium angle=110, rounded corners=2pt, inner sep=1.2pt, fill=white, font={$\scriptstyle \textcolor{myyellow}{\mathtt{?d}}$}, tikzit fill={rgb,255: red,235; green,190; blue,35}]
\tikzstyle{laby}=[draw=myyellow, anchor=north, trapezium, trapezium angle=110, rounded corners=2pt, inner sep=1.2pt, fill=white, font={$\scriptstyle \textcolor{myyellow}{l}$}, tikzit fill={rgb,255: red,235; green,190; blue,35}]
\tikzstyle{lab2y}=[draw=myyellow, anchor=north, trapezium, trapezium angle=110, rounded corners=2pt, inner sep=1.2pt, fill=white, font={$\scriptstyle \textcolor{myyellow}{\hskip 0.125em\smash{l'}\vphantom{l} \hskip -0.125em}$}, tikzit fill={rgb,255: red,235; green,190; blue,35}]
\tikzstyle{oney}=[draw=myyellow, anchor=north, trapezium, trapezium angle=110, rounded corners=2pt, inner sep=1.2pt, fill=white, font={$\scriptstyle \textcolor{myyellow}{\mathtt{1}}$}, tikzit fill={rgb,255: red,235; green,190; blue,35}]
\tikzstyle{pary}=[draw=myyellow, anchor=north, trapezium, trapezium angle=110, rounded corners=2pt, inner sep=1.2pt, fill=white, font={$\scriptstyle \textcolor{myyellow}{\parr}$}, tikzit fill={rgb,255: red,235; green,190; blue,35}]
\tikzstyle{tensy}=[draw=myyellow, anchor=north, trapezium, trapezium angle=110, rounded corners=2pt, inner sep=1.2pt, fill=white, font={$\scriptstyle \textcolor{myyellow}{\varotimes}$}, tikzit fill={rgb,255: red,235; green,190; blue,35}]
\tikzstyle{whyy}=[draw=myyellow, anchor=north, trapezium, trapezium angle=110, rounded corners=2pt, inner sep=1.2pt, fill=white, font={$\scriptstyle \textcolor{myyellow}{?}$}, tikzit fill={rgb,255: red,235; green,190; blue,35}]
\tikzstyle{wtxty}=[text=white, text=myyellow, tikzit fill={rgb,255: red,235; green,190; blue,35}]
\tikzstyle{lowLy}=[yshift=-0.8425ex, text=myyellow, draw=myyellow, tikzit fill={rgb,255: red,235; green,190; blue,35}]
\tikzstyle{lowRy}=[yshift=-0.885ex, text=myyellow, draw=myyellow, tikzit fill={rgb,255: red,235; green,190; blue,35}]
\tikzstyle{recty}=[draw=myyellow, text=myyellow, anchor=north, inner sep=1.5pt, shape=rectangle, tikzit fill={rgb,255: red,235; green,190; blue,35}]
\tikzstyle{bothy}=[draw=black, anchor=north, trapezium, trapezium angle=110, rounded corners=2pt, inner sep=1.2pt, fill=myyellow, font={$\scriptstyle \bot$}, tikzit fill={rgb,255: red,235; green,190; blue,35}]
\tikzstyle{parhy}=[draw=black, anchor=north, trapezium, trapezium angle=110, rounded corners=2pt, inner sep=1.2pt, fill=myyellow, font={$\scriptstyle \parr$}, tikzit fill={rgb,255: red,235; green,190; blue,35}]
\tikzstyle{formo}=[inner sep=1.5pt, text=myorange, tikzit fill={rgb,255: red,245; green,130; blue,30}]
\tikzstyle{doto}=[font={$\textcolor{myorange}{\bullet}$}, tikzit fill={rgb,255: red,245; green,130; blue,30}]
\tikzstyle{axo}=[draw=myorange, anchor=north, trapezium, trapezium angle=110, rounded corners=2pt, inner sep=1.2pt, fill=white, font={$\scriptstyle \textcolor{myorange}{\mathtt{ax}}$}, tikzit fill={rgb,255: red,245; green,130; blue,30}, tikzit shape=rectangle]
\tikzstyle{boto}=[draw=myorange, anchor=north, trapezium, trapezium angle=110, rounded corners=2pt, inner sep=1.2pt, fill=white, font={$\scriptstyle \textcolor{myorange}{\bot}$}, tikzit fill={rgb,255: red,245; green,130; blue,30}]
\tikzstyle{bango}=[draw=myorange, anchor=north, trapezium, trapezium angle=110, rounded corners=2pt, inner sep=1.2pt, fill=white, font={$\scriptstyle \textcolor{myorange}{!}$}, tikzit fill={rgb,255: red,245; green,130; blue,30}]
\tikzstyle{cuto}=[draw=myorange, anchor=north, trapezium, trapezium angle=110, rounded corners=2pt, inner sep=1.2pt, fill=white, font={$\scriptstyle \textcolor{myorange}{\mathtt{cut}}$}, tikzit fill={rgb,255: red,245; green,130; blue,30}, tikzit shape=rectangle]
\tikzstyle{dero}=[draw=myorange, anchor=north, trapezium, trapezium angle=110, rounded corners=2pt, inner sep=1.2pt, fill=white, font={$\scriptstyle \textcolor{myorange}{\mathtt{?d}}$}, tikzit fill={rgb,255: red,245; green,130; blue,30}]
\tikzstyle{labo}=[draw=myorange, anchor=north, trapezium, trapezium angle=110, rounded corners=2pt, inner sep=1.2pt, fill=white, font={$\scriptstyle \textcolor{myorange}{l}$}, tikzit fill={rgb,255: red,245; green,130; blue,30}]
\tikzstyle{lab2o}=[draw=myorange, anchor=north, trapezium, trapezium angle=110, rounded corners=2pt, inner sep=1.2pt, fill=white, font={$\scriptstyle \textcolor{myorange}{\hskip 0.125em\smash{l'}\vphantom{l} \hskip -0.125em}$}, tikzit fill={rgb,255: red,245; green,130; blue,30}]
\tikzstyle{oneo}=[draw=myorange, anchor=north, trapezium, trapezium angle=110, rounded corners=2pt, inner sep=1.2pt, fill=white, font={$\scriptstyle \textcolor{myorange}{\mathtt{1}}$}, tikzit fill={rgb,255: red,245; green,130; blue,30}]
\tikzstyle{paro}=[draw=myorange, anchor=north, trapezium, trapezium angle=110, rounded corners=2pt, inner sep=1.2pt, fill=white, font={$\scriptstyle \textcolor{myorange}{\parr}$}, tikzit fill={rgb,255: red,245; green,130; blue,30}]
\tikzstyle{tenso}=[draw=myorange, anchor=north, trapezium, trapezium angle=110, rounded corners=2pt, inner sep=1.2pt, fill=white, font={$\scriptstyle \textcolor{myorange}{\varotimes}$}, tikzit fill={rgb,255: red,245; green,130; blue,30}]
\tikzstyle{whyo}=[draw=myorange, anchor=north, trapezium, trapezium angle=110, rounded corners=2pt, inner sep=1.2pt, fill=white, font={$\scriptstyle \textcolor{myorange}{?}$}, tikzit fill={rgb,255: red,245; green,130; blue,30}]
\tikzstyle{wtxto}=[text=white, text=myorange, tikzit fill={rgb,255: red,245; green,130; blue,30}]
\tikzstyle{lowLo}=[yshift=-0.8425ex, text=myorange, draw=myorange, tikzit fill={rgb,255: red,245; green,130; blue,30}]
\tikzstyle{lowRo}=[yshift=-0.885ex, text=myorange, draw=myorange, tikzit fill={rgb,255: red,245; green,130; blue,30}]
\tikzstyle{recto}=[draw=myorange, text=myorange, anchor=north, inner sep=1.5pt, shape=rectangle, tikzit fill={rgb,255: red,245; green,130; blue,30}]
\tikzstyle{botho}=[draw=black, anchor=north, trapezium, trapezium angle=110, rounded corners=2pt, inner sep=1.2pt, fill=myorange, font={$\scriptstyle \bot$}, tikzit fill={rgb,255: red,245; green,130; blue,30}]
\tikzstyle{parho}=[draw=black, anchor=north, trapezium, trapezium angle=110, rounded corners=2pt, inner sep=1.2pt, fill=myorange, font={$\scriptstyle \parr$}, tikzit fill={rgb,255: red,245; green,130; blue,30}]
\tikzstyle{wire}=[draw=black, -, draw opacity=0, -latex, line width=0.2pt, postaction={draw, opacity=1, -, line width=1pt, shorten >=3pt, shorten <=-1pt}, tikzit draw={rgb,255: red,100; green,75; blue,200}]
\tikzstyle{wirB}=[draw=black, -, draw opacity=0, -latex, line width=0.2pt, postaction={draw, opacity=1, -, line width=1pt, shorten >=-1pt, shorten <=-1pt}, shorten >=-4pt, tikzit draw={rgb,255: red,130; green,75; blue,200}]
\tikzstyle{wEnd}=[draw=black, -, draw opacity=0, -latex, line width=0.2pt, postaction={draw, opacity=1, -, line width=1pt, shorten >=3pt}, tikzit draw={rgb,255: red,160; green,75; blue,200}]
\tikzstyle{dwir}=[draw=black, ->, draw opacity=0, latex reversed-, line width=0.2pt, postaction={draw, opacity=1, -, line width=1pt, shorten >=-5pt}, tikzit draw={rgb,255: red,100; green,75; blue,200}]
\tikzstyle{dwiB}=[draw=black, ->, draw opacity=0, latex reversed-, line width=0.2pt, postaction={draw, opacity=1, -, line width=1pt}, tikzit draw={rgb,255: red,130; green,75; blue,200}]
\tikzstyle{coL}=[out=180, in=90, looseness=0.85, -, draw opacity=0, -latex, line width=0.2pt, postaction={draw, opacity=1, -, line width=1pt, shorten >=3pt, shorten <=-1pt}, tikzit draw={rgb,255: red,100; green,75; blue,200}]
\tikzstyle{coR}=[out=00, in=90, looseness=0.85, -, draw opacity=0, -latex, line width=0.2pt, postaction={draw, opacity=1, -, line width=1pt, shorten >=3pt, shorten <=-1pt}, tikzit draw={rgb,255: red,245; green,200; blue,100}]
\tikzstyle{ciL}=[out=-90, in=180, looseness=0.85, -, draw opacity=0, -latex, line width=0.2pt, postaction={draw, opacity=1, -, line width=1pt, shorten >=3pt}, tikzit draw={rgb,255: red,100; green,75; blue,200}]
\tikzstyle{ciR}=[out=-90, in=00, looseness=0.85, -, draw opacity=0, -latex, line width=0.2pt, postaction={draw, opacity=1, -, line width=1pt, shorten >=3pt}, tikzit draw={rgb,255: red,245; green,200; blue,100}]
\tikzstyle{dciL}=[out=-90, in=180, looseness=0.85, ->, draw opacity=0, latex reversed-, line width=0.2pt, postaction={draw, opacity=1, -, line width=1pt, shorten <=1pt, shorten >=-1pt}, tikzit draw={rgb,255: red,100; green,75; blue,200}]
\tikzstyle{dciR}=[out=-90, in=00, looseness=0.85, ->, draw opacity=0, latex reversed-, line width=0.2pt, postaction={draw, opacity=1, -, line width=1pt, shorten <=1pt, shorten >=-1pt}, tikzit draw={rgb,255: red,245; green,200; blue,100}]
\tikzstyle{chR}=[out=0, in=180, looseness=0.85, ->, draw opacity=0, -latex, line width=0.2pt, postaction={draw, opacity=1, -, line width=1pt, shorten <=-1pt, shorten >=3pt}, tikzit draw={rgb,255: red,245; green,200; blue,100}]
\tikzstyle{fil0}=[-, fill={rgb,255: red,238; green,238; blue,238}, draw=black, fill opacity=0, tikzit draw={rgb,255: red,120; green,120; blue,120}]
\tikzstyle{bfrm}=[draw=black, line width=1pt, -, fill={rgb,255: red,0; green,0; blue,0}, rounded corners, tikzit draw={rgb,255: red,0; green,0; blue,0}]
\tikzstyle{chL}=[out=180, in=0, looseness=0.85, ->, draw opacity=0, -latex, line width=0.2pt, postaction={draw, opacity=1, -, line width=1pt, shorten <=-1pt, shorten >=3pt}, tikzit draw={rgb,255: red,100; green,75; blue,200}]
\tikzstyle{clL}=[out=180, in=180, looseness=0.85, ->, draw opacity=0, -latex, line width=0.2pt, postaction={draw, opacity=1, -, line width=1pt, shorten <=-1pt, shorten >=3pt}, tikzit draw={rgb,255: red,140; green,75; blue,220}]
\tikzstyle{clR}=[out=0, in=0, looseness=0.85, ->, draw opacity=0, -latex, line width=0.2pt, postaction={draw, opacity=1, -, line width=1pt, shorten <=-1pt, shorten >=3pt}, tikzit draw={rgb,255: red,255; green,170; blue,80}]
\tikzstyle{eoL}=[out=180, in=90, looseness=0.85, -, draw opacity=0, -latex, line width=0.2pt, postaction={draw, opacity=1, -, line width=1pt, shorten >=-1pt, shorten <=-1pt}, shorten >=-4pt, tikzit draw={rgb,255: red,120; green,110; blue,200}]
\tikzstyle{eoR}=[out=00, in=90, looseness=0.85, -, draw opacity=0, -latex, line width=0.2pt, postaction={draw, opacity=1, -, line width=1pt, shorten >=-1pt, shorten <=-1pt}, shorten >=-4pt, tikzit draw={rgb,255: red,220; green,190; blue,120}]
\tikzstyle{qoL}=[out=180, in=90, looseness=0.85, ->, draw opacity=0, -latex, line width=0.2pt, postaction={draw, opacity=1, -, line width=1pt, shorten >=3pt}, tikzit draw={rgb,255: red,80; green,140; blue,200}]
\tikzstyle{qoR}=[out=00, in=90, looseness=0.85, ->, draw opacity=0, -latex, line width=0.2pt, postaction={draw, opacity=1, -, line width=1pt, shorten >=3pt}, tikzit draw={rgb,255: red,220; green,180; blue,90}]
\tikzstyle{dqiL}=[out=-90, in=180, looseness=0.85, ->, draw opacity=0, latex reversed-, line width=0.2pt, postaction={draw, opacity=1, -, line width=1pt, shorten <=1pt, shorten >=0pt}, tikzit draw={rgb,255: red,90; green,160; blue,200}]
\tikzstyle{dqiR}=[out=-90, in=00, looseness=0.85, ->, draw opacity=0, latex reversed-, line width=0.2pt, postaction={draw, opacity=1, -, line width=1pt, shorten <=1pt, shorten >=0pt}, tikzit draw={rgb,255: red,220; green,170; blue,110}]
\tikzstyle{rarr}=[->, tikzit draw={rgb,255: red,90; green,90; blue,90}]
\tikzstyle{hc1L}=[out=-90, in=180, looseness=0.85, -, draw opacity=0, -latex, line width=0.3pt, postaction={draw, opacity=1, -, line width=1.5pt, shorten >=3pt}, tikzit draw={rgb,255: red,255; green,0; blue,0}]
\tikzstyle{hc1R}=[out=-90, in=00, looseness=0.85, -, draw opacity=0, -latex, line width=0.3pt, postaction={draw, opacity=1, -, line width=1.5pt, shorten >=3pt}, tikzit draw={rgb,255: red,245; green,20; blue,20}]
\tikzstyle{hu1R}=[out=0, in=270, looseness=0.85, -, draw opacity=0, -latex, line width=0.3pt, postaction={draw, opacity=1, -, line width=1.5pt, shorten >=3pt, shorten <=-0.3pt}, tikzit draw={rgb,255: red,225; green,10; blue,20}]
\tikzstyle{wirer}=[draw=myred, -, draw opacity=0, -latex, line width=0.2pt, postaction={draw, opacity=1, -, line width=1pt, shorten >=3pt, shorten <=-1pt}, tikzit draw={rgb,255: red,230; green,45; blue,55}]
\tikzstyle{wirBr}=[draw=myred, -, draw opacity=0, -latex, line width=0.2pt, postaction={draw, opacity=1, -, line width=1pt, shorten >=-1pt, shorten <=-1pt}, shorten >=-4pt, tikzit draw={rgb,255: red,230; green,45; blue,55}]
\tikzstyle{wEndr}=[draw=myred, -, draw opacity=0, -latex, line width=0.2pt, postaction={draw, opacity=1, -, line width=1pt, shorten >=3pt}, tikzit draw={rgb,255: red,230; green,45; blue,55}]
\tikzstyle{dwirr}=[draw=myred, ->, draw opacity=0, latex reversed-, line width=0.2pt, postaction={draw, opacity=1, -, line width=1pt, shorten >=-5pt}, tikzit draw={rgb,255: red,230; green,45; blue,55}]
\tikzstyle{dwiBr}=[draw=myred, ->, draw opacity=0, latex reversed-, line width=0.2pt, postaction={draw, opacity=1, -, line width=1pt}, tikzit draw={rgb,255: red,230; green,45; blue,55}]
\tikzstyle{coLr}=[draw=myred, out=180, in=90, looseness=0.85, -, draw opacity=0, -latex, line width=0.2pt, postaction={draw, opacity=1, -, line width=1pt, shorten >=3pt, shorten <=-1pt}, tikzit draw={rgb,255: red,230; green,45; blue,55}]
\tikzstyle{coRr}=[draw=myred, out=00, in=90, looseness=0.85, -, draw opacity=0, -latex, line width=0.2pt, postaction={draw, opacity=1, -, line width=1pt, shorten >=3pt, shorten <=-1pt}, tikzit draw={rgb,255: red,230; green,45; blue,55}]
\tikzstyle{ciLr}=[draw=myred, out=-90, in=180, looseness=0.85, -, draw opacity=0, -latex, line width=0.2pt, postaction={draw, opacity=1, -, line width=1pt, shorten >=3pt}, tikzit draw={rgb,255: red,230; green,45; blue,55}]
\tikzstyle{ciRr}=[draw=myred, out=-90, in=00, looseness=0.85, -, draw opacity=0, -latex, line width=0.2pt, postaction={draw, opacity=1, -, line width=1pt, shorten >=3pt}, tikzit draw={rgb,255: red,230; green,45; blue,55}]
\tikzstyle{dciLr}=[draw=myred, out=-90, in=180, looseness=0.85, ->, draw opacity=0, latex reversed-, line width=0.2pt, postaction={draw, opacity=1, -, line width=1pt, shorten <=1pt, shorten >=-1pt}, tikzit draw={rgb,255: red,230; green,45; blue,55}]
\tikzstyle{dciRr}=[draw=myred, out=-90, in=00, looseness=0.85, ->, draw opacity=0, latex reversed-, line width=0.2pt, postaction={draw, opacity=1, -, line width=1pt, shorten <=1pt, shorten >=-1pt}, tikzit draw={rgb,255: red,230; green,45; blue,55}]
\tikzstyle{chRr}=[draw=myred, out=0, in=180, looseness=0.85, ->, draw opacity=0, -latex, line width=0.2pt, postaction={draw, opacity=1, -, line width=1pt, shorten <=-1pt, shorten >=3pt}, tikzit draw={rgb,255: red,230; green,45; blue,55}]
\tikzstyle{fil0r}=[-, fill=myred, draw=myred, fill opacity=0, tikzit draw={rgb,255: red,230; green,45; blue,55}]
\tikzstyle{bfrmr}=[draw=myred, line width=1pt, -, fill=myred, rounded corners, tikzit draw={rgb,255: red,230; green,45; blue,55}]
\tikzstyle{chLr}=[draw=myred, out=180, in=0, looseness=0.85, ->, draw opacity=0, -latex, line width=0.2pt, postaction={draw, opacity=1, -, line width=1pt, shorten <=-1pt, shorten >=3pt}, tikzit draw={rgb,255: red,230; green,45; blue,55}]
\tikzstyle{clLr}=[draw=myred, out=180, in=180, looseness=0.85, ->, draw opacity=0, -latex, line width=0.2pt, postaction={draw, opacity=1, -, line width=1pt, shorten <=-1pt, shorten >=3pt}, tikzit draw={rgb,255: red,230; green,45; blue,55}]
\tikzstyle{clRr}=[draw=myred, out=0, in=0, looseness=0.85, ->, draw opacity=0, -latex, line width=0.2pt, postaction={draw, opacity=1, -, line width=1pt, shorten <=-1pt, shorten >=3pt}, tikzit draw={rgb,255: red,230; green,45; blue,55}]
\tikzstyle{eoLr}=[draw=myred, out=180, in=90, looseness=0.85, -, draw opacity=0, -latex, line width=0.2pt, postaction={draw, opacity=1, -, line width=1pt, shorten >=-1pt, shorten <=-1pt}, shorten >=-4pt, tikzit draw={rgb,255: red,230; green,45; blue,55}]
\tikzstyle{eoRr}=[draw=myred, out=00, in=90, looseness=0.85, -, draw opacity=0, -latex, line width=0.2pt, postaction={draw, opacity=1, -, line width=1pt, shorten >=-1pt, shorten <=-1pt}, shorten >=-4pt, tikzit draw={rgb,255: red,230; green,45; blue,55}]
\tikzstyle{qoLr}=[draw=myred, out=180, in=90, looseness=0.85, ->, draw opacity=0, -latex, line width=0.2pt, postaction={draw, opacity=1, -, line width=1pt, shorten >=3pt}, tikzit draw={rgb,255: red,230; green,45; blue,55}]
\tikzstyle{qoRr}=[draw=myred, out=00, in=90, looseness=0.85, ->, draw opacity=0, -latex, line width=0.2pt, postaction={draw, opacity=1, -, line width=1pt, shorten >=3pt}, tikzit draw={rgb,255: red,230; green,45; blue,55}]
\tikzstyle{dqiLr}=[draw=myred, out=-90, in=180, looseness=0.85, ->, draw opacity=0, latex reversed-, line width=0.2pt, postaction={draw, opacity=1, -, line width=1pt, shorten <=1pt, shorten >=0pt}, tikzit draw={rgb,255: red,230; green,45; blue,55}]
\tikzstyle{dqiRr}=[draw=myred, out=-90, in=00, looseness=0.85, ->, draw opacity=0, latex reversed-, line width=0.2pt, postaction={draw, opacity=1, -, line width=1pt, shorten <=1pt, shorten >=0pt}, tikzit draw={rgb,255: red,230; green,45; blue,55}]
\tikzstyle{rarrr}=[draw=myred, ->, tikzit draw={rgb,255: red,230; green,45; blue,55}]
\tikzstyle{hc1Lr}=[draw=myred, out=-90, in=180, looseness=0.85, -, draw opacity=0, -latex, line width=0.3pt, postaction={draw, opacity=1, -, line width=1.5pt, shorten >=3pt}, tikzit draw={rgb,255: red,230; green,45; blue,55}]
\tikzstyle{hc1Rr}=[draw=myred, out=-90, in=00, looseness=0.85, -, draw opacity=0, -latex, line width=0.3pt, postaction={draw, opacity=1, -, line width=1.5pt, shorten >=3pt}, tikzit draw={rgb,255: red,230; green,45; blue,55}]
\tikzstyle{hu1Rr}=[draw=myred, out=0, in=270, looseness=0.85, -, draw opacity=0, -latex, line width=0.3pt, postaction={draw, opacity=1, -, line width=1.5pt, shorten >=3pt, shorten <=-0.3pt}, tikzit draw={rgb,255: red,230; green,45; blue,55}]
\tikzstyle{wireb}=[draw=myblue, -, draw opacity=0, -latex, line width=0.2pt, postaction={draw, opacity=1, -, line width=1pt, shorten >=3pt, shorten <=-1pt}, tikzit draw={rgb,255: red,40; green,105; blue,230}]
\tikzstyle{wirBb}=[draw=myblue, -, draw opacity=0, -latex, line width=0.2pt, postaction={draw, opacity=1, -, line width=1pt, shorten >=-1pt, shorten <=-1pt}, shorten >=-4pt, tikzit draw={rgb,255: red,40; green,105; blue,230}]
\tikzstyle{wEndb}=[draw=myblue, -, draw opacity=0, -latex, line width=0.2pt, postaction={draw, opacity=1, -, line width=1pt, shorten >=3pt}, tikzit draw={rgb,255: red,40; green,105; blue,230}]
\tikzstyle{dwirb}=[draw=myblue, ->, draw opacity=0, latex reversed-, line width=0.2pt, postaction={draw, opacity=1, -, line width=1pt, shorten >=-5pt}, tikzit draw={rgb,255: red,40; green,105; blue,230}]
\tikzstyle{dwiBb}=[draw=myblue, ->, draw opacity=0, latex reversed-, line width=0.2pt, postaction={draw, opacity=1, -, line width=1pt}, tikzit draw={rgb,255: red,40; green,105; blue,230}]
\tikzstyle{coLb}=[draw=myblue, out=180, in=90, looseness=0.85, -, draw opacity=0, -latex, line width=0.2pt, postaction={draw, opacity=1, -, line width=1pt, shorten >=3pt, shorten <=-1pt}, tikzit draw={rgb,255: red,40; green,105; blue,230}]
\tikzstyle{coRb}=[draw=myblue, out=00, in=90, looseness=0.85, -, draw opacity=0, -latex, line width=0.2pt, postaction={draw, opacity=1, -, line width=1pt, shorten >=3pt, shorten <=-1pt}, tikzit draw={rgb,255: red,40; green,105; blue,230}]
\tikzstyle{ciLb}=[draw=myblue, out=-90, in=180, looseness=0.85, -, draw opacity=0, -latex, line width=0.2pt, postaction={draw, opacity=1, -, line width=1pt, shorten >=3pt}, tikzit draw={rgb,255: red,40; green,105; blue,230}]
\tikzstyle{ciRb}=[draw=myblue, out=-90, in=00, looseness=0.85, -, draw opacity=0, -latex, line width=0.2pt, postaction={draw, opacity=1, -, line width=1pt, shorten >=3pt}, tikzit draw={rgb,255: red,40; green,105; blue,230}]
\tikzstyle{dciLb}=[draw=myblue, out=-90, in=180, looseness=0.85, ->, draw opacity=0, latex reversed-, line width=0.2pt, postaction={draw, opacity=1, -, line width=1pt, shorten <=1pt, shorten >=-1pt}, tikzit draw={rgb,255: red,40; green,105; blue,230}]
\tikzstyle{dciRb}=[draw=myblue, out=-90, in=00, looseness=0.85, ->, draw opacity=0, latex reversed-, line width=0.2pt, postaction={draw, opacity=1, -, line width=1pt, shorten <=1pt, shorten >=-1pt}, tikzit draw={rgb,255: red,40; green,105; blue,230}]
\tikzstyle{chRb}=[draw=myblue, out=0, in=180, looseness=0.85, ->, draw opacity=0, -latex, line width=0.2pt, postaction={draw, opacity=1, -, line width=1pt, shorten <=-1pt, shorten >=3pt}, tikzit draw={rgb,255: red,40; green,105; blue,230}]
\tikzstyle{fil0b}=[-, fill=myblue, draw=myblue, fill opacity=0, tikzit draw={rgb,255: red,40; green,105; blue,230}]
\tikzstyle{bfrmb}=[draw=myblue, line width=1pt, -, fill=myblue, rounded corners, tikzit draw={rgb,255: red,40; green,105; blue,230}]
\tikzstyle{chLb}=[draw=myblue, out=180, in=0, looseness=0.85, ->, draw opacity=0, -latex, line width=0.2pt, postaction={draw, opacity=1, -, line width=1pt, shorten <=-1pt, shorten >=3pt}, tikzit draw={rgb,255: red,40; green,105; blue,230}]
\tikzstyle{clLb}=[draw=myblue, out=180, in=180, looseness=0.85, ->, draw opacity=0, -latex, line width=0.2pt, postaction={draw, opacity=1, -, line width=1pt, shorten <=-1pt, shorten >=3pt}, tikzit draw={rgb,255: red,40; green,105; blue,230}]
\tikzstyle{clRb}=[draw=myblue, out=0, in=0, looseness=0.85, ->, draw opacity=0, -latex, line width=0.2pt, postaction={draw, opacity=1, -, line width=1pt, shorten <=-1pt, shorten >=3pt}, tikzit draw={rgb,255: red,40; green,105; blue,230}]
\tikzstyle{eoLb}=[draw=myblue, out=180, in=90, looseness=0.85, -, draw opacity=0, -latex, line width=0.2pt, postaction={draw, opacity=1, -, line width=1pt, shorten >=-1pt, shorten <=-1pt}, shorten >=-4pt, tikzit draw={rgb,255: red,40; green,105; blue,230}]
\tikzstyle{eoRb}=[draw=myblue, out=00, in=90, looseness=0.85, -, draw opacity=0, -latex, line width=0.2pt, postaction={draw, opacity=1, -, line width=1pt, shorten >=-1pt, shorten <=-1pt}, shorten >=-4pt, tikzit draw={rgb,255: red,40; green,105; blue,230}]
\tikzstyle{qoLb}=[draw=myblue, out=180, in=90, looseness=0.85, ->, draw opacity=0, -latex, line width=0.2pt, postaction={draw, opacity=1, -, line width=1pt, shorten >=3pt}, tikzit draw={rgb,255: red,40; green,105; blue,230}]
\tikzstyle{qoRb}=[draw=myblue, out=00, in=90, looseness=0.85, ->, draw opacity=0, -latex, line width=0.2pt, postaction={draw, opacity=1, -, line width=1pt, shorten >=3pt}, tikzit draw={rgb,255: red,40; green,105; blue,230}]
\tikzstyle{dqiLb}=[draw=myblue, out=-90, in=180, looseness=0.85, ->, draw opacity=0, latex reversed-, line width=0.2pt, postaction={draw, opacity=1, -, line width=1pt, shorten <=1pt, shorten >=0pt}, tikzit draw={rgb,255: red,40; green,105; blue,230}]
\tikzstyle{dqiRb}=[draw=myblue, out=-90, in=00, looseness=0.85, ->, draw opacity=0, latex reversed-, line width=0.2pt, postaction={draw, opacity=1, -, line width=1pt, shorten <=1pt, shorten >=0pt}, tikzit draw={rgb,255: red,40; green,105; blue,230}]
\tikzstyle{rarrb}=[draw=myblue, ->, tikzit draw={rgb,255: red,40; green,105; blue,230}]
\tikzstyle{hc1Lb}=[draw=myblue, out=-90, in=180, looseness=0.85, -, draw opacity=0, -latex, line width=0.3pt, postaction={draw, opacity=1, -, line width=1.5pt, shorten >=3pt}, tikzit draw={rgb,255: red,40; green,105; blue,230}]
\tikzstyle{hc1Rb}=[draw=myblue, out=-90, in=00, looseness=0.85, -, draw opacity=0, -latex, line width=0.3pt, postaction={draw, opacity=1, -, line width=1.5pt, shorten >=3pt}, tikzit draw={rgb,255: red,40; green,105; blue,230}]
\tikzstyle{hu1Rb}=[draw=myblue, out=0, in=270, looseness=0.85, -, draw opacity=0, -latex, line width=0.3pt, postaction={draw, opacity=1, -, line width=1.5pt, shorten >=3pt, shorten <=-0.3pt}, tikzit draw={rgb,255: red,40; green,105; blue,230}]
\tikzstyle{wireg}=[draw=mygreen, -, draw opacity=0, -latex, line width=0.2pt, postaction={draw, opacity=1, -, line width=1pt, shorten >=3pt, shorten <=-1pt}, tikzit draw={rgb,255: red,35; green,170; blue,85}]
\tikzstyle{wirBg}=[draw=mygreen, -, draw opacity=0, -latex, line width=0.2pt, postaction={draw, opacity=1, -, line width=1pt, shorten >=-1pt, shorten <=-1pt}, shorten >=-4pt, tikzit draw={rgb,255: red,35; green,170; blue,85}]
\tikzstyle{wEndg}=[draw=mygreen, -, draw opacity=0, -latex, line width=0.2pt, postaction={draw, opacity=1, -, line width=1pt, shorten >=3pt}, tikzit draw={rgb,255: red,35; green,170; blue,85}]
\tikzstyle{dwirg}=[draw=mygreen, ->, draw opacity=0, latex reversed-, line width=0.2pt, postaction={draw, opacity=1, -, line width=1pt, shorten >=-5pt}, tikzit draw={rgb,255: red,35; green,170; blue,85}]
\tikzstyle{dwiBg}=[draw=mygreen, ->, draw opacity=0, latex reversed-, line width=0.2pt, postaction={draw, opacity=1, -, line width=1pt}, tikzit draw={rgb,255: red,35; green,170; blue,85}]
\tikzstyle{coLg}=[draw=mygreen, out=180, in=90, looseness=0.85, -, draw opacity=0, -latex, line width=0.2pt, postaction={draw, opacity=1, -, line width=1pt, shorten >=3pt, shorten <=-1pt}, tikzit draw={rgb,255: red,35; green,170; blue,85}]
\tikzstyle{coRg}=[draw=mygreen, out=00, in=90, looseness=0.85, -, draw opacity=0, -latex, line width=0.2pt, postaction={draw, opacity=1, -, line width=1pt, shorten >=3pt, shorten <=-1pt}, tikzit draw={rgb,255: red,35; green,170; blue,85}]
\tikzstyle{ciLg}=[draw=mygreen, out=-90, in=180, looseness=0.85, -, draw opacity=0, -latex, line width=0.2pt, postaction={draw, opacity=1, -, line width=1pt, shorten >=3pt}, tikzit draw={rgb,255: red,35; green,170; blue,85}]
\tikzstyle{ciRg}=[draw=mygreen, out=-90, in=00, looseness=0.85, -, draw opacity=0, -latex, line width=0.2pt, postaction={draw, opacity=1, -, line width=1pt, shorten >=3pt}, tikzit draw={rgb,255: red,35; green,170; blue,85}]
\tikzstyle{dciLg}=[draw=mygreen, out=-90, in=180, looseness=0.85, ->, draw opacity=0, latex reversed-, line width=0.2pt, postaction={draw, opacity=1, -, line width=1pt, shorten <=1pt, shorten >=-1pt}, tikzit draw={rgb,255: red,35; green,170; blue,85}]
\tikzstyle{dciRg}=[draw=mygreen, out=-90, in=00, looseness=0.85, ->, draw opacity=0, latex reversed-, line width=0.2pt, postaction={draw, opacity=1, -, line width=1pt, shorten <=1pt, shorten >=-1pt}, tikzit draw={rgb,255: red,35; green,170; blue,85}]
\tikzstyle{chRg}=[draw=mygreen, out=0, in=180, looseness=0.85, ->, draw opacity=0, -latex, line width=0.2pt, postaction={draw, opacity=1, -, line width=1pt, shorten <=-1pt, shorten >=3pt}, tikzit draw={rgb,255: red,35; green,170; blue,85}]
\tikzstyle{fil0g}=[-, fill=mygreen, draw=mygreen, fill opacity=0, tikzit draw={rgb,255: red,35; green,170; blue,85}]
\tikzstyle{bfrmg}=[draw=mygreen, line width=1pt, -, fill=mygreen, rounded corners, tikzit draw={rgb,255: red,35; green,170; blue,85}]
\tikzstyle{chLg}=[draw=mygreen, out=180, in=0, looseness=0.85, ->, draw opacity=0, -latex, line width=0.2pt, postaction={draw, opacity=1, -, line width=1pt, shorten <=-1pt, shorten >=3pt}, tikzit draw={rgb,255: red,35; green,170; blue,85}]
\tikzstyle{clLg}=[draw=mygreen, out=180, in=180, looseness=0.85, ->, draw opacity=0, -latex, line width=0.2pt, postaction={draw, opacity=1, -, line width=1pt, shorten <=-1pt, shorten >=3pt}, tikzit draw={rgb,255: red,35; green,170; blue,85}]
\tikzstyle{clRg}=[draw=mygreen, out=0, in=0, looseness=0.85, ->, draw opacity=0, -latex, line width=0.2pt, postaction={draw, opacity=1, -, line width=1pt, shorten <=-1pt, shorten >=3pt}, tikzit draw={rgb,255: red,35; green,170; blue,85}]
\tikzstyle{eoLg}=[draw=mygreen, out=180, in=90, looseness=0.85, -, draw opacity=0, -latex, line width=0.2pt, postaction={draw, opacity=1, -, line width=1pt, shorten >=-1pt, shorten <=-1pt}, shorten >=-4pt, tikzit draw={rgb,255: red,35; green,170; blue,85}]
\tikzstyle{eoRg}=[draw=mygreen, out=00, in=90, looseness=0.85, -, draw opacity=0, -latex, line width=0.2pt, postaction={draw, opacity=1, -, line width=1pt, shorten >=-1pt, shorten <=-1pt}, shorten >=-4pt, tikzit draw={rgb,255: red,35; green,170; blue,85}]
\tikzstyle{qoLg}=[draw=mygreen, out=180, in=90, looseness=0.85, ->, draw opacity=0, -latex, line width=0.2pt, postaction={draw, opacity=1, -, line width=1pt, shorten >=3pt}, tikzit draw={rgb,255: red,35; green,170; blue,85}]
\tikzstyle{qoRg}=[draw=mygreen, out=00, in=90, looseness=0.85, ->, draw opacity=0, -latex, line width=0.2pt, postaction={draw, opacity=1, -, line width=1pt, shorten >=3pt}, tikzit draw={rgb,255: red,35; green,170; blue,85}]
\tikzstyle{dqiLg}=[draw=mygreen, out=-90, in=180, looseness=0.85, ->, draw opacity=0, latex reversed-, line width=0.2pt, postaction={draw, opacity=1, -, line width=1pt, shorten <=1pt, shorten >=0pt}, tikzit draw={rgb,255: red,35; green,170; blue,85}]
\tikzstyle{dqiRg}=[draw=mygreen, out=-90, in=00, looseness=0.85, ->, draw opacity=0, latex reversed-, line width=0.2pt, postaction={draw, opacity=1, -, line width=1pt, shorten <=1pt, shorten >=0pt}, tikzit draw={rgb,255: red,35; green,170; blue,85}]
\tikzstyle{rarrg}=[draw=mygreen, ->, tikzit draw={rgb,255: red,35; green,170; blue,85}]
\tikzstyle{hc1Lg}=[draw=mygreen, out=-90, in=180, looseness=0.85, -, draw opacity=0, -latex, line width=0.3pt, postaction={draw, opacity=1, -, line width=1.5pt, shorten >=3pt}, tikzit draw={rgb,255: red,35; green,170; blue,85}]
\tikzstyle{hc1Rg}=[draw=mygreen, out=-90, in=00, looseness=0.85, -, draw opacity=0, -latex, line width=0.3pt, postaction={draw, opacity=1, -, line width=1.5pt, shorten >=3pt}, tikzit draw={rgb,255: red,35; green,170; blue,85}]
\tikzstyle{hu1Rg}=[draw=mygreen, out=0, in=270, looseness=0.85, -, draw opacity=0, -latex, line width=0.3pt, postaction={draw, opacity=1, -, line width=1.5pt, shorten >=3pt, shorten <=-0.3pt}, tikzit draw={rgb,255: red,35; green,170; blue,85}]
\tikzstyle{wirep}=[draw=mypurple, -, draw opacity=0, -latex, line width=0.2pt, postaction={draw, opacity=1, -, line width=1pt, shorten >=3pt, shorten <=-1pt}, tikzit draw={rgb,255: red,145; green,75; blue,210}]
\tikzstyle{wirBp}=[draw=mypurple, -, draw opacity=0, -latex, line width=0.2pt, postaction={draw, opacity=1, -, line width=1pt, shorten >=-1pt, shorten <=-1pt}, shorten >=-4pt, tikzit draw={rgb,255: red,145; green,75; blue,210}]
\tikzstyle{wEndp}=[draw=mypurple, -, draw opacity=0, -latex, line width=0.2pt, postaction={draw, opacity=1, -, line width=1pt, shorten >=3pt}, tikzit draw={rgb,255: red,145; green,75; blue,210}]
\tikzstyle{dwirp}=[draw=mypurple, ->, draw opacity=0, latex reversed-, line width=0.2pt, postaction={draw, opacity=1, -, line width=1pt, shorten >=-5pt}, tikzit draw={rgb,255: red,145; green,75; blue,210}]
\tikzstyle{dwiBp}=[draw=mypurple, ->, draw opacity=0, latex reversed-, line width=0.2pt, postaction={draw, opacity=1, -, line width=1pt}, tikzit draw={rgb,255: red,145; green,75; blue,210}]
\tikzstyle{coLp}=[draw=mypurple, out=180, in=90, looseness=0.85, -, draw opacity=0, -latex, line width=0.2pt, postaction={draw, opacity=1, -, line width=1pt, shorten >=3pt, shorten <=-1pt}, tikzit draw={rgb,255: red,145; green,75; blue,210}]
\tikzstyle{coRp}=[draw=mypurple, out=00, in=90, looseness=0.85, -, draw opacity=0, -latex, line width=0.2pt, postaction={draw, opacity=1, -, line width=1pt, shorten >=3pt, shorten <=-1pt}, tikzit draw={rgb,255: red,145; green,75; blue,210}]
\tikzstyle{ciLp}=[draw=mypurple, out=-90, in=180, looseness=0.85, -, draw opacity=0, -latex, line width=0.2pt, postaction={draw, opacity=1, -, line width=1pt, shorten >=3pt}, tikzit draw={rgb,255: red,145; green,75; blue,210}]
\tikzstyle{ciRp}=[draw=mypurple, out=-90, in=00, looseness=0.85, -, draw opacity=0, -latex, line width=0.2pt, postaction={draw, opacity=1, -, line width=1pt, shorten >=3pt}, tikzit draw={rgb,255: red,145; green,75; blue,210}]
\tikzstyle{dciLp}=[draw=mypurple, out=-90, in=180, looseness=0.85, ->, draw opacity=0, latex reversed-, line width=0.2pt, postaction={draw, opacity=1, -, line width=1pt, shorten <=1pt, shorten >=-1pt}, tikzit draw={rgb,255: red,145; green,75; blue,210}]
\tikzstyle{dciRp}=[draw=mypurple, out=-90, in=00, looseness=0.85, ->, draw opacity=0, latex reversed-, line width=0.2pt, postaction={draw, opacity=1, -, line width=1pt, shorten <=1pt, shorten >=-1pt}, tikzit draw={rgb,255: red,145; green,75; blue,210}]
\tikzstyle{chRp}=[draw=mypurple, out=0, in=180, looseness=0.85, ->, draw opacity=0, -latex, line width=0.2pt, postaction={draw, opacity=1, -, line width=1pt, shorten <=-1pt, shorten >=3pt}, tikzit draw={rgb,255: red,145; green,75; blue,210}]
\tikzstyle{fil0p}=[-, fill=mypurple, draw=mypurple, fill opacity=0, tikzit draw={rgb,255: red,145; green,75; blue,210}]
\tikzstyle{bfrmp}=[draw=mypurple, line width=1pt, -, fill=mypurple, rounded corners, tikzit draw={rgb,255: red,145; green,75; blue,210}]
\tikzstyle{chLp}=[draw=mypurple, out=180, in=0, looseness=0.85, ->, draw opacity=0, -latex, line width=0.2pt, postaction={draw, opacity=1, -, line width=1pt, shorten <=-1pt, shorten >=3pt}, tikzit draw={rgb,255: red,145; green,75; blue,210}]
\tikzstyle{clLp}=[draw=mypurple, out=180, in=180, looseness=0.85, ->, draw opacity=0, -latex, line width=0.2pt, postaction={draw, opacity=1, -, line width=1pt, shorten <=-1pt, shorten >=3pt}, tikzit draw={rgb,255: red,145; green,75; blue,210}]
\tikzstyle{clRp}=[draw=mypurple, out=0, in=0, looseness=0.85, ->, draw opacity=0, -latex, line width=0.2pt, postaction={draw, opacity=1, -, line width=1pt, shorten <=-1pt, shorten >=3pt}, tikzit draw={rgb,255: red,145; green,75; blue,210}]
\tikzstyle{eoLp}=[draw=mypurple, out=180, in=90, looseness=0.85, -, draw opacity=0, -latex, line width=0.2pt, postaction={draw, opacity=1, -, line width=1pt, shorten >=-1pt, shorten <=-1pt}, shorten >=-4pt, tikzit draw={rgb,255: red,145; green,75; blue,210}]
\tikzstyle{eoRp}=[draw=mypurple, out=00, in=90, looseness=0.85, -, draw opacity=0, -latex, line width=0.2pt, postaction={draw, opacity=1, -, line width=1pt, shorten >=-1pt, shorten <=-1pt}, shorten >=-4pt, tikzit draw={rgb,255: red,145; green,75; blue,210}]
\tikzstyle{qoLp}=[draw=mypurple, out=180, in=90, looseness=0.85, ->, draw opacity=0, -latex, line width=0.2pt, postaction={draw, opacity=1, -, line width=1pt, shorten >=3pt}, tikzit draw={rgb,255: red,145; green,75; blue,210}]
\tikzstyle{qoRp}=[draw=mypurple, out=00, in=90, looseness=0.85, ->, draw opacity=0, -latex, line width=0.2pt, postaction={draw, opacity=1, -, line width=1pt, shorten >=3pt}, tikzit draw={rgb,255: red,145; green,75; blue,210}]
\tikzstyle{dqiLp}=[draw=mypurple, out=-90, in=180, looseness=0.85, ->, draw opacity=0, latex reversed-, line width=0.2pt, postaction={draw, opacity=1, -, line width=1pt, shorten <=1pt, shorten >=0pt}, tikzit draw={rgb,255: red,145; green,75; blue,210}]
\tikzstyle{dqiRp}=[draw=mypurple, out=-90, in=00, looseness=0.85, ->, draw opacity=0, latex reversed-, line width=0.2pt, postaction={draw, opacity=1, -, line width=1pt, shorten <=1pt, shorten >=0pt}, tikzit draw={rgb,255: red,145; green,75; blue,210}]
\tikzstyle{rarrp}=[draw=mypurple, ->, tikzit draw={rgb,255: red,145; green,75; blue,210}]
\tikzstyle{hc1Lp}=[draw=mypurple, out=-90, in=180, looseness=0.85, -, draw opacity=0, -latex, line width=0.3pt, postaction={draw, opacity=1, -, line width=1.5pt, shorten >=3pt}, tikzit draw={rgb,255: red,145; green,75; blue,210}]
\tikzstyle{hc1Rp}=[draw=mypurple, out=-90, in=00, looseness=0.85, -, draw opacity=0, -latex, line width=0.3pt, postaction={draw, opacity=1, -, line width=1.5pt, shorten >=3pt}, tikzit draw={rgb,255: red,145; green,75; blue,210}]
\tikzstyle{hu1Rp}=[draw=mypurple, out=0, in=270, looseness=0.85, -, draw opacity=0, -latex, line width=0.3pt, postaction={draw, opacity=1, -, line width=1.5pt, shorten >=3pt, shorten <=-0.3pt}, tikzit draw={rgb,255: red,145; green,75; blue,210}]
\tikzstyle{wirey}=[draw=myyellow, -, draw opacity=0, -latex, line width=0.2pt, postaction={draw, opacity=1, -, line width=1pt, shorten >=3pt, shorten <=-1pt}, tikzit draw={rgb,255: red,235; green,190; blue,35}]
\tikzstyle{wirBy}=[draw=myyellow, -, draw opacity=0, -latex, line width=0.2pt, postaction={draw, opacity=1, -, line width=1pt, shorten >=-1pt, shorten <=-1pt}, shorten >=-4pt, tikzit draw={rgb,255: red,235; green,190; blue,35}]
\tikzstyle{wEndy}=[draw=myyellow, -, draw opacity=0, -latex, line width=0.2pt, postaction={draw, opacity=1, -, line width=1pt, shorten >=3pt}, tikzit draw={rgb,255: red,235; green,190; blue,35}]
\tikzstyle{dwiry}=[draw=myyellow, ->, draw opacity=0, latex reversed-, line width=0.2pt, postaction={draw, opacity=1, -, line width=1pt, shorten >=-5pt}, tikzit draw={rgb,255: red,235; green,190; blue,35}]
\tikzstyle{dwiBy}=[draw=myyellow, ->, draw opacity=0, latex reversed-, line width=0.2pt, postaction={draw, opacity=1, -, line width=1pt}, tikzit draw={rgb,255: red,235; green,190; blue,35}]
\tikzstyle{coLy}=[draw=myyellow, out=180, in=90, looseness=0.85, -, draw opacity=0, -latex, line width=0.2pt, postaction={draw, opacity=1, -, line width=1pt, shorten >=3pt, shorten <=-1pt}, tikzit draw={rgb,255: red,235; green,190; blue,35}]
\tikzstyle{coRy}=[draw=myyellow, out=00, in=90, looseness=0.85, -, draw opacity=0, -latex, line width=0.2pt, postaction={draw, opacity=1, -, line width=1pt, shorten >=3pt, shorten <=-1pt}, tikzit draw={rgb,255: red,235; green,190; blue,35}]
\tikzstyle{ciLy}=[draw=myyellow, out=-90, in=180, looseness=0.85, -, draw opacity=0, -latex, line width=0.2pt, postaction={draw, opacity=1, -, line width=1pt, shorten >=3pt}, tikzit draw={rgb,255: red,235; green,190; blue,35}]
\tikzstyle{ciRy}=[draw=myyellow, out=-90, in=00, looseness=0.85, -, draw opacity=0, -latex, line width=0.2pt, postaction={draw, opacity=1, -, line width=1pt, shorten >=3pt}, tikzit draw={rgb,255: red,235; green,190; blue,35}]
\tikzstyle{dciLy}=[draw=myyellow, out=-90, in=180, looseness=0.85, ->, draw opacity=0, latex reversed-, line width=0.2pt, postaction={draw, opacity=1, -, line width=1pt, shorten <=1pt, shorten >=-1pt}, tikzit draw={rgb,255: red,235; green,190; blue,35}]
\tikzstyle{dciRy}=[draw=myyellow, out=-90, in=00, looseness=0.85, ->, draw opacity=0, latex reversed-, line width=0.2pt, postaction={draw, opacity=1, -, line width=1pt, shorten <=1pt, shorten >=-1pt}, tikzit draw={rgb,255: red,235; green,190; blue,35}]
\tikzstyle{chRy}=[draw=myyellow, out=0, in=180, looseness=0.85, ->, draw opacity=0, -latex, line width=0.2pt, postaction={draw, opacity=1, -, line width=1pt, shorten <=-1pt, shorten >=3pt}, tikzit draw={rgb,255: red,235; green,190; blue,35}]
\tikzstyle{fil0y}=[-, fill=myyellow, draw=myyellow, fill opacity=0, tikzit draw={rgb,255: red,235; green,190; blue,35}]
\tikzstyle{bfrmy}=[draw=myyellow, line width=1pt, -, fill=myyellow, rounded corners, tikzit draw={rgb,255: red,235; green,190; blue,35}]
\tikzstyle{chLy}=[draw=myyellow, out=180, in=0, looseness=0.85, ->, draw opacity=0, -latex, line width=0.2pt, postaction={draw, opacity=1, -, line width=1pt, shorten <=-1pt, shorten >=3pt}, tikzit draw={rgb,255: red,235; green,190; blue,35}]
\tikzstyle{clLy}=[draw=myyellow, out=180, in=180, looseness=0.85, ->, draw opacity=0, -latex, line width=0.2pt, postaction={draw, opacity=1, -, line width=1pt, shorten <=-1pt, shorten >=3pt}, tikzit draw={rgb,255: red,235; green,190; blue,35}]
\tikzstyle{clRy}=[draw=myyellow, out=0, in=0, looseness=0.85, ->, draw opacity=0, -latex, line width=0.2pt, postaction={draw, opacity=1, -, line width=1pt, shorten <=-1pt, shorten >=3pt}, tikzit draw={rgb,255: red,235; green,190; blue,35}]
\tikzstyle{eoLy}=[draw=myyellow, out=180, in=90, looseness=0.85, -, draw opacity=0, -latex, line width=0.2pt, postaction={draw, opacity=1, -, line width=1pt, shorten >=-1pt, shorten <=-1pt}, shorten >=-4pt, tikzit draw={rgb,255: red,235; green,190; blue,35}]
\tikzstyle{eoRy}=[draw=myyellow, out=00, in=90, looseness=0.85, -, draw opacity=0, -latex, line width=0.2pt, postaction={draw, opacity=1, -, line width=1pt, shorten >=-1pt, shorten <=-1pt}, shorten >=-4pt, tikzit draw={rgb,255: red,235; green,190; blue,35}]
\tikzstyle{qoLy}=[draw=myyellow, out=180, in=90, looseness=0.85, ->, draw opacity=0, -latex, line width=0.2pt, postaction={draw, opacity=1, -, line width=1pt, shorten >=3pt}, tikzit draw={rgb,255: red,235; green,190; blue,35}]
\tikzstyle{qoRy}=[draw=myyellow, out=00, in=90, looseness=0.85, ->, draw opacity=0, -latex, line width=0.2pt, postaction={draw, opacity=1, -, line width=1pt, shorten >=3pt}, tikzit draw={rgb,255: red,235; green,190; blue,35}]
\tikzstyle{dqiLy}=[draw=myyellow, out=-90, in=180, looseness=0.85, ->, draw opacity=0, latex reversed-, line width=0.2pt, postaction={draw, opacity=1, -, line width=1pt, shorten <=1pt, shorten >=0pt}, tikzit draw={rgb,255: red,235; green,190; blue,35}]
\tikzstyle{dqiRy}=[draw=myyellow, out=-90, in=00, looseness=0.85, ->, draw opacity=0, latex reversed-, line width=0.2pt, postaction={draw, opacity=1, -, line width=1pt, shorten <=1pt, shorten >=0pt}, tikzit draw={rgb,255: red,235; green,190; blue,35}]
\tikzstyle{rarry}=[draw=myyellow, ->, tikzit draw={rgb,255: red,235; green,190; blue,35}]
\tikzstyle{hc1Ly}=[draw=myyellow, out=-90, in=180, looseness=0.85, -, draw opacity=0, -latex, line width=0.3pt, postaction={draw, opacity=1, -, line width=1.5pt, shorten >=3pt}, tikzit draw={rgb,255: red,235; green,190; blue,35}]
\tikzstyle{hc1Ry}=[draw=myyellow, out=-90, in=00, looseness=0.85, -, draw opacity=0, -latex, line width=0.3pt, postaction={draw, opacity=1, -, line width=1.5pt, shorten >=3pt}, tikzit draw={rgb,255: red,235; green,190; blue,35}]
\tikzstyle{hu1Ry}=[draw=myyellow, out=0, in=270, looseness=0.85, -, draw opacity=0, -latex, line width=0.3pt, postaction={draw, opacity=1, -, line width=1.5pt, shorten >=3pt, shorten <=-0.3pt}, tikzit draw={rgb,255: red,235; green,190; blue,35}]
\tikzstyle{wireo}=[draw=myorange, -, draw opacity=0, -latex, line width=0.2pt, postaction={draw, opacity=1, -, line width=1pt, shorten >=3pt, shorten <=-1pt}, tikzit draw={rgb,255: red,245; green,130; blue,30}]
\tikzstyle{wirBo}=[draw=myorange, -, draw opacity=0, -latex, line width=0.2pt, postaction={draw, opacity=1, -, line width=1pt, shorten >=-1pt, shorten <=-1pt}, shorten >=-4pt, tikzit draw={rgb,255: red,245; green,130; blue,30}]
\tikzstyle{wEndo}=[draw=myorange, -, draw opacity=0, -latex, line width=0.2pt, postaction={draw, opacity=1, -, line width=1pt, shorten >=3pt}, tikzit draw={rgb,255: red,245; green,130; blue,30}]
\tikzstyle{dwiro}=[draw=myorange, ->, draw opacity=0, latex reversed-, line width=0.2pt, postaction={draw, opacity=1, -, line width=1pt, shorten >=-5pt}, tikzit draw={rgb,255: red,245; green,130; blue,30}]
\tikzstyle{dwiBo}=[draw=myorange, ->, draw opacity=0, latex reversed-, line width=0.2pt, postaction={draw, opacity=1, -, line width=1pt}, tikzit draw={rgb,255: red,245; green,130; blue,30}]
\tikzstyle{coLo}=[draw=myorange, out=180, in=90, looseness=0.85, -, draw opacity=0, -latex, line width=0.2pt, postaction={draw, opacity=1, -, line width=1pt, shorten >=3pt, shorten <=-1pt}, tikzit draw={rgb,255: red,245; green,130; blue,30}]
\tikzstyle{coRo}=[draw=myorange, out=00, in=90, looseness=0.85, -, draw opacity=0, -latex, line width=0.2pt, postaction={draw, opacity=1, -, line width=1pt, shorten >=3pt, shorten <=-1pt}, tikzit draw={rgb,255: red,245; green,130; blue,30}]
\tikzstyle{ciLo}=[draw=myorange, out=-90, in=180, looseness=0.85, -, draw opacity=0, -latex, line width=0.2pt, postaction={draw, opacity=1, -, line width=1pt, shorten >=3pt}, tikzit draw={rgb,255: red,245; green,130; blue,30}]
\tikzstyle{ciRo}=[draw=myorange, out=-90, in=00, looseness=0.85, -, draw opacity=0, -latex, line width=0.2pt, postaction={draw, opacity=1, -, line width=1pt, shorten >=3pt}, tikzit draw={rgb,255: red,245; green,130; blue,30}]
\tikzstyle{dciLo}=[draw=myorange, out=-90, in=180, looseness=0.85, ->, draw opacity=0, latex reversed-, line width=0.2pt, postaction={draw, opacity=1, -, line width=1pt, shorten <=1pt, shorten >=-1pt}, tikzit draw={rgb,255: red,245; green,130; blue,30}]
\tikzstyle{dciRo}=[draw=myorange, out=-90, in=00, looseness=0.85, ->, draw opacity=0, latex reversed-, line width=0.2pt, postaction={draw, opacity=1, -, line width=1pt, shorten <=1pt, shorten >=-1pt}, tikzit draw={rgb,255: red,245; green,130; blue,30}]
\tikzstyle{chRo}=[draw=myorange, out=0, in=180, looseness=0.85, ->, draw opacity=0, -latex, line width=0.2pt, postaction={draw, opacity=1, -, line width=1pt, shorten <=-1pt, shorten >=3pt}, tikzit draw={rgb,255: red,245; green,130; blue,30}]
\tikzstyle{fil0o}=[-, fill=myorange, draw=myorange, fill opacity=0, tikzit draw={rgb,255: red,245; green,130; blue,30}]
\tikzstyle{bfrmo}=[draw=myorange, line width=1pt, -, fill=myorange, rounded corners, tikzit draw={rgb,255: red,245; green,130; blue,30}]
\tikzstyle{chLo}=[draw=myorange, out=180, in=0, looseness=0.85, ->, draw opacity=0, -latex, line width=0.2pt, postaction={draw, opacity=1, -, line width=1pt, shorten <=-1pt, shorten >=3pt}, tikzit draw={rgb,255: red,245; green,130; blue,30}]
\tikzstyle{clLo}=[draw=myorange, out=180, in=180, looseness=0.85, ->, draw opacity=0, -latex, line width=0.2pt, postaction={draw, opacity=1, -, line width=1pt, shorten <=-1pt, shorten >=3pt}, tikzit draw={rgb,255: red,245; green,130; blue,30}]
\tikzstyle{clRo}=[draw=myorange, out=0, in=0, looseness=0.85, ->, draw opacity=0, -latex, line width=0.2pt, postaction={draw, opacity=1, -, line width=1pt, shorten <=-1pt, shorten >=3pt}, tikzit draw={rgb,255: red,245; green,130; blue,30}]
\tikzstyle{eoLo}=[draw=myorange, out=180, in=90, looseness=0.85, -, draw opacity=0, -latex, line width=0.2pt, postaction={draw, opacity=1, -, line width=1pt, shorten >=-1pt, shorten <=-1pt}, shorten >=-4pt, tikzit draw={rgb,255: red,245; green,130; blue,30}]
\tikzstyle{eoRo}=[draw=myorange, out=00, in=90, looseness=0.85, -, draw opacity=0, -latex, line width=0.2pt, postaction={draw, opacity=1, -, line width=1pt, shorten >=-1pt, shorten <=-1pt}, shorten >=-4pt, tikzit draw={rgb,255: red,245; green,130; blue,30}]
\tikzstyle{qoLo}=[draw=myorange, out=180, in=90, looseness=0.85, ->, draw opacity=0, -latex, line width=0.2pt, postaction={draw, opacity=1, -, line width=1pt, shorten >=3pt}, tikzit draw={rgb,255: red,245; green,130; blue,30}]
\tikzstyle{qoRo}=[draw=myorange, out=00, in=90, looseness=0.85, ->, draw opacity=0, -latex, line width=0.2pt, postaction={draw, opacity=1, -, line width=1pt, shorten >=3pt}, tikzit draw={rgb,255: red,245; green,130; blue,30}]
\tikzstyle{dqiLo}=[draw=myorange, out=-90, in=180, looseness=0.85, ->, draw opacity=0, latex reversed-, line width=0.2pt, postaction={draw, opacity=1, -, line width=1pt, shorten <=1pt, shorten >=0pt}, tikzit draw={rgb,255: red,245; green,130; blue,30}]
\tikzstyle{dqiRo}=[draw=myorange, out=-90, in=00, looseness=0.85, ->, draw opacity=0, latex reversed-, line width=0.2pt, postaction={draw, opacity=1, -, line width=1pt, shorten <=1pt, shorten >=0pt}, tikzit draw={rgb,255: red,245; green,130; blue,30}]
\tikzstyle{rarro}=[draw=myorange, ->, tikzit draw={rgb,255: red,245; green,130; blue,30}]
\tikzstyle{hc1Lo}=[draw=myorange, out=-90, in=180, looseness=0.85, -, draw opacity=0, -latex, line width=0.3pt, postaction={draw, opacity=1, -, line width=1.5pt, shorten >=3pt}, tikzit draw={rgb,255: red,245; green,130; blue,30}]
\tikzstyle{hc1Ro}=[draw=myorange, out=-90, in=00, looseness=0.85, -, draw opacity=0, -latex, line width=0.3pt, postaction={draw, opacity=1, -, line width=1.5pt, shorten >=3pt}, tikzit draw={rgb,255: red,245; green,130; blue,30}]
\tikzstyle{hu1Ro}=[draw=myorange, out=0, in=270, looseness=0.85, -, draw opacity=0, -latex, line width=0.3pt, postaction={draw, opacity=1, -, line width=1.5pt, shorten >=3pt, shorten <=-0.3pt}, tikzit draw={rgb,255: red,245; green,130; blue,30}]

%% file: macros.tex
\newcommand{\dom}{\mathsf{dom}}

\newcommand{\Nat}{\mathbb{N}}
\newcommand{\card}[1]{\sharp #1}
\newcommand{\len}[1]{\lvert #1 \rvert}

\renewcommand{\emptyset}{\emptysetAlt}

\newcommand{\fragmstyle}[1]{\text{$\mathsf{#1}$}}

\newcommand{\one}{\boldsymbol{1}}
\newcommand{\zero}{\boldsymbol{0}}

\newcommand{\LL}{\fragmstyle{LL}\xspace}
\newcommand{\MLL}{\fragmstyle{M}\LL}
\newcommand{\MLLnoUnits}{\MLL}

\newcommand{\MELL}{\fragmstyle{ME}\LL}

\newcommand{\VMELL}{\fragmstyle{V}\MELL}
\newcommand{\LLpol}{\fragmstyle{MELL}_{\mathit{pol}}}

\newcommand{\IMELL}{\fragmstyle{I}\MELL}

\newcommand{\bangMELL}{\fragmstyle{\oc}\MELL}
\newcommand{\CbN}{\LL^\fragmstyle{CbN}\xspace}
\newcommand{\CbV}{\LL^\fragmstyle{CbV}\xspace}
\newcommand{\STLC}{\fragmstyle{STLC}\xspace}
\newcommand{\outfragm}[1]{#1_\out}
\newcommand{\inpfragm}[1]{#1_\inp}
\newcommand{\CbNtext}{CbN\xspace}
\newcommand{\CbVtext}{CbV\xspace}

\newcommand{\cutelim}{cut elimination\xspace}
\newcommand{\Cutelim}{Cut elimination\xspace}
\newcommand{\cf}{cut-free\xspace}

\newcommand{\formula}[1]{#1 formula}
\newcommand{\formulae}[1]{#1 formulae}
\newcommand{\formulaWithName}[2]{#2 formula #1}
\newcommand{\formulaeWithName}[2]{#2 formulae #1}

\newcommand{\sequent}[1]{sequent of #1}
\newcommand{\sequents}[1]{sequents of #1} 
\newcommand{\provable}[1]{provable in #1}

\newcommand{\seqc}{sequent calculus\xspace}
\newcommand{\Seqc}{Sequent calculus\xspace}
\newcommand{\seqcp}{\seqc proof\xspace}

\newcommand{\seqcps}{\seqc proofs\xspace}

\newcommand{\seqcpt}[1]{\seqcp in #1}

\newcommand{\seqcpst}[1]{\seqcps in #1}

\newcommand{\seqcptWithName}[2]{\seqcp #1 in #2}

\newcommand{\seqcptWithConclusion}[2]{\seqcp in #1 with conclusion #2}

\newcommand{\seqcptWithNameAndConclusion}[3]{\seqcp #1 in #2 with conclusion #3}

\newcommand{\seqcpstWithNameAndConclusion}[3]{\seqcps #1 in #2 with conclusions #3}

\newcommand{\axsc}{$\mathit{ax}$\xspace}
\newcommand{\cutsc}{$\mathit{cut}$\xspace}
\newcommand{\drsc}{$\wn\mathit{d}$\xspace}

\newcommand{\mix}{\mathit{mix}}

\newcommand{\exch}{\text{$\mathit{ex}$}\xspace}

\newcommand{\ve}{\mathsf{N}}
\newcommand{\ar}{\mathsf{A}}
\newcommand{\tail}{\mathsf{tail}}
\newcommand{\head}{\mathsf{head}}
\newcommand{\nullgraph}{G_\epsilon}

\newcommand{\cc}{\mathsf{cc}}
\newcommand{\precg}[1]{\prec_{#1}}
\newcommand{\succg}[1]{\succ_{#1}}
\newcommand{\precio}{\precg{\io}}
\newcommand{\succio}{\succg{\io}}
\newcommand{\outdeg}{\mathsf{d}^{\textup{out}}}
\newcommand{\pf}{partial functional\xspace}

\newcommand{\pfg}{\pf~graph\xspace}

\newcommand{\rn}[1]{$#1$-normal\xspace}
\newcommand{\rnf}[1]{$#1$-normal form\xspace}

\newcommand{\sg}{switching graph\xspace}

\newcommand{\sgs}{switching graphs\xspace}

\newcommand{\gs}{ground-structure\xspace}

\newcommand{\gss}{ground-structures\xspace}

\newcommand{\pol}{polarized\xspace}

\newcommand{\gst}[1]{\gs of #1}
\newcommand{\ps}{proof-structure\xspace}

\newcommand{\pss}{proof-structures\xspace}
\newcommand{\Pss}{Proof-structures\xspace}

\newcommand{\undgs}[1]{\mathcal{G}_{#1}}
\newcommand{\pst}[1]{\ps of #1}
\newcommand{\psst}[1]{\pss of #1}

\newcommand{\pnt}[1]{\pn of #1}
\newcommand{\pnst}[1]{\pns of #1}
\newcommand{\pntWithName}[2]{\pn #1 of #2}
\newcommand{\pnstWithName}[2]{\pns #1 of #2}
\newcommand{\norm}[2]{\mathsf{norm}_#1(#2)}
\newcommand{\cut}[3]{\mathsf{cut}(#1, #2, #3)}
\newcommand{\graft}[2]{\mathsf{graft}(#1, #2)}
\newcommand{\enbox}[1]{#1^\oc}
\newcommand{\share}[1]{\mathsf{share}(#1)}
\newcommand{\deradd}[1]{#1^+}
\newcommand{\derrem}[1]{#1^-}
\newcommand{\pabs}[2]{\lambda #1 \, #2}
\newcommand{\initVar}[2]{\mathsf{init}_#1(#2)}
\newcommand{\init}[1]{\mathsf{init}(#1)}

\newcommand{\pnbox}{\mathsf{box}}
\newcommand{\undbox}[1]{\pnbox_{#1}}
\newcommand{\nd}{\mathsf{N}}

\newcommand{\tout}{\otimes_{\out}}
\newcommand{\tin}{\otimes_{\mathsf{i}}}
\newcommand{\pout}{\parr_{\out}}
\newcommand{\pin}{\parr_{\mathsf{i}}}
\newcommand{\poutter}{$\pout$-terminal\xspace}
\newcommand{\ioter}{ter\-minimal\xspace}

\newcommand{\pn}{proof-net\xspace}

\newcommand{\pns}{proof-nets\xspace}
\newcommand{\Pns}{Proof-nets\xspace}
\newcommand{\axpn}{\mathtt{ax}}
\newcommand{\cutpn}{\mathtt{cut}}
\newcommand{\onepn}{\mathtt{1}}
\newcommand{\drpn}{\wn\mathtt{d}}
\newcommand{\drcl}{\mathsf{d}_{\mathsf{c}}}
\newcommand{\wkpn}{\wn\mathtt{w}}
\newcommand{\ctpn}{\wn\mathtt{c}}
\newcommand{\unpn}{unary $\wn$\xspace}

\newcommand{\Axcut}{Axiom cut\xspace}

\newcommand{\Ucut}{Unit cut\xspace}

\newcommand{\AC}{\mathsf{AC}}
\newcommand{\C}{\mathsf{C}}
\newcommand{\ACC}{\AC\C}
\newcommand{\bset}{\mathsf{N}_\oc}
\newcommand{\w}{\mathsf{w}}
\newcommand{\Cw}{\C_{\sharp \w}}
\newcommand{\ACCw}{\AC\Cw}
\newcommand{\out}{\mathsf{o}}
\newcommand{\inp}{\mathsf{i}}

\newcommand{\pport}{principal port\xspace}

\newcommand{\ap}{auxiliary port\xspace}

\newcommand{\aps}{auxiliary ports\xspace}

\newcommand{\Dcut}{Dereliction cut\xspace}

\newcommand{\Scut}{Structural cut\xspace}

\newcommand{\Ccut}{Commutative cut\xspace}

\newcommand{\io}{{\mkern 2mu \mathsf{i} \mkern -1mu / \mkern -1.5mu \out}}
\newcommand{\iograph}[1]{#1^\io}
\newcommand{\ACio}{\AC_\io}
\newcommand{\Cio}{\C_\io}
\newcommand{\ACCio}{\ACC_\io}
\newcommand{\psset}{\mathcal{R}}
\newcommand{\ppset}{\psset^\mathnormal{pol}}
\newcommand{\psf}[1]{\psset^{#1}}

\newcommand{\conv}{$\wn$ reduction\xspace}

\newcommand{\tocut}{\tor{\cutpn}}
\newcommand{\tor}[1]{\to_{#1}}
\newcommand{\ltor}[1]{{}_{#1\!\!}\leftarrow}
\newcommand{\town}{\tor{\rwn}}

\newcommand{\cutwn}{\cutpn\rwn}
\newcommand{\tocutwn}{\tor{\cutwn}}

\newcommand{\mapstor}[1]{\mapsto_{#1}}

\newcommand{\mapstocutwn}{\mapstor{\cutpn\rwn}}
\newcommand{\rwn}{\wn}
\newcommand{\rax}{\axpn}
\newcommand{\rob}{{\onepn\mkern-2mu/\mkern-2mu\bot}}
\newcommand{\rtp}{{\otimes\mkern-2mu/\mkern-2mu\parr}}
\newcommand{\rdr}{{\oc\mkern-2mu/\mkern-2mu\drpn}}
\newcommand{\rst}{{\oc\mkern-2mu/\mkern-2mu\wn}}
\newcommand{\rcm}{{\oc\mkern-2mu/\mkern-2mu\oc}}
\newcommand{\rmu}{\mathsf{m}}
\newcommand{\re}{\mathsf{e}}
\newcommand{\ree}{\re\re}

\newcommand{\rmuw}{{\rmu\rwn}}
\newcommand{\rew}{{\re\rwn}}
\newcommand{\rfl}{{\mathsf{fl}}}
\newcommand{\rer}{{\mathsf{er}}}
\newcommand{\rpu}{{\mathsf{pu}}}
\newcommand{\deseq}{*}
\newcommand{\Deseq}[1]{#1^\deseq}
\newcommand{\seq}{sequential\xspace}

\newcommand{\bpn}{bang \pn}

\newcommand{\bpns}{bang \pns}
\newcommand{\Bpns}{Bang \pns}
\newcommand{\bpnset}{\psset^\br}
\newcommand{\nam}{named\xspace}

\newcommand{\nmfun}[1]{\nu_{#1}}
\newcommand{\injnam}{injectively named\xspace}

\newcommand{\gua}{guarded\xspace}

\newcommand{\wnsafe}{$\wn$-safe\xspace}
\newcommand{\esub}[2]{[#1/#2]}

\newcommand{\lc}{$\lambda$-calculus\xspace}
\newcommand{\lci}{$\lambda$-calculi\xspace}

\newcommand{\stset}{\mathcal{T}}
\newcommand{\intuArrow}[2]{#1 \Rightarrow #2}
\newcommand{\vset}{\mathcal{X}}

\newcommand{\lt}{$\lambda$-term\xspace}
\newcommand{\lts}{$\lambda$-terms\xspace}
\newcommand{\lterm}{\lt}
\newcommand{\lterms}{\lts}
\newcommand{\lcontext}{$\lambda$-context\xspace}
\newcommand{\lcontexts}{$\lambda$-contexts\xspace}

\newcommand{\ljudgm}{$\lambda$-judgment\xspace}

\newcommand{\lderivation}{$\lambda$-derivation\xspace}

\newcommand{\labs}[3]{\lambda #1^{#2} #3}
\newcommand{\lapp}[2]{(#1)#2}
\newcommand{\lsub}[2]{\{#1/#2\}}
\newcommand{\anlset}{\Lambda}
\newcommand{\tlset}{\Lambda}
\newcommand{\oftype}[2]{#1 : #2}

\newcommand{\sctxset}{\mathsf{C}_\lambda}
\newcommand{\ltyp}{\mathsf{t}_\lambda}
\newcommand{\bc}{bang calculus\xspace}

\newcommand{\bt}{bang term\xspace}

\newcommand{\bterm}{$\oc$-term\xspace}

\newcommand{\bterms}{$\oc$-terms\xspace}

\newcommand{\bcabs}[3]{\lambda #1^{#2} #3}
\newcommand{\bcapp}[2]{\langle #1 \rangle#2}
\newcommand{\bcbang}[1]{#1^\oc}
\newcommand{\bcder}[1]{\der #1}
\newcommand{\bcsub}[2]{\{#1/#2\}}
\newcommand{\anbcset}{\oc \Lambda}
\newcommand{\tbcset}{\oc \Lambda}
\newcommand{\bctx}{bang context\xspace}

\newcommand{\bcontext}{$\oc$-context\xspace}
\newcommand{\bcontexts}{$\oc$-contexts\xspace}

\newcommand{\bjudgm}{$\oc$-judgment\xspace}
\newcommand{\bjudgms}{$\oc$-judgments\xspace}
\newcommand{\bderivation}{$\oc$-derivation\xspace}

\newcommand{\bctxset}{\mathsf{C}_\br}
\newcommand{\linArrow}[2]{#1 \multimap #2}
\newcommand{\bctyp}{\mathsf{t}_\br}
\newcommand{\tob}{\tor{\beta}}
\newcommand{\betav}{{\beta_\mathsf{v}}}
\newcommand{\tobv}{\tor{\betav}}

\DeclareMathOperator{\der}{\mathsf{der}}
\newcommand{\br}{\mathsf{b}}

\newcommand{\betab}{\beta\oc}
\newcommand{\brder}{\mathsf{d}}

\newcommand{\btl}{$\lambda$\xspace}
\newcommand{\bta}{$@$\xspace}

\newcommand{\ljud}[3]{#1 \vdash \oftype{#2}{#3}}
\newcommand{\bjud}[3]{#1 \vdash \oftype{#2}{#3}}
 
\newcommand{\bderive}[4]{#1 \rhd \bjud{#2}{#3}{#4}} 
\newcommand{\fv}[1]{\mathsf{fv}(#1)}

\newcommand{\wh}{\circ}
\newcommand{\bl}{\bullet}
\newcommand{\vmellt}{\star}
\newcommand{\VMELLt}[1]{#1^\vmellt}
\newcommand{\cbnt}{\mathtt{n}}
\newcommand{\cbvt}{\mathtt{v}}

\newcommand{\CBNt}[1]{#1^\cbnt}
\newcommand{\CBVt}[1]{#1^\cbvt}

\newcommand{\cbnT}{\mathtt{N}}
\newcommand{\cbvT}{\mathtt{V}}
\newcommand{\CBNT}[1]{#1^\cbnT}
\newcommand{\CBVT}[1]{#1^\cbvT}

\newcommand{\cbn}{call-by-name\xspace}
\newcommand{\cbv}{call-by-value\xspace}
\newcommand{\CBNp}[1]{#1^\cbnT}
\newcommand{\CBVp}[1]{#1^\cbvT}
\newcommand{\WH}[1]{#1^\wh}
\newcommand{\BL}[1]{#1^\bl}

%% file: Sections/Introduction.tex
%

One of the main contributions of Linear Logic ($\LL$~\cite{girard1987linear}) to the field of proof theory is the shift from a term-like or tree-like representation of proofs to a more general \emph{graph}-like one. Graphs representing proofs in \LL are usually called \emph{\pns}: their syntax is richer and more expressive than the one of the \lc, or the one of \seqc, or natural deduction. Contrary to the mentioned formalisms, \pns are special inhabitants of the wider land of \emph{\pss}. A \ps is any ``graph'' that can be built in the language of \pns, and it need not represent a proof. \Pss are well-behaved for performing computations: 
the procedure of \cutelim can be applied to \pss, and that they can be interpreted in denotational models thanks to the notion of experiment \cite{girard1987linear}. 

In the theory of \pns, one of the main research lines 
is the study of \emph{correctness criteria} for \pss: 
finding abstract properties allowing us to isolate, among \pss, exactly those that are indeed proofs. 
The property $\ACC$ (a.k.a.~Danos-Regnier criterion) states that every 
\sg (see our \Cref{def:switching}) 
is acyclic and connected: it is one of the most famous correctness criteria for the purely multiplicative fragment $\MLLnoUnits$ \cite{danos1989structure}. 
In the multiplicative exponential fragment \MELL (the fragment of \LL in which it is natural to represent and study the majority of logical and computational phenomena), one often considers \pss satisfying the property $\AC$, obtained from $\ACC$ by dropping connectivity. 
As $\ACC$ for $\MLLnoUnits$ \pss, $\AC$ is stable with respect to \cutelim of \MELL \pss; but $\AC$ 
characterizes those \MELL \pss that are 
proofs~only in an extended sense: the \MELL \seqcps with 
the $\mix$ rule \cite{fleury1994the,llhandbook}. 

There are several ways to present the contributions of the present paper. 
The first one is 
to better understand the logical meaning of the connectivity property. 
From a conceptual point of view, the inception of \pns and correctness criteria in proof theory is a real breakthrough. Most notably, the Danos-Regnier correctness criterion relates a purely logical property (to be a logical proof) and a purely geometric one (\sgs are trees). Since connectivity is undeniably a simple and solid geometric property, one can try to reverse the point of view and address the question of the logical counterpart of connectivity. For example, in~\cite{guerrieri_et_al:LIPIcs.FSCD.2016.20} a remarkable property (that does not hold for general \pss of \MELL) is proven for \emph{connected} \pss of \MELL: the element of order $2$ of the Taylor expansion of a \ps is enough to entirely rebuild the \ps. In the same spirit, in this paper we revisit an old property which is in between $\AC$ and $\ACC$. Connectivity fails in \MELL due to the presence of the weakening rule of logic and of the rule for the multiplicative constant bottom: it is then natural to relate the number of connected components ($\cc$) and the nodes weakening or bottom ($\w$) through the equality $\smash{\cc = \w + 1}$, 
a property called $\ACCw$ in this paper: 
it was first considered in~\cite{fleury1994the}, where 
it is shown to be a necessary condition for a \ps to represent an actual proof, but it is not sufficient due to the counterexample in \Cref{subfig:fleury-retore-intuitionistic}. 
$\ACCw$ is stable under \cutelim of \psst{\MELL} 
(\Cref{thm:accw-cut-elim}, but its idea dates back to \cite{fleury1994the}).
It is common practice \cite{TdF2000,lamarche2008proof,llhandbook} to add ``jumps'' to \pss in order to recover the property $\ACC$, 
but such a technique has many drawbacks
(there is no canonical position for jumps in general, their behaviour through \cutelim is a mess). 
As simple and solid geometric properties are likely to have a logical meaning, we wondered whether it was possible to identify interesting subsets of \pss~of \MELL for which $\ACCw$~is a correctness criterion. 
We found a geometric restriction such that, for every \ps satisfying it, $\ACCw$ is sufficient to be a proof: \Cref{thm:poutter-accio-sequential}. An equivalent reformulation~shows~that~$\ACCw$~can~be~checked~in~linear~time~relatively~to~the~size of the \ps (\Cref{rmk:complexity}).

A second 
reason for interest in our work is an (apparently) unrelated question: is it possible to provide a unifying setting for the classical and intuitionistic polarizations? The classical polarization was introduced in~\cite{girard1991lc} and one can show that classical logic can be represented in the classical polarized fragment of $\LL$ (see for example~\cite{laurent2005polarized}). 
Intuitionistic logic and the $\lambda$-calculus can be embedded into the intuitionistic polarized fragment of \LL
\cite{danos1990logique,regnier1992lambda,Bierman94,DanosRegnier95,Lamarche95,Lamarche96}.
Specific correctness criteria were introduced in the classical case (Laurent, see~\cite{laurent1999polarized,laurent2003polarized}) and in the intuitionistic one (Lamarche, see~\cite{lamarche2008proof}). 
However, up to now (to the best of our knowledge) these two settings have been considered separately. 
It is quite natural to wonder whether or not there exists a fragment of $\LL$ where the two polarizations coexist: this is one of the features of the fragment \VMELL introduced here (Definition~\ref{def:vmell}). 
It turns out that \pss of \VMELL without axiom cuts and multiplicative cuts satisfy the geometric condition of \Cref{thm:poutter-accio-sequential} and then $\ACCw$ is a correctness criterion in this setting (\Cref{cor:vmell,rmk:vmell-normal}). 
So, \VMELL corroborates the validity of our approach: connectivity naturally appears at the crossroad of classical and intuitionistic polarizations. 
The (undirected) world of $\ACCw$ \pss is computationally sound (\Cref{thm:accw-cut-elim}) and underlies the specific (directed) correctness criteria of~\cite{laurent1999polarized,laurent2003polarized} and~\cite{lamarche2008proof}: $\ACCw$ is equivalent to Laurent's criterion in the classical polarized setting (
\Cref{cor:classical-polarization}), it captures an essential feature of the intuitionistic polarization (\Cref{cor:imell}) and it is a correctness criterion for the~fragment~\VMELL,~containing both classically and intuitionistically \pol formulae.

A third 
reason for interest in our work is ``computational''.
$\LL$ inspired the introduction of the \emph{\bc} \cite{ehrhard2016the,GuerrieriManzonetto18} subsuming both the \cbn and \cbv evaluation mechanisms of the \lc thanks to Girard's two translations of intuitionistic logic into $\LL$. A natural question is to precisely identify the image in $\LL$ of the \bc and to study its static and dynamic properties. 
The resulting fragment \bangMELL is contained~in~\VMELL~and~thus $\ACCw$ is a correctness criterion for \bangMELL.

Even though this is not the focus of our work, we point out that \VMELL is also meaningful from a semantic perspective, as its intersection with intuitionistic \MELL is essentially the counterpart in \MELL of the half-polarized typing system for call-by-push-value in \cite{ehrhard19probabilistic}, and \VMELL itself inherits the semantic interpretation of that system (\Cref{rmk:vmell-semantics}).

 \begin{outline}
  In \Cref{sec:geometric-preliminaries}, we recall some elementary notions of graph theory, define \pss 
  and its properties $\AC$ and $\ACCw$, and generalize the notion of sequentializable \ps to the untyped setting (as in \cite{llhandbook}): intuitively, a \seq \ps represents a \seqcp with no formula. 
We also show that a sequential \ps always satisfies $\ACCw$ (\Cref{prop:accw-necessary}) and prove that $\ACCw$ is stable under cut elimination (\Cref{thm:accw-cut-elim}). 
  In \Cref{sec:pure-polarized}, we introduce a very weak form of typing discipline (\Cref{fig:polarized}), thus obtaining polarized \pss (\Cref{def:polarized}), and we prove our main result: a \poutter \ps satisfying $\ACCw$ is sequential (\Cref{prop:poutter-accio-sequential}, \Cref{thm:poutter-accio-sequential}). 
  The proof involves the equivalence between $\ACCw$ and the property $\ACCio$, which is computationally easy to check (\Cref{rmk:complexity}). In \Cref{sec:typed-setting}, we move to the typed setting: we recall the classical and intuitionistic polarizations $\LLpol$ and \IMELL, 
  and we introduce the new unifying fragment \VMELL (\Cref{def:vmell}). 
  Since a \pst{\VMELL} (without axiom cuts and multiplicative cuts) is \poutter, when it satisfies $\ACCw$ it is sequentializable (\Cref{cor:vmell}): $\ACCw$ is a correctness criterion for \VMELL (\Cref{rmk:vmell-normal}), 
  generalizing Laurent's correctness criterion for $\LLpol$ \cite[Definition~5]{laurent2003polarized}. 
  In \Cref{sec:bang-calculus}, we present new results explaining the precise relationship between \pns and the \bc. 
  First, we notice that the image in \LL of the \bc is the fragment \bangMELL, contained in \VMELL, and then we introduce a translation of the \bc terms into \pns that factorizes both Girard's \cbn and \cbv translations of the \lc (\Cref{thm:factorization}). We prove that every reduction step in the \bc is simulated by a finite number of \cutelim steps in \pns (\Cref{thm:simulation})~and~finally~we provide~an~explicit characterization of \pns that represent terms of the \bc~(\Cref{thm:characterization}).

 \end{outline}

%
  
%

%% file: Sections/Geometric.tex
We first recall some elementary notions of graph theory with their notations.

  A \emph{graph} is a quadruple $G = (\ve(G), \ar(G), \tail_G, \head_G)$, where $\ve(G)$ is a~finite set of \emph{nodes}, $\ar(G)$ is a finite set, disjoint from $\ve(G)$, of \emph{arcs}, and $\tail_G$, $\head_G$ are two functions~from~$\ar(G)$~to $\ve(G)$. 
The \emph{ends} of $a \in \ar(G)$ are $\tail_G(a)$ and $\head_G(a)$.
  The graph $G$ is \emph{null} if $\ve(G) = \varnothing$.

A graph $G$ is naturally directed, as $\head_G$ and $\tail_G$ induce the orientation of any arc $a$ from $\tail_G(a)$ to $\head_G(a)$. 
Nonetheless, we can also see $G$ as undirected thanks to the definition of (undirected) path. 
  A \emph{path} of $G$ (\emph{from} $n_0$ \emph{to} $n_k$) is a finite sequence $\gamma = n_0 a_1 n_1 a_2 \dots n_{k - 1} a_k n_k$ of alternating nodes and arcs of $G$ such that: for all $i \in \{1, \dots, k\}$, the ends of $a_i$ are $n_{i - 1}$ and $n_i$; and for all $i, j \in \{0, \dots, k\}$ with $i \neq j$, if $n_i = n_j$, then $\{i, j\} = \{0, k\}$; and $a_1 \neq a_k$ if $k \geq 2$. 
  A path $\gamma$ of $G$ is: \emph{directed} if $\tail_G(a_i) = n_{i - 1}$ and $\head_G(a_i) = n_i$ for all $i \in \{1, \dots, k\}$;~\emph{empty}~if $k = 0$; a \emph{cycle} if $k > 0$ and $n_0 = n_k$. 

  A graph $G$ is: (\emph{directed}) \emph{acyclic} if it has no (directed) cycle; \emph{connected} when, for every $n, m \in \ve(G)$, there exists a path of $G$ from $n$ to $m$. Every graph $G$ can be expressed in a unique way as a disjoint union of non-null connected graphs, which are called the \emph{connected components} of $G$. The set consisting of the connected components of $G$ is denoted by $\cc(G)$\index{cc(G)@$\cc(G)$}.\footnotemark
  \footnotetext{Note that the null graph is connected and it is the only graph with no connected components.}
  The number of elements of a set $S$ is denoted by $\card{S}$.



 \begin{proposition}[{\cite[Theorem 11.10.3]{lehman2010mathematics}}]
  \label{prop:connected-components}
  For any acyclic graph $G$:  $\sharp \cc(G) = \sharp \ve(G) - \sharp \ar(G)$.
 \end{proposition}


  \label{def:directed-acyclic-order}
  If $G$ is a graph, we define a binary relation $\precg{G}$ on $\ve(G)$: for all $n_1, n_2 \in \ve(G)$, we~have~that $n_1 \precg{G} n_2$ if and only if there exists a non-empty directed path of $G$ from $n_1$ to $n_2$.

 \begin{remark}
  \label{rmk:directed-acyclic-order}
  Let $G$ be a directed acyclic graph. Then $\precg{G}$ is a strict order relation.
 \end{remark}

  In a graph $G$, the \emph{outdegree} of $n \in \ve(G)$ is $\outdeg_G(n) \coloneq \card{\{a \in \ar(G)} : n = \tail_G(a)\}$. 
  We say that: $n$ is a \emph{sink} of $G$ if $\outdeg_G(n) = 0$;
  $G$ is \emph{\pf} (\emph{fun} for short) if $\outdeg_G(n) \leq 1$ for all $n \in \ve(G)$.
By basic combinatorial results, the property below holds (proof in \Cref{subsec:pfg-sink}).


 \begin{proposition}[restate = pfgSink, name = ]
  \label{prop:pfg-sink}
  A directed acyclic non-null connected fun graph has exactly one~sink.
 \end{proposition}
 


\subparagraph*{\texorpdfstring{Untyped \pss}{Untyped proof-structures}}

 We now present a special kind of graphs: \pss \cite{girard1987linear,llhandbook}, used to represent proofs in linear logic.
 For the moment, for the sake of generality, we stress their geometric properties, thus they are untyped (do not refer to any linear logic formulae).

 In a graph $G$, $a \in \ar(G)$ is a \emph{conclusion} of $\tail_G(a)$ and a \emph{premise} of $\head_G(a)$. In pictures, a double arrow (to a $\bullet$ node) is a shorthand for $k \geq 0$ arcs (each targeting a distinct $\bullet$ node).
 
 \begin{definition}
  \label{def:ground-structure}
  A \emph{\gs}\index{ground-structure@\gs} is a graph with node labels $\axpn$, $\cutpn$, $\onepn$, $\bot$, $\otimes$, $\parr$, $\oc$, $\drpn$, $\wn$, $\bullet$, satisfying the local constraints in \Cref{fig:nodes-local-constraints}, and with total orderings (separately) on the premises of its $\bullet$ nodes, on the conclusions of every $\oc$ node and on the premises~of~every~$\otimes$~and~$\parr$~node.
  
  \begin{figure}
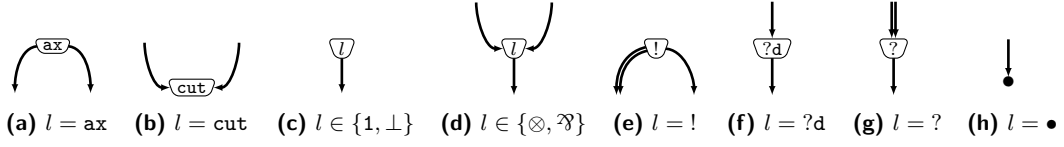

   \centering
   \begin{subfigure}{0.1\textwidth}
    \centering
    \input{Figures/axiom.tikz}
    \caption{$l = \axpn$}
   \end{subfigure}
   \hskip 0.5em
   \begin{subfigure}{0.1125\textwidth}
    \centering
    \input{Figures/cut.tikz}
    \caption{$l = \cutpn$}
   \end{subfigure}
   \hskip 0.5em
   \begin{subfigure}{0.1375\textwidth}
    \centering
    \input{Figures/one-gen.tikz}
    \caption{$l \in \{\onepn,\bot\}$}
   \end{subfigure}
   \hskip 0.5em
   \begin{subfigure}{0.1375\textwidth}
    \centering
    \input{Figures/tensor-gen.tikz}
    \caption{$l \in \{\otimes,\parr\}$}
   \end{subfigure}
   \hskip 0.5em
   \begin{subfigure}{0.0875\textwidth}
    \centering
    \input{Figures/bang.tikz}
    \caption{$l = \oc$}
   \end{subfigure}
   \hskip 0.5em
   \begin{subfigure}{0.1\textwidth}
    \centering
    \input{Figures/dereliction.tikz}
    \caption{$l = \drpn$}
   \end{subfigure}
   \hskip 0.5em
   \begin{subfigure}{0.0875\textwidth}
    \centering
    \input{Figures/why-not.tikz}
    \caption{$l = \wn$}
   \end{subfigure}
   \hskip 0.5em
   \begin{subfigure}{0.0875\textwidth}
    \centering
    \input{Figures/conclusion.tikz}
    \caption{$l = \bullet$}
   \end{subfigure}
   \caption{\label{fig:nodes-local-constraints}The local constraints that a node of a \gs with label $l$ has to respect.}
  \end{figure}

  Let $\mathcal{G}$ be a \gs. 
  If $n$ is a $\oc$ (resp.~$\otimes$ or $\parr$) node of $\mathcal{G}$, the \emph{\pport} (resp.~\emph{right premise}) of $n$ is its greatest conclusion (resp.~premise). The other~conclusions (resp.~premise) of $n$ are (resp.~is) called the \emph{\aps} (resp.~\emph{left~premise})~of $n$. A $\wn$ node with $k$ premises is a $\wkpn$ node when $k = 0$, a $\ctpn$ node when $k \geq 2$. If $n$ is a $\bullet$ node, the premise of $n$ is called a \emph{conclusion} of $\mathcal{G}$. If $n$ is not a $\bullet$ node and all its conclusions are conclusions of $\mathcal{G}$, the node $n$ is \emph{terminal}. 
  We denote by $\bset(\mathcal{G})$~(resp.~$\w(\mathcal{G})$) the set~of~$\oc$~(resp.~$\bot$~or~$\wkpn$)~nodes~in~$\mathcal{G}$.
 \end{definition}
 
 \begin{definition}
  \label{def:pre-proof-structure}
	We define the set $\psset^d$ by induction on $d \in \Nat$:
	$R = (\undgs{R}, \undbox{R}) \in \psset^{d}$ if $\undgs{R}$ is a \gs and  $\undbox{R}$ is a function mapping each $n \in \bset(\undgs{R})$ to its \emph{box} $\undbox{R}(n) \in \psset^{d-1}$ 
	whose underlying \gs has the same number of conclusions as $n$.\footnotemark
	\footnotetext{So, $R = (\undgs{R}, \undbox{R}) \in \psset^{0}$ if $\undgs{R}$ is a \gs without $\oc$ nodes and $\undbox{R}$ is the empty function.}
	The \emph{depth} of $R \in \bigcup_{d \in \Nat} \psset^d$ is the smallest $d$ such that $R \in \psset^d$. 
		A \emph{conclusion} of $R$ is a conclusion of $\undgs{R}$.
	
	We set $\hat{\psset} \coloneqq \bigcup_{d \in \Nat} \psset^d$. 
	Given $R \in \hat{\psset}$, we write $R' \sqsubset R$ if $R' = \undbox{R}(n)$ for some $n \in \bset(\undgs{R})$.  
	We define\footnotemark 
	\footnotetext{\label{foot:depth}The definition (as well as \Cref{def:outp-terminal,def:ACCw,def:ground-structure-mell}) is by induction on the depth of $R$.}
	$\nd(R) \coloneq \nd(\undgs{R}) \cup \bigcup_{R' \sqsubset R} \nd(R')$ and $\ar(R) \coloneq \ar(\undgs{R}) \cup \bigcup_{R' \sqsubset R} \ar(R')$. 
	
	\label{def:proof-structure}
	A \emph{\ps}\index{proof-structure@\ps} is an $R \in \hat{\psset}$ such that:\textsuperscript{\normalfont\ref{foot:depth}} for all $n,n' \in \bset(\undgs{R})$, if $n \neq n'$ then $\nd(\undbox{R}(n)) \allowbreak\cap \nd(\undbox{R}(n')) = \emptyset = \ar(\undbox{R}(n)) \cap \ar(\undbox{R}(n'))$; and each $R' \sqsubset R$ is a \ps with $\nd(R') \cap \nd(\undgs{R}) = \emptyset = \ar(R') \cap \ar(\undgs{R})$. 
	The set of \pss is denoted by $\psset$.
 \end{definition}
 
 A \ps $R$ is depicted by replacing every $\oc$ node $n$ of $\undgs{R}$ with a rectangular~frame containing $\undbox{R}(n)$, the \pport of $n$ being under the $\oc$ symbol. 
Some \pss can be constructed inductively according to a sequential procedure (\cite{llhandbook}, see the next definition): \Cref{rmk:sequential} shows they correspond to \seqc proofs in 
\MELL and its fragments.
 
 \begin{definition}[by induction on $\card{\ar(R)}$]
  \label{def:sequential}
  A \ps $R$ is \emph{$n$-\seq} if $n \in \nd(\undgs{R})$ and $R$ has one of the shapes in \Cref{fig:desequentialization} (ignoring formulae and sequents), depending on the label of $n$, for some \seq $R_1, R_2 \in \psset$, where \emph{\seq} means $m$-\seq for some~node~$m$.
  
  \begin{figure}
  	\vspace*{-\baselineskip}
   \centering
   \begin{subfigure}{0.1125\textwidth}
    \centering
    \input{Figures/desequentialization-axiom.tikz}
    \caption{\label{subfig:desequentialization-axiom}\axsc rule.}
   \end{subfigure}
   \quad
   \begin{subfigure}{0.1875\textwidth}
    \centering
    \input{Figures/desequentialization-cut.tikz}
    \caption{\label{subfig:desequentialization-cut}\cutsc rule.}
   \end{subfigure}
   \quad
   \begin{subfigure}{0.1\textwidth}
    \centering
    \input{Figures/desequentialization-one.tikz}
    \caption{\label{subfig:desequentialization-one}$\one$ rule.}
   \end{subfigure}
   \quad
   \begin{subfigure}{0.1125\textwidth}
    \centering
    \input{Figures/desequentialization-bottom.tikz}
    \caption{\label{subfig:desequentialization-bottom}$\bot$ rule.}
   \end{subfigure}
   \vskip 1em
   \begin{subfigure}{0.1875\textwidth}
    \centering
    \input{Figures/desequentialization-tensor.tikz}
    \caption{\label{subfig:desequentialization-tensor}$\otimes$ rule.}
   \end{subfigure}
   \quad
   \begin{subfigure}{0.15\textwidth}
    \centering
    \input{Figures/desequentialization-par.tikz}
    \caption{\label{subfig:desequentialization-par}$\parr$ rule.}
   \end{subfigure}
   \quad
   \begin{subfigure}{0.1375\textwidth}
    \centering
    \input{Figures/desequentialization-bang.tikz}
    \caption{\label{subfig:desequentialization-bang}$\oc$ rule.}
   \end{subfigure}
   \quad
   \begin{subfigure}{0.1125\textwidth}
    \centering
    \input{Figures/desequentialization-dereliction.tikz}
    \caption{\label{subfig:desequentialization-dereliction}\drsc rule.}
   \end{subfigure}
   \quad
   \begin{subfigure}{0.15\textwidth}
    \centering
    \input{Figures/desequentialization-why-not.tikz}
    \caption{\label{subfig:desequentialization-why-not}$?$ rule.}
   \end{subfigure}
   \caption{\label{fig:desequentialization}Definition of \seq \ps (and desequentialization of \seqcpst{\MELL} into \psst{\MELL}). It is intended that $n$ replaces $\bullet$ nodes in $R_1$ and $R_2$.}
  \end{figure}
 \end{definition}

 We relate the inductive notion of sequential \ps to purely geometric properties.

 \begin{definition}
  \label{def:switching}
  A \emph{switching} $\varphi$ of a \gs $\mathcal{G}$ is a function mapping every $\parr$ or $\ctpn$ node $n$ of $\mathcal{G}$ to a premise of $n$. The \emph{\sg of $\mathcal{G}$ induced by $\varphi$}, denoted by $\mathcal{G}^\varphi$, is the graph we obtain by erasing, for each $\parr$ or $\ctpn$ node $n$ of $\mathcal{G}$, every premise of $n$ except~for~$\varphi(n)$.
 
 
%
%

	\label{def:correctness-conditions}
	Let $\mathcal{G}$ be a \gs, and $\varphi$ be a switching of $\mathcal{G}$. We write: 
	\begin{itemize}
		\item
		$\mathcal{G}^\varphi \models \AC$ if $\mathcal{G}^\varphi$ is acyclic;
		and $\mathcal{G}^\varphi \models \Cw$ if $\card{\cc(\mathcal{G}^\varphi)} = \card{\w(\mathcal{G})} + 1$;

	\item
	$\mathcal{G} \models \AC$ (resp.~$\mathcal{G} \models \Cw$) if $\mathcal{G}^\varphi \models \AC$ (resp.~$\mathcal{G}^\varphi \models \Cw$) for every switching $\varphi$ of $\mathcal{G}$;
	\item
	$\mathcal{G} \models \ACCw$ if $\mathcal{G} \models \AC$ and $\mathcal{G} \models \Cw$.
\end{itemize}
\end{definition}
 \begin{remark}
  \label{rmk:switching-connected-components}
  Let $\mathcal{G}$ be a \gs. By \Cref{prop:connected-components}, if $\mathcal{G} \models \AC$ then $\card{\cc(\mathcal{G}^{\varphi_1})} = \card{\cc(\mathcal{G}^{\varphi_2})}$ for every $\varphi_1, \varphi_2$ switchings of $\mathcal{G}$, since they only erase the same number of arcs.
 \end{remark}


 \begin{definition}
 	\label{def:ACCw}
  Let $R \in \psset$. We write $R \models \AC$ (resp.~$\Cw$, $\ACCw$) if $\undgs{R} \models \AC$ (resp.~$\Cw$, $\ACCw$) and $R' \models \AC$ (resp.~$\Cw$,~$\ACCw$) for every $R' \sqsubset R$.
 \end{definition}
%
%
%
%
 \begin{proposition}
  \label{prop:accw-necessary}
  If $R \in \psset$ is \seq then $R \models \ACCw$ (proof by easy induction on~$\card{\ar(R)}$).
 \end{proposition}
 
  The converse of \Cref{prop:accw-necessary} fails, as shown in \Cref{subfig:fleury-retore-intuitionistic} (due to \cite{fleury1994the}, ignore formulae).\medskip
  \begin{figure}[t]
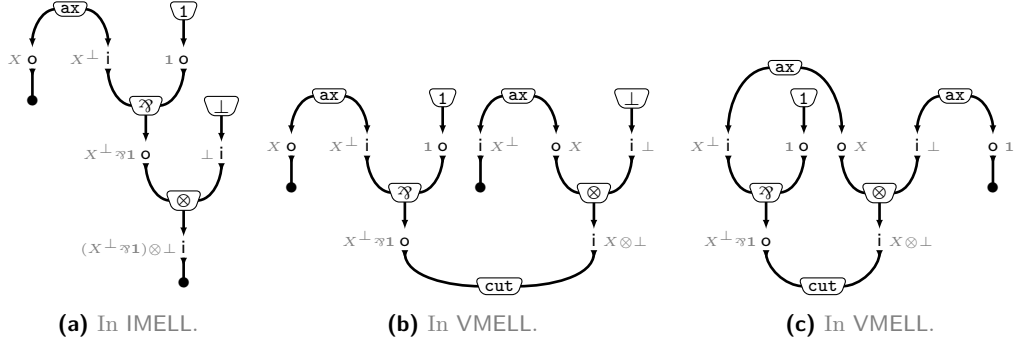

  	\vspace*{-1.5\baselineskip}
   \centering
   \subcaptionbox{\label{subfig:fleury-retore-intuitionistic}\textcolor{gray}{In \IMELL.}}{
    \centering
    \input{Figures/fleury-retore-intuitionistic.tikz}
   }
   \subcaptionbox{\label{subfig:fleury-retore-cut}\textcolor{gray}{In \VMELL.}}{
    \centering
    \input{Figures/fleury-retore-cut.tikz}
   } \quad
   \subcaptionbox{\label{subfig:fleury-retore-cycle}\textcolor{gray}{In \VMELL.}}{
    \centering
    \input{Figures/fleury-retore-cycle.tikz}
   }
   \caption{Three non-sequential \pss. The first two satisfy $\ACCw$, the last one does~not.}
  \end{figure}
 
\emph{\Cutelim} is the reduction $\tocut$ on $\psset$ generated by any of the root steps in \Cref{fig:cutelim} (note that they do not refer to formulae).
For each reduction $\tor{\mathsf{r}}$ defined in \Cref{fig:cutelim}, $R \in \psset$ is \emph{$\mathsf{r}$-normal} if there is no $R' \in \psset$ such that $R \tor{\mathsf{r}} R'$.
We show that $\ACCw$ is stable under~$\tocut$; the proof is in~\Cref{subsec:cut-elimination-steps}, though it is folklore and its crucial elements are in~\cite{fleury1994the}.
 \begin{figure}[!tb]
 	\centering
 	\begin{subfigure}{0.2875\textwidth}
 		\centering
 		\scalebox{.9}{\input{Figures/cutelim-axiom.tikz}}
 		\caption{\label{fig:ax-cut}\Axcut elimination step
 			($n \neq m$, possibly $n = n'$).
 		}
 	\end{subfigure}
 	\begin{subfigure}{0.2375\textwidth}
 		\centering
 		\scalebox{.9}{\input{Figures/cutelim-unit.tikz}}
 		\caption{\Ucut elimination step (it erases a connected component).}
 	\end{subfigure}
 	\begin{subfigure}{0.45\textwidth}
 		\centering
 		\scalebox{.9}{\input{Figures/cutelim-multiplicative.tikz}}
 		\caption{
 			\label{fig:tensor-par-cut}Tensor/par cut elimination step.}
 	\end{subfigure}
 	\vskip 1em
 	\begin{subfigure}{0.4\textwidth}
 		\centering
 		\scalebox{.9}{\input{Figures/cutelim-dereliction.tikz}}
 		\caption{\label{subfig:cutelim-rdr}\Dcut elimination step: the box is opened.}
 	\end{subfigure}
 	\begin{subfigure}{0.5875\textwidth}
 		\centering
		\scalebox{.9}{\input{Figures/cutelim-commutative.tikz}}
 		\caption{\label{subfig:cutelim-rcm}\Ccut elimination step ($n \neq m$).}
 	\end{subfigure} 
 	\vskip 1em
 	\begin{subfigure}{0.975\textwidth}
 		\centering
		\scalebox{.9}{\input{Figures/cutelim-structural.tikz}}
 		\caption{\label{subfig:cutelim-rst}\Scut elimination step: the box is erased if $k = 0$ ($\wkpn$ node), duplicated if $k \geq 2$ ($\ctpn$ node).}
 	\end{subfigure}
 	\caption{\label{fig:cutelim}
 		Root steps of \cutelim $\tocut \mathop{\coloneqq} \tor{\rax} \cup \tor{\rob} \cup \tor{\rtp} \cup \tor{\rdr} \cup \tor{\rcm} \cup \tor{\rst}$ on~$\psset$.
 		We set $\tor{\rmu} \mathop{\coloneqq} \tor{\rax} \cup \tor{\rtp}$ and $\tor{\re} \mathop{\coloneqq} \tor{\rcm} \cup \tor{\rst}$.
 	}
 \end{figure}

 \begin{theorem}[restate = accwCutElim, name = ]
  \label{thm:accw-cut-elim}
  Let $R, R' \in \psset$ such that $R \tocut R'$. 
  If $R \models \ACCw$, then $R' \models \ACCw$.
 \end{theorem}

%% file: Sections/Polarized.tex
Untyped \pss are highly wild. We propose, in this section, a geometric condition on these \pss under which the notion of $\ACCw$ is equivalent to sequentiality: this is  \Cref{thm:poutter-accio-sequential}. A \ps $R$ satisfies this geometric condition when it meets two requirements (Definition~\ref{def:outp-terminal}): 1) a global one: by labelling the arcs of $R$ with dual ``pure formulae'' $\inp$ and $\out$, one can turn $R$ into a \emph{polarized} \ps; 2) a local~one:~the~chosen~polarized labelling of $R$ is \poutter.

Pure polarization is a form of very weak typing that adds constraints on how to assemble nodes in a \ps; it is still liberal enough to allow the \emph{untyped} call-by-name \cite{danos1990logique,regnier1992lambda} and call-by-value \cite{Accattoli15} $\lambda$-calculi to be embedded into our pure \pol \pss; but it is not stable wrt cut elimination (\Cref{rmk:PolNoStableCutElim}). For a \poutter (polarized) \ps, the property $\ACCw$ (which refers to \emph{all} switching graphs of a \ps) turns out to be equivalent to simpler geometric properties of the \ps itself (referring to a \emph{unique} graph) and thus (thanks to \Cref{thm:poutter-accio-sequential}) for \poutter \pss checking sequentiality is easy (\Cref{rmk:complexity}). 
Omitted proofs 
are in~\Cref{subsec:complements-untyped-polarized}.

\begin{definition}
 \label{def:polarized}
A \gs $\mathcal{G}$ is \emph{\pol} if it comes with arc labels $\inp$ (for~\emph{input}) and $\out$ (for \emph{output}) satisfying the local constraints in \Cref{fig:polarized}.
A $\drpn$ (resp.~$\bullet$) node of $\mathcal{G}$ is called \emph{classical} (resp.~\emph{output}) if its premise~is~output,~\emph{intuitionistic} (resp.~\emph{input}) otherwise. 
We denote by $\drcl(\mathcal{G})$ (resp.~$\out(\mathcal{G})$) the set of classical $\drpn$ (resp.~output $\bullet$) nodes of $\mathcal{G}$. 
A $\parr$ node of $\mathcal{G}$ is a $\pout$ (resp.~$\pin$)~node~if its conclusion is output (resp.~input). We say that $\mathcal{G}$ is \emph{\poutter} 
if the conclusion of every $\pout$ node in $\mathcal{G}$ is a premise of a $\pout$ or $\bullet$ node of $\mathcal{G}$.

 \begin{figure}[t]
 \vspace*{-\baselineskip}
  \centering
   \adjustbox{valign=c}{\input{Figures/polarized-axiom.tikz}}
  \quad
   \adjustbox{valign=c}{\input{Figures/polarized-cut.tikz}}
  \quad
   \adjustbox{valign=c}{\input{Figures/polarized-one.tikz}}
  \quad
   \adjustbox{valign=c}{\input{Figures/polarized-bottom.tikz}}
  \quad
   \adjustbox{valign=c}{\input{Figures/polarized-bang.tikz}}
  \qquad
   \adjustbox{valign=c}{\input{Figures/polarized-dereliction-classical.tikz}}
   \adjustbox{valign=c}{or }
   \adjustbox{valign=c}{\input{Figures/polarized-dereliction-intuitionistic.tikz}}
  \qquad
   \adjustbox{valign=c}{\input{Figures/polarized-why-not.tikz}}
  \vskip 0.5em
   \adjustbox{valign=c}{\input{Figures/polarized-tout.tikz}}
   \adjustbox{valign=c}{or }
   \adjustbox{valign=c}{\input{Figures/polarized-tin-left.tikz}}
   \adjustbox{valign=c}{or }
   \adjustbox{valign=c}{\input{Figures/polarized-tin-right.tikz}}
  \qquad
   \adjustbox{valign=c}{\input{Figures/polarized-pin.tikz}}
   \adjustbox{valign=c}{or }
   \adjustbox{valign=c}{\input{Figures/polarized-pout-left.tikz}}
   \adjustbox{valign=c}{or }
   \adjustbox{valign=c}{\input{Figures/polarized-pout-right.tikz}}
  \caption{\label{fig:polarized}Local constraints for the arc labels of a \pol \gs. In~the~case~of~a~$\oc$ node, the intended meaning is that the \pport is labelled by $\out$ and every \ap~is~labelled~by~$\inp$.}
 \end{figure}
\end{definition}

Polarization narrows the way we can build \gss: \textit{e.g.} in a polarized \gs, a $\otimes$ (resp.~$\parr$) node cannot have all premises as conclusions of $\bot$ (resp.~$\one$) nodes.

\begin{definition}
	\label{def:accio}
 The \emph{$\io$ graph} of a \pol \gs $\mathcal{G}$, denoted by $\iograph{\mathcal{G}}$, is obtained by erasing the input premise of every $\pout$ node, and  reversing the orientation of every input arc. 
 We write $\mathcal{G} \models \ACio$ if $\iograph{\mathcal{G}}$ is directed acyclic; 
 we write $\mathcal{G} \models \Cio$ if $\card{\drcl(\mathcal{G})} + \card{\out(\mathcal{G})} = 1$; 
 we write $\mathcal{G} \models \ACCio$ if $\mathcal{G} \models \ACio$ and~$\mathcal{G} \models \Cio$.
\end{definition}

\Cref{lemma:connected-components}.\ref{p:connected-components-io-connectivity} below shows that $\Cio$ is related to connectivity of switching graphs.

\begin{remark}
 \label{rmk:io-graph}
 Let $\mathcal{G}$ be a \pol \gs, and $n \in \nd(\mathcal{G})$. 
 It is easy to check that:
 \begin{enumerate}[i.]
  \item
   $\outdeg_{\iograph{\mathcal{G}}}(n) \geq 2$ if and only if $n$ is a $\pin$ or $\ctpn$ node; \label{itm:io-graph-outdegree}
  \item
   $n$ is a sink in $\iograph{\mathcal{G}}$ if and only if $n$ is a $\bot$, $\wkpn$, classical $\drpn$ or output $\bullet$ node. \label{itm:io-graph-sinks}
 \end{enumerate}
\end{remark}

\begin{proposition}[restate = acyclicity, name = ]
 \label{prop:acyclicity}
 Let $\mathcal{G}$ be a \pol \gs. 
 If $\mathcal{G} \models \AC$, then $\mathcal{G} \models \ACio$.
\end{proposition}



\Cref{prop:acyclicity} holds because, if $\iograph{\mathcal{G}}$ has a directed cycle $\delta$, then, for each $\parr$ or $\ctpn$ node $n$ of $\mathcal{G}$, at most one premise of $n$ is an arc of $\delta$, thus $\delta$ is a cycle of $\mathcal{G}^\varphi$ for some switching $\varphi$~of~$\mathcal{G}$.
The converse of \Cref{prop:acyclicity} is false, as shown in \Cref{subfig:fleury-retore-cycle} (ignore formulae in~gray).\medskip

%
	To build $\iograph{\mathcal{G}}$ from a \pol \gs $\mathcal{G}$, we erase the output premise of every $\pout$ node: the further erasure of a premise of every $\pin$ and $\ctpn$ node yields a switching $\varphi$ of $\mathcal{G}$. 
	\Cref{rmk:switching-connected-components,rmk:io-graph,prop:pfg-sink,prop:acyclicity} then allow us to prove the following result.



\begin{lemma}[restate = connectedComponentsSinks, name = ]
	\label{lemma:connected-components}
	Let $\mathcal{G}$ be a \pol \gs with $\mathcal{G} \models \AC$. 
	Then:
	\begin{enumerate}[i.]
		\item\label{p:connected-components-sinks} for every switching $\varphi$ of~$\mathcal{G}$,
		one has $\card{\cc(\mathcal{G}^\varphi)} = \card{\w(\mathcal{G})} + \card{\drcl(\mathcal{G})} + \card{\out(\mathcal{G})}$;
		\item\label{p:connected-components-io-connectivity} $\mathcal{G} \models \Cw$ iff $\mathcal{G} \models \Cio$ (immediate consequence of \Cref{p:connected-components-sinks}).
	\end{enumerate}
	
\end{lemma}

We now relate $\ACCio$ to sequentiality under a suitable geometric condition (Prop.~\ref{prop:poutter-accio-sequential}).

\begin{definition}
 Let $\mathcal{G}$ be a \pol \gs such that $\mathcal{G} \models \ACio$. 
 A node in $\mathcal{G}$ is \emph{\ioter} if it is minimal with respect to $\precg{\iograph{\mathcal{G}}}$ (\Cref{rmk:directed-acyclic-order}) in the set of terminal nodes~of~$\mathcal{G}$.
\end{definition}

 \begin{lemma}[restate = lemmaIoter, name = ]
  \label{lemma:ioter}
  Let $\mathcal{G}$ be a \pol \gs with no $\pout$ nodes, such that $\mathcal{G} \models \ACio$; let $n$ be a 
  $\cutpn$, $\otimes$ or $\drpn$ node of $\mathcal{G}$ which is terminimal.
  For every output premise $a$ of $n$:
  \begin{enumerate}[i.]
   \item \label{itm:ioter-connected-component}
	if $C$ is the \emph{$a$-split} of $\mathcal{G}$, that is, the connected component of $G$  containing $a$, where $G$ is the graph obtained from $\mathcal{G}$ by adding a fresh $\bullet$ node $n_0$ and setting $\head_G(a) = n_0$,~then~$C$~has no $\drpn$ node, and $n_0$ is the only $\bullet$ node of $C$ whose premise is output;
   \item \label{itm:ioter-directed-path}
    if $\gamma$ is a path of $\mathcal{G}$ to $n$ and the last arc of $\gamma$ is $a$, then $\gamma$ is a directed path of $\iograph{\mathcal{G}}$.
  \end{enumerate}
 \end{lemma}
 
 In both \Cref{lemma:ioter}.\ref{itm:ioter-connected-component} and \Cref{lemma:ioter}.\ref{itm:ioter-directed-path}, the hypothesis that $n$ is \ioter is crucial: neither of the two holds if we take $n$ as the non-\ioter \textcolor{myred}{$\otimes$} node in \Cref{subfig:splitting}.
 
 \begin{definition}
 	\label{def:outp-terminal}
	A \ps $R$ is \emph{\pol} if its underlying \gs $\undgs{R}$ is \pol and $R'$ is a \pol \ps for all $R' \sqsubset R$.
	The set of \pol \pss is denoted by $\psset^\mathnormal{pol}$.
	Given $R \in \psset^\mathnormal{pol}$, $R$ is \emph{\poutter} (resp.~$R \models \ACCio$) if $\undgs{R}$ is \poutter (resp.~$\undgs{R} \models \ACCio$) and $R'$ is \poutter (resp.~$R' \models \ACCio$) \mbox{for each $R' \sqsubset R$}.
 \end{definition}

 \begin{proposition}[restate = poutterAccioSequential, name = ]
  \label{prop:poutter-accio-sequential}
  Let $R \in \psset^\mathnormal{pol}$ be \poutter. 
  If $R \models \ACCio$,  then $R$ is \seq.
 \end{proposition}

 \begin{proof}[Proof (details in \Cref{subsec:complements-untyped-polarized})]
  By induction on $\card{\ar(R)}$. If $\undgs{R}$ has a terminal $\bot$, $\parr$, $\drpn$ or $\wn$ node, 
we easily conclude from the induction hypothesis.
Now assume $\undgs{R}$ has no $\pout$ node (otherwise it would have a terminal $\parr$ node). 
  If every terminal node of $\undgs{R}$ is $\axpn$, $\onepn$ or $\oc$, then $\card{\out(\undgs{R})} = 1$ (as $\undgs{R} \models \Cio$), so $\undgs{R}$ consists of exactly one terminal $\axpn$, $\onepn$ or $\oc$~and~one~concludes.
  
  Otherwise $\undgs{R}$ has a terminal $\cutpn$ or $\otimes$ node.
  Since $R \models \ACio$, there is a \ioter $\cutpn$ or $\otimes$ node $n$ of $\undgs{R}$. 
  Thus $n$ has at least one output premise $a_1$. 
  From $R \models \ACio$ and \Cref{lemma:ioter}.\ref{itm:ioter-directed-path}, we deduce that $n$ appears in no cycle of $\undgs{R}$, so $R$ has one of the (untyped) shapes in \Cref{subfig:desequentialization-cut,subfig:desequentialization-tensor}, depending on the label of $n$. To apply the induction hypothesis and conclude that $R$ is \seq, we prove that $\undgs{R_i} \models \Cio$ for $i \in \{1, 2\}$. 
  Let $R_1, R_2 \in \psset$ such that $\undgs{R_1}$ is the $a$-split of $\undgs{R}$, and $\undgs{R_2}$ is obtained from $\undgs{R}$ by erasing $\undgs{R_1}$ and $n$. 
  As $a$ is output, by  \Cref{lemma:ioter}.\ref{itm:ioter-connected-component} $\drcl(\undgs{R_1}) = \varnothing$ and~$\out(\undgs{R_1}) = \{a\}$, thus $\card{\drcl(\undgs{R_2})} = \card{\drcl(\undgs{R})}$ and  $\card{\out(\undgs{R_2})} = \card{\out(\undgs{R})}$, hence
  $\card{\drcl(\undgs{R_1})} + \card{\out(\undgs{R_1})} = 0 + 1 = 1$ and  
  $\card{\drcl(\undgs{R_2})} + \card{\out(\undgs{R_2})} = \card{\drcl(\undgs{R})} + \card{\out(\undgs{R})} = 1$. 
 \end{proof}
 
  In the proof of \Cref{prop:poutter-accio-sequential}, we need a \emph{\ioter} $\cutpn$ or $\otimes$ node to ``split'' $R$: if we use the non-\ioter \textcolor{myred}{$\otimes$} node in \Cref{subfig:splitting}, then there is only one way to ``split'' (because of connectivity), and neither of the resulting \pss is \seq. 
  The \poutter hypothesis is also crucial: the (non-\seq) \ps~(satisfying $\ACCio$)~in~\Cref{subfig:fleury-retore-intuitionistic}~is ``split'' by its (\ioter) $\otimes$ node in two \pss~not~satisfying~$\ACCio$~(nor~$\ACCw$).
       
  \begin{figure}
   \centering
   \begin{subfigure}{0.5125\textwidth}
    \centering
    \begin{prooftree}
     \hypo{}
     \infer1[\scriptsize\axsc]{\infcstrut \vdash X^\perp, X}
     \hypo{}
     \infer1[\scriptsize\axsc]{\infcstrut \vdash X, X^\perp}
     \infer1[\scriptsize$\bot$]{\infcstrut \vdash X, X^\perp, \bot}
     \infer2[\scriptsize\textcolor{myred}{$\otimes$}]{\infcstrut \vdash X^\perp, X \otimes \bot, X, X^\bot}
     \infer1[\scriptsize$\parr$]{\infcstrut \vdash X, X \otimes \bot, X^\bot \parr X^\bot}
     \hypo{}
     \infer1[\scriptsize$\one$]{\infcstrut \vdash \one}
     \infer2[\scriptsize\textcolor{mygreen}{$\otimes$}]{\infcstrut \vdash X, X \otimes \bot, (X^\bot \parr X^\bot) \otimes \one}
    \end{prooftree}
    \caption{\label{subfig:sequent-calculus-proof}A \seqcp $\pi$ (the \exch rule is implicit).}
   \end{subfigure}
   \quad
   \begin{subfigure}{0.3875\textwidth}
   	\vspace*{-\baselineskip}
    \centering
    \input{Figures/necessity-io-terminal.tikz}
    \caption{\label{subfig:splitting}The desequentialization $R = \Deseq{\pi}$ of $\pi$.}
   \end{subfigure}
   \caption{A \poutter \ps having a \ioter \textcolor{mygreen}{$\otimes$} node and a non-\ioter~\textcolor{myred}{$\otimes$}~node.}
  \end{figure}
 
 \begin{theorem}
  \label{thm:poutter-accio-sequential}
  Let $R \in \ppset$ be \poutter: $R \models \ACCw$ iff $R \models \ACCio$ iff $R$ is sequential.
 \end{theorem}
 
 \begin{proof}
  Immediate consequence of \Cref{prop:accw-necessary,prop:acyclicity,prop:poutter-accio-sequential}, and \Cref{lemma:connected-components}.\ref{p:connected-components-io-connectivity}.
 \end{proof}
 
 \begin{remark}
  \label{rmk:complexity}
Checking whether or not $R \models \ACCw$~holds is (obviously) at most exponential, as $\AC$ requires to check a number of (switching) graphs that is exponential in the number of $\parr$ and $\ctpn$ nodes. On the other hand, checking $\ACCio$ is at most linear in the number of nodes and arcs of a \pol \ps $R$ (\cite[Section 4.2]{sedgewick2011algorithms}). By \Cref{thm:poutter-accio-sequential}, for a \poutter~$R$, we have $R \models \ACCw$ if and only if $R \models \ACCio$, so checking $\ACCw$~becomes~linear.
 \end{remark}
 
 \begin{remark}
  \label{rmk:PolNoStableCutElim}
$\ppset$ is not closed wrt cut elimination (\Cref{fig:cutelim}): by applying the $\tor{\rdr}$ step in \Cref{subfig:cutelim-rdr} to $R\in\ppset$ one obtains an untyped \ps $R'$ inheriting from $R$ a labelling of its arcs that (in general) does not respect the constraints of \Cref{fig:polarized}. This is in sharp contrast with the polarization present in \cite{danos1990logique,regnier1992lambda,DanosRegnier95,Accattoli15} and it is the reason why we tend to think at polarization as a global geometric constraint allowing us to apply \Cref{thm:poutter-accio-sequential}~rather than think at $\ppset$ as a subsystem of \pss with its~own~logical~dignity.
 \end{remark}

%% file: Sections/Typed.tex
We isolate in \MELL the notable fragment \VMELL (\Cref{def:vmell}), containing both classically and intuitionistically polarized formulae. The type discipline imposed by \formulae{\VMELL} guarantees that any \pst{\VMELL} (without axiom cuts and 
$\rtp$ cuts) satisfies the hypothesis of \Cref{thm:poutter-accio-sequential} (it is \poutter), thus when it satisfies $\ACCw$ it is sequentializable (\Cref{cor:vmell}) and by \Cref{rmk:complexity} this can be checked in linear time in the size of $R$. Contrary to polarization of \Cref{sec:pure-polarized} and as usual in the typed setting, the set of \psst{\VMELL} is closed wrt \cutelim (\Cref{rmk:CutElimPrevervesTypes}) and $\ACCw$ is thus a correctness~criterion~for~\VMELL.

\subparagraph*{Formulae and \texorpdfstring{\seqc}{sequent calculus}}

Given countably infinite \emph{atomic formulae}, denoted by $X$, \emph{\formulae{\MELL}} are defined by the grammar below:
 \[
  A \Coloneqq X \mid \one \mid \bot \mid A \otimes A \mid A \parr A \mid \oc A \mid \wn A
 \]
 \emph{Linear negation} on atomic formulae is an involution $(\cdot)^\perp$ without fixpoints, extended inductively to \formulae{\MELL} by: $\one^\bot \coloneq \bot$, $(A \otimes B)^\bot \coloneq A^\bot \parr B^\bot$\!, $(\oc A)^\bot \coloneq \wn A^\bot$ and $(A^\bot)^\bot \coloneqq A$.
 
 \begin{definition}
  A \emph{fragment} $\mathcal{F}$ of \MELL is a set of \formulae{\MELL} closed under sub-formulae (for every $A \in \mathcal{F}$ and every sub-formula $B$ of $A$, $B \in \mathcal{F}$) and linear~negation (for every $A \in \mathcal{F}$, $A^\perp \in \mathcal{F}$). If $\mathcal{F}$ and $\mathcal{F}'$ are fragments of \MELL, we say that $\mathcal{F}$~is~a~\emph{fragment} of $\mathcal{F}'$ when $\mathcal{F} \subseteq \mathcal{F}'$. 
  A \emph{sequent} $\Gamma$ of $\mathcal{F}$ is a finite sequence of \formulae{$\mathcal{F}$}. 
  A \emph{\seqcp} in $\mathcal{F}$ with \emph{conclusion} $\Gamma$ is a finite tree of \sequents{$\mathcal{F}$} with root $\Gamma$ built via the~rules~in~\Cref{fig:sequent-calculus-rules}.
     
  \begin{figure}
  	\vspace*{-\baselineskip}
   \centering
   \begin{prooftree}
    \hypo{\infcstrut}
    \infer1[\scriptsize\axsc{}]{\infcstrut \vdash A,A^\bot}
   \end{prooftree}
   \quad
   \begin{prooftree}
    \hypo{\infcstrut \vdash \Gamma, A}
    \hypo{\infcstrut \vdash  A^\bot, \Delta}
    \infer2[\scriptsize\cutsc{}]{\infcstrut \vdash \Gamma, \Delta}
   \end{prooftree}
   \quad
   \begin{prooftree}
    \hypo{\infcstrut}
    \infer1[\scriptsize$\one$]{\infcstrut \vdash \one}
   \end{prooftree}
   \quad
   \begin{prooftree}
    \hypo{\infcstrut \vdash \Gamma}
    \infer1[\scriptsize$\bot$]{\infcstrut \vdash \Gamma, \bot}
   \end{prooftree}
   \quad
   \begin{prooftree}
    \hypo{\infcstrut \vdash \Gamma, A}
    \hypo{\infcstrut \vdash B, \Delta}
    \infer2[\scriptsize$\otimes$]{\infcstrut \vdash \Gamma, A \otimes B, \Delta}
   \end{prooftree}
   \vskip .5em
   \begin{prooftree}
    \hypo{\infcstrut \vdash \Gamma, A, B}
    \infer1[\scriptsize$\parr$]{\infcstrut \vdash \Gamma, A \parr B}
   \end{prooftree}
   \quad
   \begin{prooftree}
    \hypo{\infcstrut \vdash \wn\Gamma, A}
    \infer1[\scriptsize$\oc$]{\infcstrut \vdash \wn \Gamma, \oc A}
   \end{prooftree}
   \quad
   \begin{prooftree}
    \hypo{\infcstrut \vdash \Gamma, A}
    \infer1[\scriptsize\drsc{}]{\infcstrut \vdash \Gamma, \wn A}
   \end{prooftree}
   \quad
   \begin{prooftree}
    \hypo{\infcstrut \vdash \Gamma, \wn A, \dots, \wn A}
    \infer1[\scriptsize$\wn$]{\infcstrut \vdash \Gamma, \wn A}
   \end{prooftree}
   \quad
   \begin{prooftree}
    \hypo{\infcstrut \vdash \Gamma, A, B, \Delta}
    \infer1[\scriptsize\exch{}]{\infcstrut \vdash \Gamma, B, A, \Delta}
   \end{prooftree}
   \caption{\label{fig:sequent-calculus-rules}\Seqc rules of \MELL.}
  \end{figure}
 \end{definition}
 
 
 \begin{definition}
  \label{def:ground-structure-mell}
  Let $\mathcal{F}$ be a fragment of \MELL. A \emph{\gst{$\mathcal{F}$}} is a \gs with arcs labeled by \formulae{$\mathcal{F}$} (called \emph{types}) according to the local constraints~in~\Cref{fig:typed}.
  
  \begin{figure}
  	\vspace*{-\baselineskip}
   \centering
   \input{Figures/typed-axiom.tikz}
   \hskip 0.5em
   \input{Figures/typed-cut.tikz}
   \hskip 0.5em
   \input{Figures/typed-one.tikz}
   \hskip 0.5em
   \input{Figures/typed-bottom.tikz}
   \hskip 0.5em
   \input{Figures/typed-tensor.tikz}
   \hskip 0.5em
   \input{Figures/typed-par.tikz}
   \hskip 0.5em
   \input{Figures/typed-bang.tikz}
   \hskip 0.5em
   \input{Figures/typed-dereliction.tikz}
   \hskip 0.5em
   \input{Figures/typed-why-not.tikz}
   \caption{\label{fig:typed}Local typing constraints for a \gst{\MELL}.}
  \end{figure}
  

	A \emph{\pst{$\mathcal{F}$}} is a \ps $R$ such that its underlying \gs $\undgs{R}$ is a \gst{$\mathcal{F}$} and, for each $\oc$ node $n$ of $\undgs{R}$ with (ordered) conclusions of types $\wn A_1, \dots, \wn A_k, \oc A$, the \ps 
	$\undbox{R}(n)$ is a \pst{$\mathcal{F}$} with (ordered) conclusions of types $\wn A_1, \dots, \wn A_k, A$. 
	The set of \psst{$\mathcal{F}$} is denoted by $\psf{\mathcal{F}}$.
 \end{definition}
 
 \begin{remark}
  \label{rmk:CutElimPrevervesTypes}
If $R \in \psf{\mathcal{F}}$ has conclusion $\Gamma$ and $R\tocut R'$, then $R'\in\psf{\mathcal{F}}$ with conclusion $\Gamma$.
 \end{remark}

 The well-known desequentialization function maps a \seqcptWithNameAndConclusion{$\pi$}{$\mathcal{F}$}{$\Gamma = A_1, \dots, A_k$} to $\Deseq{\pi} \in \psf{\mathcal{F}}$ with (ordered) conclusions of types $A_1, \dots, A_k$. 
 The function $\pi \mapsto \Deseq{\pi}$ is defined by induction on $\pi$ (\Cref{fig:desequentialization}; \Cref{subsec:desequentialization} and \mbox{\cite{llhandbook,girard1987linear}}~for~details).
 
 \begin{definition}
 	\label{def:proof-net}
  Let $\mathcal{F}$ be a fragment of \MELL and $R \in \psf{\mathcal{F}}$. We say that $R$ is \emph{sequentializable} or a \emph{\pn} when~there exists a \seqcptWithName{$\pi$}{$\mathcal{F}$} such that $R = \Deseq{\pi}$.
 \end{definition}
 
 
 For example, $R = \Deseq{\pi}$ for the \seqcp $\pi$ and $R \in \psset$ in \Cref{subfig:sequent-calculus-proof,subfig:splitting}.
 
  \begin{remark}
 	\label{rmk:sequential}
 	An 
 	$R \in \psf{\mathcal{F}}$ is \seq (Def.~\ref{def:sequential} and Fig.~\ref{fig:desequentialization}, taking formulae into account) iff it is sequentializable.
 	The right-to-left direction is proved by induction on any 
 	\seqcptWithName{$\pi$}{$\mathcal{F}$} such that $\Deseq{\pi} = R$.
 	The converse is proven by induction on $\card{\ar{(R)}}$.
 \end{remark}
 
\subparagraph*{Some typed polarized fragments}

	 We recall the classical \cite{laurent1999polarized,laurent2003polarized} and intuitionistic \cite{Lamarche95,lamarche2008proof} polarizations 
	 of 
	 \LL without additives. 
	 To provide their standard presentations, we split the set of atomic formulae into two disjoint subsets 
	 so that $X^\perp$ has the opposite~polarity~of~$X$.
 
	The fragment $\LLpol$ of \MELL is defined by the grammar below:
	\begin{align*}
		N & \Coloneqq X \mid \bot \mid N \parr N \mid \wn P & \text{(\emph{negatives})} && P & \Coloneqq X^\bot \mid \one \mid P \otimes P \mid \oc N & \text{(\emph{positives})}
	\end{align*}
  	The fragment \IMELL of \MELL is defined by the grammar below of \emph{output} ($O$) and \emph{input} ($I$):
  	\begin{align*}
  		O & \Coloneqq X \mid \one \mid O \otimes O \mid O \parr I \mid I \parr O \mid \oc O & I & \Coloneqq X^\bot \mid \bot \mid I \parr I \mid I \otimes O \mid O \otimes I \mid \wn I
  	\end{align*}
 
 Thus, in $\LLpol$ $N^\bot$ is positive, $P^\bot$ is negative; in \IMELL $O^\bot$ is input, $I^\bot$ is output. 
 Note that $ \psf{\LLpol} \cup  \psf{\IMELL} \subsetneq \ppset \subsetneq \psset$.
 \smallskip
 
 In $\LLpol$ there is no $\pout$ node, a \ps is \poutter, and $\ACCio$ becomes the correctness criterion studied by Laurent \cite{laurent2003polarized}. 
 We can thus instantiate \Cref{thm:poutter-accio-sequential} as~below.
 
 \begin{corollary}
  \label{cor:classical-polarization}
  If $R \in \psf{\LLpol}$, then: $R \models \ACCw$ iff $R \models \ACCio$ iff $R$ is sequentializable.
 \end{corollary}
 
 
 In \IMELL, every provable sequent has exactly one output formula \cite{Lamarche95,lamarche2008proof} and this property is characterized geometrically by $\ACCw$ (\Cref{cor:imell} below, together 
 with \Cref{prop:accw-necessary}).

 
 \begin{remark}
  \label{rmk:imell-output-conclusion}
	If $R \in \psf{\IMELL}$, then  
	$R$ has no classical $\drpn$ nodes (because the premise of a $\drpn$ node can only be input by type constraints) and $\card{\out(\undgs{R'})} = 1$ for every $R' \sqsubset R$ (because of the type constraint on the \pport).
	From this and \Cref{lemma:connected-components}.\ref{p:connected-components-io-connectivity} we deduce \Cref{cor:imell}. 
 \end{remark}
 
 \begin{corollary}
  \label{cor:imell}
  If $R \in \psf{\IMELL}$, then: $R \models \ACCw$ if and only if $R \models \AC$ and $\card{\out(\undgs{R})} = 1$.
 \end{corollary}
 
 
  The example in \Cref{subfig:fleury-retore-intuitionistic} shows that the converse of \Cref{prop:accw-necessary} fails even if we~restrict to \psst{\IMELL}. Hence, there is no analogue of \Cref{cor:classical-polarization}~for~\IMELL.
 
\smallskip
 We now introduce \VMELL, a new fragment of \MELL with remarkable properties. 
 
 \begin{definition}
  \label{def:vmell}
  The fragment \VMELL of \MELL is defined by the following grammar:
  \begin{align*}
   O & \Coloneqq P \mid O \parr N \mid N \parr O & \text{(\emph{outputs})} && P & \Coloneqq X \mid \one \mid P \otimes P \mid \oc O \mid \oc N & \text{(\emph{positives})} \\[-.1em]
   I & \Coloneqq N \mid I \otimes P \mid P \otimes I & \text{(\emph{inputs})} && N & \Coloneqq X \mid \bot \mid N \parr N \mid \wn I \mid \wn P & \text{(\emph{negatives})}
  \end{align*}
 \end{definition}
 
 \begin{remark}
  \label{rmk:vmell-semantics}
  Every atomic \formula{\VMELL} is both positive and negative. By 
  fixing a polarity for every atom (with $X$ and $X^\bot$ having opposite polarities), the intersection of \VMELL and \IMELL is essentially the counterpart in \MELL of Ehrhard's half-polarized typing system for call-by-push-value \cite{ehrhard16cbpv,ehrhard19probabilistic}, and \VMELL inherits the semantic interpretation of~that~system, where positive formulae carry a structure of $\oc$-coalgebra: they are ``equipped~with structural~rules''.
 \end{remark}
 
 \begin{remark}
  \label{rmk:vmell}
  $\LLpol$ is a fragment of \VMELL. Proofs of \IMELL are embedded in \VMELL~via the translation (defined via linear negation on the input formulae,~details~in~\Cref{subsec:embedding-imell-vmell}):
  \begin{align*}
   \VMELLt{X} & \coloneqq \oc X & \VMELLt{(O_1 \otimes O_2)} & \coloneqq \oc (\VMELLt{O_1} \otimes \VMELLt{O_2}) & \VMELLt{(I \parr O)} & \coloneqq \oc (\VMELLt{I} \parr \VMELLt{O}) \\[-0.1em]
   \VMELLt{\one} & \coloneqq \oc \one & \VMELLt{(O \parr I)} & \coloneqq \oc (\VMELLt{O} \parr \VMELLt{I}) & \VMELLt{(\oc O)} & \coloneqq \VMELLt{O}
  \end{align*}
 \end{remark}
 
 Note that $\psf{\VMELL} \subsetneq \ppset$.
 In the absence of axiom cuts and 
 $\rtp$ cuts, 
 every \pst{\VMELL} is \poutter. 
 So, by \Cref{thm:poutter-accio-sequential} (and \Cref{thm:accw-cut-elim} for \Cref{itm:normal-form-sequentializable}):
 
 \begin{corollary}
  \label{cor:vmell}
  Let $R \in \psf{\VMELL}$.
  \begin{enumerate}[i.]
   \item \label{itm:vmell}
    If $R$ is \rn{\rmu}, then we have that: $R \models \ACCw$ iff $R \models \ACCio$ iff $R$ is sequentializable.
   \item \label{itm:normal-form-sequentializable}
    If $R \models \ACCw$, then the\footnotemark
    \footnotetext{
    	In $\psset$, reduction $\tor{\rmu}$ is strongly normalizing and confluent \cite{llhandbook}, hence the $\rmu$-normal form is unique.}  
    \rnf{\rmu} of $R$ is sequentializable.
  \end{enumerate}
 \end{corollary}

  The \rn{\rmu} hypothesis is crucial in \Cref{cor:vmell}.\ref{itm:vmell}: omitting it~makes the converse of \Cref{prop:accw-necessary} fail (\Cref{subfig:fleury-retore-cut}). 
  Also, the \rnf{\rmu} of $R \in \psf{\VMELL}$ might be non-sequentializable if $\ACCw$ is replaced by $\ACCio$ in \Cref{cor:vmell}.\ref{itm:normal-form-sequentializable} (\Cref{subfig:fleury-retore-cycle}). Indeed, the~same example shows that $\ACio$ (and thus $\ACCio$) is not stable under \cutelim (in \VMELL). Hence, contrary to what happens in $\LLpol$ (\Cref{cor:classical-polarization}), the properties~$\ACCw$~and~$\ACCio$ are not equivalent in \VMELL in the presence of axiom cuts or 
  $\rtp$ cuts.
 
 \begin{remark}
 	\label{rmk:vmell-normal}
 By \Cref{cor:vmell}.\ref{itm:vmell}, $\ACCw$ is a correctness criterion for \VMELL \cf \pss, and 
 this is already an interesting result in the theory of \pns, similar to \cite{DBLP:journals/apal/AbrusciR99} in the non-commutative case and to \cite{laurent2004slicing} in presence of the additives. 
 \Cref{cor:vmell}.\ref{itm:vmell} is a stronger result because it applies to \rn{\rmu} \pss, possibly with cuts. Furthermore, by \Cref{rmk:complexity}, for every \rn{\rmu} $R \in \psf{\VMELL}$, checking if $R$ is sequentializable is linear.

 \end{remark}
 

%% file: Sections/Bang.tex
%
%
%
	

We now introduce the fragments $\bangMELL$, $\CbN$, $\CbV$, of computational interest as they correspond to the \bc and the call-by-name and call-by-value \lci, respectively.
%
	The three fragments of \MELL are defined by the grammars below.
	\begin{align*}
		\bangMELL &&&& O & \Coloneqq X \mid \wn I \parr O \mid \oc O & \text{(\emph{outputs})} &&&& I & \Coloneqq X^\bot \mid \oc O \otimes I \mid \wn I & \text{(\emph{inputs})}
		\\[-,2em]
		\CbN &&&& O & \Coloneqq X \mid \wn I \parr O & \text{(\emph{outputs})} &&&& I & \Coloneqq X^\bot \mid \oc O \otimes I  & \text{(\emph{inputs})}
		\\[-,2em]
		\CbV &&&& O & \Coloneqq X \mid \wn I \parr \oc O & \text{(\emph{outputs})} &&&& I & \Coloneqq X^\bot \mid \oc O \otimes \wn I & \text{(\emph{inputs})}
	\end{align*}
	
	We set $\linArrow{O}{O'} \coloneq O^\bot \parr O'$ for every output \formulae{\bangMELL} $O, O'$.

One has $\CbN \cup \CbV \!\subsetneq \bangMELL \subsetneq \VMELL \cap \IMELL$, with $\CbN \!\subsetneq \LLpol \cap \IMELL$ (in that an output in $\CbN$ is both output in \IMELL and negative in $\LLpol$, and dually on inputs/positives), but $\CbV \!\not\subseteq \LLpol$. 
A provable \sequent{\bangMELL} (resp. $\CbN$\!, $\CbV$) has exactly one output \formula{\bangMELL (resp. $\CbN$\!, $\CbV$)}.
The negation of an output 
is input, and vice versa, for each fragment. 
We shall see how to relate those~fragments~to~their~calculi.

We also define the set \STLC of \emph{simple types} by the grammar: $\sigma, \tau \Coloneqq X \mid \tau \Rightarrow \sigma$.
	
	\subparagraph*{The simply typed bang and \lci}
	Given a countably infinite set $\vset$ of variables ranging
	over $x,y,z,\dots$, the set $\tbcset$ of \emph{\bterms} of the simply typed \bc and the set
	$\tlset$ of \emph{\lterms} of the simply typed $\lambda$-calculus (in Church-style presentation) are defined by:
	\begin{align*}
	\tbcset &&&\ni& s,t &\Coloneqq x \mid \bcabs{x}{\oc O}{t} \mid \bcapp{s}{t} \mid \bcder{t} \mid \bcbang{t} & \textup{(where $O$ is any \formula{\bangMELL})}\\
	\tlset &&&\ni & M, N &\Coloneqq x \mid \labs{x}{\tau}{M} \mid \lapp{M}{N}  & \textup{(where $\tau$ is any \formula{\STLC})}
	\end{align*}
	A \emph{term} is either a \bterm or a \lt.
	The only binder is $\lambda x$, and terms are identified up to $\alpha$-equivalence. 
	We write $\fv{t}$ for the set of free variables of the term $t$, and  $t\bcsub{s}{x}$ for the term obtained from $t$ by the capture-avoiding substitution of $s$ for the free occurrences~of~$x$~in~$t$. 
	
	Reduction in the \emph{simply typed \bc} is
	$\mathop{\tor{\br}} \coloneqq \mathop{\tor{\betab}} \cup \tor{\brder}$
	where $\tor{\betab}$ (resp.~$\tor{\brder}$) is the contextual closure of the rule $\bcapp{\bcabs{x}{\oc O}{t}}{\bcbang{s}} \mapstor{\betab} t \bcsub{s}{x}$ (resp.~$\bcder{\bcbang{t}} \mapstor{\brder} t$).
	The $\oc$ construction on \bterms reflects the $\oc$ modality in linear logic
	denoting replicable resources in the syntax:~in~the $\betab$-rule, the
	argument of the function must be replicable to allow the $\betab$-redex to~be~fired.
	
	Reduction in the \emph{call-by-name} (\emph{\CbNtext}) \emph{simply typed $\lambda$-calculus} is
	$\mathop{\tob}$, the contextual closure of the rule $\lapp{\labs{x}{\tau}{M}}{N} \mapstor{\beta} M \lsub{N}{x}$.
	Reduction in the \emph{call-by-value} (\emph{\CbVtext}) \emph{simply typed $\lambda$-calculus} is
	$\mathop{\tobv}$, the contextual closure of the rule $\lapp{\labs{x}{\tau}{M}}{V} \mapstor{\betav} M \lsub{V}{x}$, where $V$ is a \emph{value}, i.e.,  variable or abstraction.
	We refer to \CbNtext and \CbVtext not as evaluation~strategies~but~as calculi, with their own equational theories generated by $\tob$ and $\tobv$ respectively.
	
	In the simply typed \bc, a \emph{\bcontext}\index{context} is a partial function with finite domain from variables to \formulaeWithName{of the form $\oc O$}{\bangMELL}. 
	A \bcontext $\Gamma$ is denoted by $\oftype{x_1\!}{\!O_1}, \dots, \oftype{x_n\!}{\!O_n}$ if the domain of $\Gamma$ is $\dom(\Gamma) = \{x_1, \dots, x_n\}$ and $x_i \neq x_j$ and $\Gamma(x_i) = O_i$ for all $i, j \in \{1, \dots, n\}$ with $i \neq j$. 
	A \emph{\bjudgm} is a triple $\ljud{\Gamma}{t}{O}$, where $\Gamma$ is a \bcontext, $t \in \tbcset$ and $O$ is an output \formula{\bangMELL}. 
	A \emph{\bderivation} 
	with \emph{conclusion} $\bjud{\Gamma}{t}{O}$, denoted by $\bderive{\pi}{\Gamma}{t}{O}$,~is~a finite tree of \bjudgms with root $\bjud{\Gamma}{t}{O}$ built up from the rules in \Cref{fig:derivations-bang-judgements} (on top).
	A \bcontext $\Gamma$ and $t \in \tbcset$ form a \emph{typed \bterm} $(\Gamma, t)$ if $\bderive{\pi}{\Gamma}{t}{O}$ for some $\pi$, $O$.

	\begin{remark}\label{rmk:unique-type}
	\begin{enumerate}[i.]
		\item \label{p:unique-type}
		(Uniqueness) If $\bderive{\pi}{\Gamma}{t}{O}$ and $\bderive{\pi'}{\Gamma}{t}{O'}$, then $O = O'$ and $\pi = \pi'$\!.
		
		\item (Free variables) If $\bderive{\pi}{\Gamma}{t}{O}$ then $\fv{t} \subseteq \dom{(\Gamma)}$.
		
		\item (Subject reduction) If $t \tor{\br} s$ and $\bderive{\pi}{\Gamma}{t}{O}$, then there is $\bderive{\pi'}{\Gamma}{s}{O}$.
	\end{enumerate}
	\end{remark}
	
	\begin{figure}[!tb]
		\vspace*{-\baselineskip}
		\small
		\centering
		\begin{prooftree}[label separation=0.2em]
			\hypo{\infcstrut}
			\infer1[\scriptsize$\vset$]{\infcstrut \Gamma, x : \oc O \vdash x : O}
		\end{prooftree} \
		\begin{prooftree}[label separation=0.2em]
			\hypo
			{\infcstrut \Gamma, x : \oc O \vdash t : O'}
			\infer1[\scriptsize\btl{}]{\infcstrut \Gamma \vdash \bcabs{x}{\oc O}{t} : \oc O \multimap O'}
		\end{prooftree} \
		\begin{prooftree}[label separation=0.2em, separation=.8em]
			\hypo
			{\infcstrut \Gamma \vdash t : \oc O \multimap O'}
			\hypo
			{\infcstrut\Gamma \vdash s : \oc O}
			\infer2[\scriptsize\bta{}]{\infcstrut \Gamma \vdash \langle t \rangle s : O'}
		\end{prooftree} \
		\begin{prooftree}[label separation=0.2em]
			\hypo
			{\infcstrut \Gamma \vdash t : O}
			\infer1[\scriptsize$\oc$]{\infcstrut \Gamma \vdash t^\oc : \oc O}
		\end{prooftree} \
		\begin{prooftree}[label separation=0.2em]
			\hypo
			{\infcstrut \Gamma \vdash t : \oc O}
			\infer1[\scriptsize$\brder$]{\infcstrut \Gamma \vdash \bcder{t} : O}
		\end{prooftree}
	\\[0.4em]
	\begin{prooftree}[label separation=0.2em]
		\hypo{\infcstrut}
		\infer1[\scriptsize$\vset$]{\infcstrut \Gamma, x : \tau \vdash x : \tau}
	\end{prooftree} \quad
	\begin{prooftree}[label separation=0.2em]
		\hypo
		{\infcstrut \Gamma, x : \tau \vdash M : \sigma}
		\infer1[\scriptsize\btl{}]{\infcstrut \Gamma \vdash \bcabs{x}{\tau}{M} : \tau \Rightarrow \sigma}
	\end{prooftree} \quad
	\begin{prooftree}[label separation=0.2em, separation=.8em]
		\hypo
		{\infcstrut \Gamma \vdash M : \tau \Rightarrow \sigma}
		\hypo
		{\infcstrut\Gamma \vdash N : \tau}
		\infer2[\scriptsize\bta{}]{\infcstrut \Gamma \vdash \lapp{M}{N} : \sigma}
	\end{prooftree}
		\caption{\label{fig:derivations-bang-judgements}
			\label{fig:derivations-simple-judgements}
			Rules of the simply typed \bc (on top) and simply typed $\lambda$-calculus (bottom).}
	\end{figure}
	
The definitions of \emph{\lcontext}, \emph{\ljudgm}, \emph{\lderivation}, \emph{typed \lterm} for the simply typed \lci are easily adapted from the corresponding ones for the simply typed \bc: 
the image of a \lcontext are \formulae{\STLC}; a \ljudgm is made of a \lcontext, a \lterm and a \formula{STLC}; a \lderivation is built from the rules in \Cref{fig:derivations-simple-judgements} (bottom). 
The analogue of \Cref{rmk:unique-type} holds in the simply typed \lci \cite{barendregt1992lambda}.
Note that the \CbNtext and~\CbVtext~simply~typed \lci only differ in their reductions, not in their terms or typing~system.
	
\subparagraph*{\Bpns}
We translate typed \bterms into particular \pnst{\bangMELL}, called~\bpns (\Cref{def:bang-pns}).
First, we refine the way we handle \pss.

Akin to \cite{DiCosmoKesner97,Tranquilli11}, 
$\town$ is the reduction on $\psset$ generated by any of the root steps in \Cref{fig:conv} (note that they do not refer to formulae).
It is easy to check that $\town$ is strongly normalizing and confluent: 
{the} $\wn$-normal form of a \ps is the canonical form with respect to associativity of $\ctpn$ nodes, neutrality of $\wkpn$ nodes 
and their positions relatively to box borders.

	\begin{figure}
		\vspace*{-\baselineskip}
		\centering
		\begin{subfigure}{0.2375\textwidth}
			\centering
			\input{Figures/why-not-flattening.tikz}
			\caption{Flattening consecutive $\wn$ nodes.}
		\end{subfigure}
        \quad
		\begin{subfigure}{0.1375\textwidth}
			\centering
			\input{Figures/why-not-unary.tikz}
			\caption{\label{fig:unary-conversion}Erasing \unpn node (possibly $n = n'$).}
		\end{subfigure}
        \quad
		\begin{subfigure}{0.525\textwidth}
			\centering
			\input{Figures/why-not-box.tikz}
%
			\caption{\label{subfig:pull}Pulling a $\wn$ node out of a box.}
		\end{subfigure}
	\caption{\label{fig:conv}Root steps of 
		reduction $\town \mathop{\coloneqq} \tor{\rfl} \cup \tor{\rer} \cup \tor{\rpu}$ on $\psset$.
		We define 
		$\tor{\rmuw} \mathop{\coloneq} \tor{\rmu} \cup \town$,
		${\tor{\rew}} \mathop{\coloneq} {\tor{\re}} \cup {\town}$, 
		$\tor{\ree} \mathop{\coloneqq} \tor{\re} \cup \tor{\rdr}$,
		${\tocutwn} \mathop{\coloneqq} {\tocut} \cup {\town}$ 
		($
		\tor{\rmu}, \tor{\re}, \tor{\rdr}, \tor{\cutpn}$ are defined in~Fig.~\ref{fig:cutelim}).
		}
	\end{figure}

\begin{definition}
	A \pntWithName{$R$}{\bangMELL} (see \Cref{def:proof-net}) is:
	\begin{itemize}
		\item \emph{\nam}\index{named@\nam} (resp.~\emph{\injnam}\index{injectively named@\injnam}) if it is equipped with a function (resp.~injective function) $\nmfun{R}$\index{νR@$\nmfun{R}$} from the set of input conclusions of ${R}$ to $\vset$;
		\item \emph{\gua}\index{guarded@\gua} if every input premise of a $\parr$, $\wn$ or $\bullet$ node 
		is a conclusion of a $\oc$, $\drpn$ or~$\wn$~node;
		\item \emph{bang}\index{bang proof-net@\bpn} if it is \gua and \injnam;
		the set of \bpns is denoted by $\bpnset$\index{Rb@$\bpnset$}.
	\end{itemize} 
\end{definition}

In figures, the variable $\nmfun{R}(a)$ is omitted or written below the input conclusion $a$.
We also assume that the set $\vset$ of variables is totally ordered; 
the order of the input conclusions of a \nam \pn $R$ 
is induced by their image via $\nmfun{R}$.
It is easy to check that guardedness~and (injective) naming are preserved under~$\tocutwn$.

\begin{remark}
	\label{rmk:bangMELL-output-conclusion}
	In  a \pntWithName{$R$}{\bangMELL} 
	$\card{\out(\undgs{R})} = 1$ by Cor.~\ref{cor:imell}, Prop.~\ref{prop:accw-necessary} and Rmk.~\ref{rmk:sequential} as $\bangMELL \subseteq \IMELL$.
	We assume the output conclusion of $R$ is maximal in the total order of the conclusions of $R$.
	If 
	$R$ is \gua, then the type of any of its input conclusions has shape~$\wn I$.
\end{remark}

By Remarks \ref{rmk:unique-type}.\ref{p:unique-type} and \ref{rmk:bangMELL-output-conclusion}, it is natural to translate \emph{typed} \bterms into \emph{bang} \pns.

 \begin{definition}
	\label{def:bang-pns}
	With every typed \bterm $(\Gamma, t)$, where $\Gamma = \oftype{x_1}{\oc O_1}, \dots, \oftype{x_k}{\oc O_k}$, we associate $\WH{(\Gamma, t)} \in \bpnset$ with exactly $k + 1$ conclusions $a_1, \dots, a_{k+1}$ according to \Cref{fig:bang-pns}, such that:
	\begin{itemize}
		\item $a_{k+1}$ is the greatest conclusion, its type is the only derivable output formula $O$ for $(\Gamma, t)$;
		\item for every $i \in \{1, \dots, k\}$, the type of $a_i$ is $\wn O_i^\bot$ and $\nmfun{R}(a_i) = x_i$.
	\end{itemize} 
\end{definition}

The definition of $\WH{(\Gamma, t)}$ is by induction on $t \in \tbcset$: it is easy to check that $\WH{(\Gamma, t)}$ is actually a \bpn (and so there is a \seqcptWithNameAndConclusion{$\pi$}{\bangMELL}{$\wn O_1^\bot, \dots, \wn O_k^\bot, O$}); 
the type $O$ of its output conclusion is uniquely determined by~\Cref{rmk:unique-type}.

\begin{figure}[!th]
\vspace*{-1.5\baselineskip}
 \centering
 \subcaptionbox{\label{fig:bang-pns-app}$\WH{(\Gamma, \bcapp{t_1}{t_2})}$ is the \rnf{\rwn} of the \pn in the picture.}[25.125em][c]{
  \centering
  \input{Figures/bang-pn-application.tikz}
 }
 \quad
 \subcaptionbox{$\WH{(\Gamma, \bcbang{s})}$ is the \rnf{\wn} of the \pn in the picture.}[12em][c]{
 	\centering
 	\input{Figures/bang-pn-bang.tikz}
 }
 \vskip -.7em
 \subcaptionbox{\label{fig:bang-pns-var}$\WH{((\Gamma, \oftype{x}{\oc O}), x)}$.}{
 	\centering
 	\input{Figures/bang-pn-variable.tikz}
 }
 \qquad
 \subcaptionbox{$\WH{(\Gamma, \bcabs{x}{\oc O}{s})}$.}{
 	\centering
 	\input{Figures/bang-pn-abstraction.tikz}
 }
 \qquad
 \subcaptionbox{$\WH{(\Gamma, \bcder{s})}$.}{
 	\centering
 	\input{Figures/bang-pn-dereliction.tikz}
 }
 \caption{\label{fig:bang-pns}The translation $\WH{(\Gamma,t)}$ of the typed \bterm $(\Gamma,t)$ of the \bc into \pns.}
\end{figure}

\begin{remark}
	\label{rmk:wn-normal}
	For every typed \bterm $(\Gamma,t)$, its translation $\WH{(\Gamma,t)}$ is \rn{\rew}.
	Moreover, $t$ is \rn{\betab} (resp.~\rn{\brder}) iff $\WH{(\Gamma,t)}$ is \rn{\rtp} (resp.~\rn{\rdr}).
	However, $\br$-normality (in $\tbcset$) does not imply $\cutwn$-normality (in $\bpnset$) because of $\tor{\rax}$: take $t = \bcapp{y}{x}$.
\end{remark}

From a \emph{dynamic} point of view, reduction $\tor{\br}$ in the \bc can be simulated by reduction $\tor{\cutwn}$ in the bang \pnst{\bangMELL}, via the translation $\WH{(\cdot)}$ (\Cref{thm:simulation}). 



\begin{lemma}
	\label{lemma:substitution}
	Let $((\Gamma, \oftype{x}{\oc O}),t), (\Gamma, \bcbang{s})$ be typed \bterms. 
	If $\WH{((\Gamma, \oftype{x}{\oc O}),t)}\esub{\WH{(\Gamma,\bcbang{s})}}{x}$ is the bang \pnt{\bangMELL} as in \Cref{fig:substitution},
	$\WH{((\Gamma, \oftype{x}{\oc O}),t)}\esub{\WH{(\Gamma, \bcbang{s})}}{x} \tor{\ree}^+ \tor{\rax}^* \tor{\rwn}^* \WH{(\Gamma,t\bcsub{s}{x})}$.
    
    \begin{figure}[!th]
    \vspace*{-.5\baselineskip}
     \centering
     \input{Figures/substitution.tikz}
	     \caption{\label{fig:substitution}The \bpn involved in \Cref{lemma:substitution}, with $\Gamma = \oftype{x_1}{\oc O_1}, \dots, \oftype{x_k}{\oc O_k}$.}
    \end{figure}
\end{lemma}

\begin{proof}[Proof (sketch)]
	By induction on $t \in \tbcset$. 
	Some $\ltor{\rwn} \cup \town$ steps are needed to accommodate the configurations of $\wn$ nodes in $\WH{((\Gamma, \oftype{x}{\oc O}),t)}\esub{\WH{(\Gamma, \bcbang{s})}}{x}$ and apply the induction hypothesis to some sub-\bpns. Such $\ltor{\rwn} \cup \town$ steps can then be postponed to the end of the reduction sequence, and eventually transformed into $\town$ steps: it is crucial here that $\WH{(\Gamma,t\bcsub{s}{x})}$~is \rn{\rwn} (\Cref{rmk:wn-normal}) and $\town$ is confluent.
\end{proof}

\begin{theorem}
	\label{thm:simulation}
	Let $(\Gamma,t), (\Gamma, s)$ be typed \bterms.
	If $t \tor{\br} s$ then $\WH{(\Gamma,t)} \tocutwn^+ \WH{(\Gamma, s)}$.
	Indeed:
if $t \tor{\betab} s$ then $\WH{(\Gamma, t)} \tor{\rtp} \tor{\ree}^+ \tor{\rax}^* \tor{\rwn}^* \WH{(\Gamma, s)}$;
		and if $t \tor{\brder} s$ then $\WH{(\Gamma, t)} \tor{\rdr} \tor{\rax} \WH{(\Gamma, s)}$.
\end{theorem}
\begin{proof}
	We prove it for the root steps of $\tor{\br}$, the other cases follow immediately from the inductive hypothesis.
	If $t = \bcapp{\bcabs{x}{\oc O}{t_1}}{\bcbang{t_2}} \mapsto_{\betab} t_1\lsub{t_2}{x} = s$ then $\WH{(\Gamma, \bcapp{\bcabs{x}{\oc O}{t_1}}{\bcbang{t_2}})} \tor{\rtp} \tor{\rax} \WH{((\Gamma, \oftype{x}{\oc O}),t_1)}\esub{\WH{(\Gamma,\bcbang{t_2})}}{x} \tor{\ree}^+ \tor{\rax}^* \tor{\rwn}^* \WH{(\Gamma,t_1\bcsub{t_2}{x})} = \WH{(\Gamma, s)}$ by \Cref{lemma:substitution}, and the $\tor{\rax}$ step after $\tor{\rtp}$ can be postponed because it does not overlap with the other steps. 
	If $t = \bcder{\bcbang{s}} \mapsto_{\brder} s$ then it is immediate to check that $\WH{(\Gamma, \bcder{\bcbang{s}})} \tor{\rdr} \tor{\rax} \WH{(\Gamma, s)}$.
\end{proof}

Akin to \cite{DiCosmoKesner97,Tranquilli11} and differently from \cite{danos1990logique,DiCosmoGuerrini99,DiCosmoKesnerPolonowski03,laurent2003polarized,Guerrini04}, dealing with canonical forms (which requires $n$-ary $\wn$ nodes and reduction $\town$) is crucial to make \Cref{lemma:substitution,thm:simulation} hold \emph{without} invoking modulo equivalence. 
As rewriting modulo some equivalence is often convoluted, our simulation result is a main achievement showing that our \bpns are a fine-tuned (logical) setting to study the properties of the \bc.
If we had dropped the $\wn$-normal condition from \Cref{fig:bang-pns-app} in the definition of translation $\WH{(\cdot)}$, the order in which different occurrences of the same variable are contracted via $\wn$ nodes would matter in $\WH{(\cdot)}$, thus \Cref{lemma:substitution} (and so \Cref{thm:simulation}) would fail because $\WH{((\Gamma, \oftype{z}{\oc X}), \bcapp{x}{\bcapp{x}{z}})}\esub{\WH{(\Gamma,\bcbang{(\bcapp{x}{y})})}}{z} \not\tocutwn^* \WH{(\Gamma, \bcapp{x}{\bcapp{x}{\bcapp{x}{y}}})}$ for any 
\bcontext $\Gamma$.


From a \emph{static} point of view, we aim to characterize the \bpns corresponding to the translation of some typed \bterm. 
In the presence of axiom cuts, one easily finds a \bpn $R$ such that, for every typed \bterm $(\Gamma, t)$, $R \neq \WH{(\Gamma, t)}$: take the \bpn in \Cref{fig:bang-pns-var}, and add an axiom cut under its output conclusion. 
As $\bangMELL \subseteq \VMELL$, \Cref{cor:vmell} holds 
in $\psf{\bangMELL}$ and so it is natural to consider $\rmu$-normal \bpns (and $\rwn$-normal for canonicity). 
We then introduce a new translation of typed \bterms into \bpns: 
for every typed \bterm $(\Gamma, t)$, $\BL{(\Gamma, t)}$ is the \rnf{\rmuw} of $\WH{(\Gamma,t)}$. 
This is well defined because $\tor{\rmuw}$ is strongly normalizing and confluent. 
Detailed proofs~are~in~\Cref{subsec:characterization}.

\begin{theorem}[restate = characterization, name = Characterization of the \bc in \pns]
	\label{thm:characterization}
	Suppose that $R \in \bpnset$~has input conclusions $c_1, \dots, c_k$ of types $\wn O_1^\bot, \dots, \wn O_k^\bot$ respectively. 
	Then, $R$ is \rn{\rmuw} iff there is a typed \bterm $(\Gamma,t)$ such that $R = \BL{(\Gamma, t)}$, where $\Gamma = \oftype{\nmfun{R}(c_1)}{\oc O_1}, \dots, \oftype{\nmfun{R}(c_k)}{\oc O_k}$.
\end{theorem}

\subparagraph*{\Bpns vs.~simply typed \lci}
Typed \lterms can be translated to \bpns via Girard's two translations of intuitionistic minimal logic into \MELL \cite{girard1987linear}: $\CBNp{(\cdot)}$ for \CbNtext and $\CBVp{(\cdot)}$ for \CbVtext.
We recall these translations from 
\formulae{\STLC} to \formulae{\MELL}:
\begin{align*}
	\CBNT{X} &\Coloneqq  X & (\CBNT{\sigma \Rightarrow \tau)}& \Coloneqq \oc \CBNT{\sigma} \multimap \CBNT{\tau}
	&&&
	\CBVT{X} &\Coloneqq  X & (\CBVT{\sigma \Rightarrow \tau)}& \Coloneqq \oc \CBVT{\sigma} \multimap \oc \CBVT{\tau}
\end{align*}

The image of $\CBNT{(\cdot)}$ (resp.~$\CBVT{(\cdot)}$) is actually $\CbN$ (resp.~$\CbV$).
We extend these definitions to \lcontexts: if $\Gamma = x_1 : \tau_1, \dots, x_k : \tau_k$ then $\CBNT{\Gamma} = {x_1 : \oc \CBNT{\tau_1}, \dots, x_k : \oc \CBNT{\tau_k}}$ and similarly for $\CBVT{\Gamma}$.
With every typed \lterm $(\Gamma, M)$ we can now associate the \bpns $\CBNp{(\Gamma, M)} \in \psf{\CbN}$ and $\CBVp{(\Gamma,M)} \in \psf{\CbV}$ according to the well-studied definitions in~\cite{danos1990logique,regnier1992lambda,laurent2003polarized,Guerrini04,Tranquilli11} for~$\CBNp{(\cdot)}$~and in~\cite{Accattoli15} for $\CBVp{(\cdot)}$; we recall these definitions in \Cref{subsec:factorization}, see \Cref{fig:cbn-pns,fig:cbv-pns}.

A \lterm can also be translated to a \bterm (in the \bc) according to the two translations studied in \cite{ehrhard2016the,GuerrieriManzonetto18} (inspired by Girard's ones), $\CBNt{(\cdot)}$ for \CbNtext and $\CBVt{(\cdot)}$ for \CbVtext: the former wraps any application argument with $\oc$ (because $\tob$ can fire any $\beta$-redex), the latter only wraps values with $\oc$ (because only values can be erased or duplicated according to $\tobv$).
\begin{align*}
	\CBNt{x} &\Coloneqq x
	&
	\CBNt{(\labs{x}{\tau}{M})} &\Coloneqq \bcabs{x}{(\CBNT{\tau})}{\CBNt{M}}
	&
	\CBNt{(\lapp{M}{N})} &\Coloneqq \bcapp{\CBNt{M}}{\bcbang{(\CBNt{N})}}
	\\[-0.35em]
	\CBVt{x} &\Coloneqq \bcbang{x}
	&
	\CBVt{(\labs{x}{\tau}{M})} &\Coloneqq \bcbang{(\bcabs{x}{(\CBVT{\tau})}{\CBVt{M}})}
	&
	\CBVt{(\lapp{M}{N})} &\Coloneqq \bcapp{\bcder{\CBVt{M}}}{\CBVt{N}}
\end{align*}

Thus, there are two ways to embed typed \lterms into \bpns via \CbNtext (resp.~\CbVtext) translations: either directly via $\CBNp{(\cdot)}$ (resp.~$\CBVp{(\cdot)}$), or by composing $\CBNt{(\cdot)}$ (resp.~$\CBVt{(\cdot)}$) with the translation $\WH{(\cdot)}$ from \bterms to \bpns. 
It turns out that the two ways coincide. 

 \begin{restatable}{theorem}{factorization}
  \label{thm:factorization}
  For any typed \lterm $(\Gamma\!, M)$,  
 	$\WH{(\CBNT{\Gamma}\!\!, \CBNt{M})} = \CBNp{(\Gamma\!, M)}$ and $\WH{(\CBVT{\Gamma}\!\!, \CBVt{M})} \tor{\rax}^* \CBVp{(\Gamma\!, M)}$.
 \end{restatable}

\emph{Factorization} (\Cref{thm:factorization}) was conjectured in \cite{ehrhard2016the,GuerrieriManzonetto18} but its proof is not trivial (see \Cref{subsec:factorization}).
It shows again that \bpns with the embedding $\WH{(\cdot)}$ are a fine-tuned setting for studying \CbNtext and \CbVtext \lci uniformly, mediated by~the~\bc.

\subparagraph*{Conclusions}
We studied the role of connectivity in the theory of \pnst{\MELL} by identifying \VMELL as a new fragment of \MELL where the geometric property $\ACCw$ is a correctness criterion. 
Such a fragment reconciles intuitionistic and classical polarizations~of~\MELL, and contains the \bc (in turn subsuming \CbNtext and \CbVtext \lci).

 We plan to extend \VMELL to the additives in future work. We can already observe that the fragment defined by adding $P \oplus P$ and $\zero$ to the positive formulae (and dually $N \with N$~and $\top$ to the negatives) contains Girard's two translations of intuitionistic logic into \LL \cite{girard1987linear}.
 We conjecture that \pnst{\VMELL} with additives are handy and well-behaved.
 
 
 We plan to investigate the quotient on \bterms induced by the translation $\BL{(\cdot)}$ to \bpns. 
 We can easily prove that it contains $\sigma$-equivalence \cite{ehrhard2016the} (an equivalence on \bterms that permutes some constructs).
 We conjecture that they coincide. 
 One of the interests of studying the quotient induced by $\BL{(\cdot)}$ 
 is that it is validated by any denotational~model~of~the~\bc based on \LL.
 We also aim to study how \bpns embed the \bc~with explicit substitutions~\cite{BucciarelliKesnerRiosViso23,ArrialGuerrieriKesner23,KesnerArrialGuerrieri24,ArrialGuerrieriKesner24}.

%% file: Sections/Appendix.tex
\subsection{Complements of \texorpdfstring{\Cref{sec:geometric-preliminaries}}{Section 2}}
 \label{subsec:complements-geometric-preliminaries}
 
 \begin{definition}
  Let $G$ be a graph, and let $\gamma = n_0 a_1 n_1 a_2 \dots n_{k - 1} a_k n_k$ be a path of $G$.
  \begin{itemize}
   \item
    The non-negative integer $k$ is called the \emph{length} of $\gamma$ and is denoted by $\len{\gamma}$;
   \item
    For each $i, j \in \{0, \dots k\}$, if $i \leq j$ and the subsequence $n_i a_{i + 1} \dots n_{j - 1} a_j n_j$ is not a cycle,~we denote it by $n_i \gamma \, n_j$. We write $n_i \gamma \coloneq n_i \gamma \, n_k$ and $\gamma \, n_j \coloneq n_0 \gamma \, n_j$.
  \end{itemize}
 \end{definition}
  
 \begin{proposition}
  \label{prop:pfg-directed-path}
  Let $G$ be a \pfg, let $n \in \ve(G)$, let $\gamma$ be a non-empty~path of $G$ to $n$, and let $a$ be the last arc of $\gamma$. If $n = \head_G(a)$, then $\gamma$ is a directed path.
 \end{proposition}
  
 \begin{proof}
  By induction on $\len{\gamma}$. Let $n' \coloneq \tail_G(a)$. Then $\gamma = \gamma \, n' a \, n$. If $\gamma \, n'$ is empty, then we are done because $\gamma = n' a \, n$ is a directed path. Otherwise, we consider the last arc $a'$ of $\gamma \, n'$. Since $G$ is \pf, we know that $\outdeg_G(n') \leq 1$. And since $n' = \tail_G(a)$ and $a' \neq a$, we have $n' = \head_G(a')$. Finally, since $\len{\gamma \, n'} < \len{\gamma \, n'} + 1 = \len{\gamma}$, we can apply the induction hypothesis on $\gamma \, n'$, from which we deduce that $\gamma \, n'$ is a directed path. Since~$\gamma = \gamma \, n' a \, n$,~we~can conclude that $\gamma$ is a directed path.
 \end{proof}
  
  %

  \begin{remark}
   The null graph $\nullgraph$ is (vacuously) a \gs and, following \Cref{def:correctness-conditions}, we do \emph{not} have $\nullgraph \models \Cw$ even though $\nullgraph$ is connected. In particular, by \Cref{def:ACCw}, for~the \ps $R$ such that $\undgs{R} = \nullgraph$ we do \emph{not} have $R \models \Cw$.
  \end{remark}

 \subsubsection{Proof of \texorpdfstring{\Cref{prop:pfg-sink}}{Proposition 3}}
  \label{subsec:pfg-sink}
 
  \begin{lemma}[{\cite[Lemma 9.1.2]{lehman2010mathematics}}]
   \label{lemma:handshaking}
   Let $G$ be a graph. Then:
   \[
    \sum_{n \in \ve(G)} \outdeg_G(n) = \card{\ar(G)}
   \]
  \end{lemma}
 
 
  \begin{lemma}[{\cite[Corollary 11.9.8]{lehman2010mathematics}}]
   \label{lemma:connected-inequality}
   Let $G$ be a connected graph. Then $\card{\ar(G)} \geq \card{\ve(G)} - 1$.
  \end{lemma}
 
  \begin{lemma}
   \label{lemma:pfg-sink}
   Every connected \pfg has at most one sink.
  \end{lemma}
 
  We thank Cléophée Robin for suggesting the following proof of \Cref{lemma:pfg-sink}.
  
  \begin{proof}
   If $G$ is a connected \pfg and $S$ is the set of sinks of $G$, then~we~have:
   \begin{align*}
    \card{S} & = \card{\ve(G)} - \card{\{n \in \ve(G) : \outdeg_G(n) \neq 0\}} & \text{(definitions of sink and $\outdeg_G$)} \\
    & = \card{\ve(G)} - \sum_{n \in \ve(G)} \outdeg_G(n) & \text{(by partial functionality)} \\
    & = \card{\ve(G)} - \card{\ar(G)} & \text{(\Cref{lemma:handshaking})} \\
    & \leq 1 & \text{(\Cref{lemma:connected-inequality})} & \qedhere
   \end{align*}
  \end{proof}
  
  \begin{remark}
   \label{rmk:dag}
   If $G$ is a directed acyclic graph and $n \in \ve(G)$, then $n$ is a sink of $G$ if and only if $n$ is maximal with respect to $\precg{G}$. In particular, if $G$ is non-null, then there exists~at~least~one sink in $G$ (by~\Cref{rmk:directed-acyclic-order}).
  \end{remark}
 
  \pfgSink*
 
  \begin{proof}
   Straightforward consequence of \Cref{lemma:pfg-sink,rmk:dag}.
  \end{proof}
  
 \subsubsection{Stability of \texorpdfstring{$\ACCw$}{ACCw}}
 \label{subsec:cut-elimination-steps}
  
  We recall the following well-known result, expressing the stability of $\AC$ with respect~to~the~\cutelim steps (see~\cite{danos1990logique,llhandbook}).
  
  \begin{lemma}
   \label{lemma:acyclicity-preserved}
   Let $R,R' \in \psset$ such that $R \tocut R'$. If $R \models \AC$, then $R' \models \AC$.
  \end{lemma}
  
  It is quite easy to check that $\AC$ is also preserved under \conv steps.
  
  \begin{remark}
   \label{rmk:acyclicity-preserved}
   Let $R, R' \in \psset$ such that $R \town R'$. If $R \models \AC$, then $R' \models \AC$.
  \end{remark}
  
  Under the $\AC$ hypothesis, the property $\Cw$ is preserved by \cutelim and \conv.
  
  \begin{proposition}
   \label{prop:accw-cut-elim}
   Let $R, R' \in \psset$ such that $R \tocutwn R'$. If $R \models \ACCw$, then $R' \models \ACCw$.
  \end{proposition}
  
  \begin{proof}
   By \Cref{lemma:acyclicity-preserved,rmk:acyclicity-preserved}, we have $R' \models \AC$, so it suffices to prove that $R' \models \Cw$. We reason by induction on $R \tocutwn R'$:
   \begin{itemize}
    \item
     \emph{Base case (root):} we have $R \mapstocutwn R'$. By the hypothesis that $R \models \AC$ and by \Cref{lemma:acyclicity-preserved}, we know that $R' \models \AC$. First of all, we prove that $\undgs{R'} \models \Cw$. To do so, we fix an~arbitrary switching $\varphi'$ of $\undgs{R'}$, and we prove that $\smash{\undgs{R'}^{\varphi'} \models \Cw}$. Let $\varphi$ be any switching of $\undgs{R}$. Since~we have $\undgs{R}^\varphi \models \Cw$ and $\smash{\undgs{R'}^{\varphi'} \models \AC}$, we can apply systematically \Cref{prop:connected-components}. We distinguish seven possible cases:
     \begin{itemize}
      \item
       \emph{$R \mapstor{\rmu} R'$, or $R \mapstor{\rcm} R'$}. We have:
       \begin{align}
        \card{\nd(\undgs{R'}^{\varphi'})} & = \card{\nd(\undgs{R}^\varphi)} - 2 \notag 
        \\
        \card{\ar(\undgs{R'}^{\varphi'})} & = \card{\ar(\undgs{R}^\varphi)} - 2 \notag 
        \\
        \card{\w(\undgs{R'})} & = \card{\w(\undgs{R})} \notag 
        \vphantom{\undgs{R'}^{\varphi'}} \\
        \intertext{Therefore:}
        \card{\cc(\undgs{R'}^{\varphi'})} & = \card{\nd(\undgs{R'}^{\varphi'})} - \card{\ar(\undgs{R'}^{\varphi'})} \notag 
        \\
        & = \card{\nd(\undgs{R}^\varphi)} - \card{\ar(\undgs{R}^\varphi)} \notag 
        \\
        & = \card{\cc(\undgs{R}^\varphi)} \notag \\
        & = \card{\w(\undgs{R})} + 1 \notag 
        \\
        & = \card{\w(\undgs{R'})} + 1 \label{eqn:accw-multiplicative} 
       \end{align}
      \item
       $R \mapstor{\rob} R'$. We have:
       \begin{align}
        \card{\nd(\undgs{R'}^{\varphi'})} & = \card{\nd(\undgs{R}^\varphi)} - 3 \notag 
        \\
        \card{\ar(\undgs{R'}^{\varphi'})} & = \card{\ar(\undgs{R}^\varphi)} - 2 \notag 
        \\
        \card{\w(\undgs{R'})} & = \card{\w(\undgs{R})} - 1 \notag 
        \vphantom{\undgs{R'}^{\varphi'}} \\
        \intertext{Therefore:}
        \card{\cc(\undgs{R'}^{\varphi'})} & = \card{\nd(\undgs{R'}^{\varphi'})} - \card{\ar(\undgs{R'}^{\varphi'})} \notag 
        \\
        & = \card{\nd(\undgs{R}^\varphi)} - \card{\ar(\undgs{R}^\varphi)} - 1 \notag 
        \\
        & = \card{\cc(\undgs{R}^\varphi)} - 1 \notag \\
        & = \card{\w(\undgs{R})} \notag 
        \\
        & = \card{\w(\undgs{R'})} + 1 
        \label{eqn:accw-units}
       \end{align}
      \item
       $R \mapstor{\rdr} R'$             . Let $R_1 \in \psset$ as in \Cref{subfig:cutelim-rdr}, let $k$ be the number of conclusions of $R_1$,~and let $\varphi_1$ be the restriction of $\varphi'$ to a switching of $\undgs{R_1}$. We have:
       \begin{align*}
        \card{\nd(\undgs{R'}^{\varphi'})} & = \card{\nd(\undgs{R}^\varphi)} + \card{\nd(\undgs{R_1}^{\varphi_1})} - k - 2 
        \\
        \card{\ar(\undgs{R'}^{\varphi'})} & = \card{\ar(\undgs{R}^\varphi)} + \card{\ar(\undgs{R_1}^{\varphi_1})} - k - 1 
        \\
        \card{\w(\undgs{R'})} & = \card{\w(\undgs{R})} + \card{\w(\undgs{R_1})} 
        \vphantom{\undgs{R'}^{\varphi'}} \\
        \intertext{Observe that we have $\undgs{R_1} \models \ACCw$, because $R \models \ACCw$ and $R_1 \sqsubset R$. We deduce~that:}
        \card{\cc(\undgs{R'}^{\varphi'})} & = \card{\nd(\undgs{R'}^{\varphi'})} - \card{\ar(\undgs{R'}^{\varphi'})} \vphantom{\undgs{R_1}^{\varphi_1}} 
        \\
        & = \card{\nd(\undgs{R}^\varphi)} - \card{\ar(\undgs{R}^\varphi)} + \card{\nd(\undgs{R_1}^{\varphi_1})} - \card{\ar(\undgs{R_1}^{\varphi_1})} - 1 \\
        & = \card{\cc(\undgs{R}^\varphi)} + \card{\cc(\undgs{R_1}^{\varphi_1})} - 1 \\
        & = \card{\w(\undgs{R})} + \card{\w(\undgs{R_1})} + 1 \vphantom{\undgs{R_1}^{\varphi_1}} 
        \\
        & = \card{\w(\undgs{R'})} + 1 
       \end{align*}
      \item
       $R \mapstor{\rst} R'$. Let $R_1 \in \psset$ and $k, h \geq 0$ as in \Cref{subfig:cutelim-rst}. We consider two further subcases:
       \begin{itemize}
        \item
         $k = 0$. We have:
         \begin{align*}
          \card{\nd(\undgs{R'}^{\varphi'})} & = \card{\nd(\undgs{R}^\varphi)} + h - 3 
          \\
          \card{\ar(\undgs{R'}^{\varphi'})} & = \card{\ar(\undgs{R}^\varphi)} - 2 
          \\
          \card{\w(\undgs{R'})} & = \card{\w(\undgs{R})} + h - 1 
          \vphantom{\undgs{R'}^{\varphi'}} \\
          \intertext{Therefore:}
          \card{\cc(\undgs{R'}^{\varphi'})} & = \card{\nd(\undgs{R'}^{\varphi'})} - \card{\ar(\undgs{R'}^{\varphi'})} 
          \\
          & = \card{\nd(\undgs{R}^\varphi)} - \card{\ar(\undgs{R}^\varphi)} + h - 1 
          \\
          & = \card{\cc(\undgs{R}^\varphi)} + h - 1 \\
          & = \card{\w(\undgs{R})} + h 
          \\
          & = \card{\w(\undgs{R'})} + 1 
         \end{align*}
        \item
         $k \geq 1$. We have:
         \begin{align*}
          \card{\nd(\undgs{R'}^{\varphi'})} & = \card{\nd(\undgs{R}^\varphi)} + (h + 1)k - 2 
          \\
          \card{\ar(\undgs{R'}^{\varphi'})} & = \card{\ar(\undgs{R}^\varphi)} + (h + 1)k - 2 
          \\
          \card{\w(\undgs{R'})} & = \card{\w(\undgs{R})} 
          \vphantom{\undgs{R'}^{\varphi'}}
         \end{align*}
         Therefore, \Cref{eqn:accw-multiplicative} holds;
       \end{itemize}
      \item
       \emph{$R \mapstor{\rfl} R'$, flattening a $\wkpn$ node.} We have:
       \begin{align*}
        \card{\nd(\undgs{R'}^{\varphi'})} & = \card{\nd(\undgs{R}^\varphi)} - 2 
        \\
        \card{\ar(\undgs{R'}^{\varphi'})} & = \card{\ar(\undgs{R}^\varphi)} - 1 
        \\
        \card{\w(\undgs{R'})} & = \card{\w(\undgs{R})} - 1 
        \vphantom{\undgs{R'}^{\varphi'}}
       \end{align*}
       Therefore, \Cref{eqn:accw-units} holds;
      \item
       \emph{$R \mapstor{\rfl} R'$, flattening a $\wn$ node with at least one premise, or $R \mapstor{\rer} R'$.} We have:
       \begin{align*}
        \card{\nd(\undgs{R'}^{\varphi'})} & = \card{\nd(\undgs{R}^\varphi)} - 1 
        \\
        \card{\ar(\undgs{R'}^{\varphi'})} & = \card{\ar(\undgs{R}^\varphi)} - 1 
        \\
        \card{\w(\undgs{R'})} & = \card{\w(\undgs{R})} 
        \vphantom{\undgs{R'}^{\varphi'}}
       \end{align*}
       Therefore, \Cref{eqn:accw-multiplicative} holds;
      \item
       $R \mapstor{\rpu} R'$. If $k$ is the number of premises of the pulled $\wn$ node, we consider the subcases:
       \begin{itemize}
        \item
         $k = 0$. We have:
         \begin{align*}
          \card{\nd(\undgs{R'}^{\varphi'})} & = \card{\nd(\undgs{R}^\varphi)} + 1 
          \\
          \card{\ar(\undgs{R'}^{\varphi'})} & = \card{\ar(\undgs{R}^\varphi)} 
          \\
          \card{\w(\undgs{R'})} & = \card{\w(\undgs{R})} + 1 
          \vphantom{\undgs{R'}^{\varphi'}} \\
          \intertext{Therefore:}
          \card{\cc(\undgs{R'}^{\varphi'})} & = \card{\nd(\undgs{R'}^{\varphi'})} - \card{\ar(\undgs{R'}^{\varphi'})} 
          \\
          & = \card{\nd(\undgs{R}^\varphi)} - \card{\ar(\undgs{R}^\varphi)} + 1 
          \\
          & = \card{\cc(\undgs{R}^\varphi)} + 1 \\
          & = \card{\w(\undgs{R})} + 2 
          \\
          & = \card{\w(\undgs{R'})} + 1 
         \end{align*}
        \item
         $k \geq 1$. We have:
         \begin{align*}
          \card{\nd(\undgs{R'}^{\varphi'})} & = \card{\nd(\undgs{R}^\varphi)} + k 
          \\
          \card{\ar(\undgs{R'}^{\varphi'})} & = \card{\ar(\undgs{R}^\varphi)} + k 
          \\
          \card{\w(\undgs{R'})} & = \card{\w(\undgs{R})} 
          \vphantom{\undgs{R'}^{\varphi'}}
         \end{align*}
         Therefore, \Cref{eqn:accw-multiplicative} holds.
       \end{itemize}
     \end{itemize}
     We now show that $R_0' \models \Cw$ for all $R_0' \sqsubset R'$. In every case except for $R \mapstor{\rdr} R'$,~$R \mapstor{\rcm} R'$ and $R \mapstor{\rpu} R'$, this is obvious because, for every $R_0' \sqsubset R'$, we actually have $R_0' \sqsubset R$, and~by hypothesis we know that $R \models \Cw$. We thus consider the three non-trivial cases:
     \begin{itemize}
      \item
       $R \mapstor{\rdr} R'$. Let $R_1 \in \psset$ as in \Cref{subfig:cutelim-rdr}. Since $R_1 \sqsubset R$, we know that $R_1 \models \Cw$. Now, if $R_0' \sqsubset R'$, then either $R_0' \sqsubset R$ or $R_0' \sqsubset R_1$, and in both cases we deduce~that~$R_0' \models \Cw$;
      \item
       $R \mapstor{\rcm} R'$. Let $R_1 \coloneq \undbox{R}(n)$, $R_2 \coloneq \undbox{R}(m)$, $R_2' \coloneq \undbox{R'}(m)$ as in \Cref{subfig:cutelim-rcm}, and~let $R_0' \sqsubset R'$. We have either $R_0' \sqsubset R$, and in this case we are done, or $R_0' = R_2'$. Let $k$ be the number of conclusions of $R_1$, and let $\varphi_2$ be a switching of $\undgs{R_2'}$ (or, equivalently,~of~$\undgs{R_2}$). Then:
       \begin{align*}
        \card{\nd(\undgs{R_2'}^{\varphi_2})} & = \card{\nd(\undgs{R_2}^{\varphi_2})} + k \\
        \card{\ar(\undgs{R_2'}^{\varphi_2})} & = \card{\ar(\undgs{R_2}^{\varphi_2})} + k \\
        \card{\w(\undgs{R_2'})} & = \card{\w(\undgs{R_2})}
        \intertext{Observe that we have $\undgs{R_2} \models \ACCw$, because $R \models \ACCw$ and $R_2 \sqsubset R$. We deduce~that:}
        \card{\cc(\undgs{R_2'}^{\varphi_2})} & = \card{\nd(\undgs{R_2'}^{\varphi_2})} - \card{\ar(\undgs{R_2'}^{\varphi_2})} \\
        & = \card{\nd(\undgs{R_2}^{\varphi_2})} - \card{\ar(\undgs{R_2}^{\varphi_2})} \vphantom{\undgs{R_2'}^{\varphi_2}} \\
        & = \card{\cc(\undgs{R_2}^{\varphi_2})} \vphantom{\undgs{R_2'}^{\varphi_2}} \\
        & = \card{\w(\undgs{R_2})} + 1 \vphantom{\undgs{R_2'}^{\varphi_2}} \\
        & = \card{\w(\undgs{R_2'})} + 1
       \end{align*}
       We thus obtain $\undgs{R_2'}^{\varphi_2} \models \Cw$. Since our choice of the switching $\varphi_2$ of $\undgs{R_2'}$ was arbitrary,~we have $\undgs{R_2'} \models \Cw$. On the other hand, for each $i \in \{1, 2\}$, we have that $R_i \models \Cw$, because $R_i \sqsubset R$ and $R \models \Cw$. For every $R_0'' \sqsubset R_2'$, we have either $R_0'' = R_1$ or $R_0'' \sqsubset R_2$, and~in both cases it follows straightforwardly that $R_0'' \models \Cw$, thus $R_2' \models \Cw$;
      \item
       $R \mapstor{\rpu} R'$. Let $R_1, R_1' \in \psset$ as in \Cref{subfig:pull}, and let $R_0' \sqsubset R'$. Then either $R_0' \sqsubset R$,~and in this case we are done, or $R_0' = R_1'$. We fix a switching $\varphi_1$ of $\undgs{R_1}$, a switching $\varphi_1'$~of~$\undgs{R_1'}$, and we prove that $\smash{\undgs{R_1'}^{\varphi_1'} \models \Cw}$. We observe that $\undgs{R_1} \models \ACCw$, because $R \models \ACCw$~and $R_1 \sqsubset R$. Let $k$ be the number of premises of the pulled $\wn$ node. Consider two subcases:
       \begin{itemize}
        \item
         $k = 0$. We have:
         \begin{align*}
          \card{\nd(\undgs{R_1'}^{\varphi_1'})} & = \card{\nd(\undgs{R_1}^{\varphi_1})} - 2 \\
          \card{\ar(\undgs{R_1'}^{\varphi_1'})} & = \card{\ar(\undgs{R_1}^{\varphi_1})} - 1 \\
          \card{\w(\undgs{R_1'})} & = \card{\w(\undgs{R_1})} - 1 \vphantom{\undgs{R_1'}^{\varphi_1'}}
          \intertext{Therefore:}
          \card{\cc(\undgs{R_1'}^{\varphi_1'})} & = \card{\nd(\undgs{R_1'}^{\varphi_1'})} - \card{\ar(\undgs{R_1'}^{\varphi_1'})} \\
          & = \card{\nd(\undgs{R_1}^{\varphi_1})} - \card{\ar(\undgs{R_1}^{\varphi_1})} - 1 \vphantom{\undgs{R_1'}^{\varphi_1'}} \\
          & = \card{\cc(\undgs{R_1}^{\varphi_1})} - 1 \vphantom{\undgs{R_1'}^{\varphi_1'}} \\
          & = \card{\w(\undgs{R_1})} \vphantom{\undgs{R_1'}^{\varphi_1'}} \\
          & = \card{\w(\undgs{R_1'})} + 1 \vphantom{\undgs{R_1'}^{\varphi_1'}}
         \end{align*}
        \item
         $k \geq 1$. We have:
         \begin{align*}
          \card{\nd(\undgs{R_1'}^{\varphi_1'})} & = \card{\nd(\undgs{R_1}^{\varphi_1})} - 1 \\
          \card{\ar(\undgs{R_1'}^{\varphi_1'})} & = \card{\ar(\undgs{R_1}^{\varphi_1})} - 1 \\
          \card{\w(\undgs{R_1'})} & = \card{\w(\undgs{R_1})} \vphantom{\undgs{R_1'}^{\varphi_1'}}
          \intertext{Therefore:}
          \card{\cc(\undgs{R_1'}^{\varphi_1'})} & = \card{\nd(\undgs{R_1'}^{\varphi_1'})} - \card{\ar(\undgs{R_1'}^{\varphi_1'})} \\
          & = \card{\nd(\undgs{R_1}^{\varphi_1})} - \card{\ar(\undgs{R_1}^{\varphi_1})} \vphantom{\undgs{R_1'}^{\varphi_1'}} \\
          & = \card{\cc(\undgs{R_1}^{\varphi_1})} \vphantom{\undgs{R_1'}^{\varphi_1'}} \\
          & = \card{\w(\undgs{R_1})} \vphantom{\undgs{R_1'}^{\varphi_1'}} + 1 \\
          & = \card{\w(\undgs{R_1'})} + 1 \vphantom{\undgs{R_1'}^{\varphi_1'}}
         \end{align*}
       \end{itemize}
       Since our choice of the switching $\varphi_1'$ of $\undgs{R_1'}$ was arbitrary, we obtain $\undgs{R_1'} \models \Cw$. We now observe that $R_1 \models \Cw$, because $R_1 \sqsubset R$ and $R \models \Cw$. For every $R_0'' \sqsubset R_1'$, we~have $R_0'' \sqsubset R_1$, and thus $R_0'' \models \Cw$. We can then conclude that $R_1' \models \Cw$.
     \end{itemize}
    \item
     \emph{Induction step (box):} in this case we have $\undgs{R} = \undgs{R'}$, there exists a $\oc$ node $n$ of $\undgs{R}$ such that $\undbox{R}(n) \tocutwn \undbox{R'}(n)$ and, for every $\oc$ node $m \neq n$ of $\undgs{R}$, we have $\undbox{R}(m) = \undbox{R'}(m)$. Since $R \models \Cw$, from $\undgs{R} = \undgs{R'}$ it follows that $\undgs{R'} \models \Cw$ and, for every $\oc$ node $m \neq n$,~from $\undbox{R}(m) = \undbox{R'}(m)$ we deduce that $\undbox{R'}(m) \models \Cw$. Moreover, since $\undbox{R}(n) \sqsubset R$ and $R \models \ACCw$, we have $\undbox{R}(n) \models \ACCw$. We can then apply the induction hypothesis~and obtain that $\undbox{R'}(n) \models \Cw$, thus $R' \models \Cw$. \qedhere
   \end{itemize}
  \end{proof}
  
%
  
  \accwCutElim*
  
  \begin{proof}
   Special case of \Cref{prop:accw-cut-elim}.
  \end{proof}
 
\subsection{Complements of \texorpdfstring{\Cref{sec:pure-polarized}}{Section 3}}
 \label{subsec:complements-untyped-polarized}
 
 We provide full details on the results of \Cref{sec:pure-polarized}.
 
 \begin{remark}
 	\label{rmk:switching}
 	If $\gamma$ is a path of a \gs $\mathcal{G}$, there exists a switching $\varphi$ of $\mathcal{G}$ such that $\gamma$ is a path of $\mathcal{G}^\varphi$ if and only if, for every $\parr$ or $\ctpn$ node $n$ in $\mathcal{G}$, at most one premise~of~$n$~is~in~$\gamma$.
 	Indeed, if more than one premise of some $\parr$ or $\ctpn$ node $n$ is in $\gamma$, then for every switching $\varphi$ of $\mathcal{G}$, at least one of those premises is erased from $\mathcal{G}^\varphi$ and so $\gamma$ is not a path of $\mathcal{G}^\varphi$.
 	And conversely, if for every $\parr$ or $\ctpn$ node $n$ in $\mathcal{G}$ at most one premise of $n$ is in $\gamma$, then just define a switching $\varphi$ of $\mathcal{G}$ so that $\varphi(n)$ is the premise of $n$ in $\gamma$, for every $\parr$ or $\ctpn$ node $n$~in~$\mathcal{G}$~such~that one of its premises is in $\gamma$; thus $\gamma$ is a path of $\mathcal{G}^\varphi$.
 \end{remark}
 
 \acyclicity*
 
 \begin{proof}
  Suppose $\iograph{\mathcal{G}}$ has a directed cycle $\delta$. We claim that there exists a switching $\varphi$ of $\mathcal{G}$ such that $\delta$ is a cycle of $\mathcal{G}^\varphi$. 
  Note that $\delta$ is a cycle in $\mathcal{G}$.
  By \Cref{rmk:switching}, it is enough to prove that, for every $\parr$ or $\ctpn$ node $n$ of $\mathcal{G}$, at most one premise of $n$ is an arc of $\delta$. This is obvious if $n$ is a $\pout$ node, because the input premise of $n$ is not an arc of $\iograph{\mathcal{G}}$. On the other hand, if $n$ is a $\pin$ or $\ctpn$ node, every premise of $n$ is input~(and~so~outgoing~from~$n$),~hence~at~most~one~premise~of~$n$~is~an	 arc~of~$\delta$,~because~$\delta$~is~a~directed~cycle~of~$\iograph{\mathcal{G}}$.
 \end{proof}
 
 \connectedComponentsSinks*
 
 \begin{proof}
 	\begin{enumerate}
 		\item  
		By \Cref{rmk:switching-connected-components} (as $\mathcal{G} \models \AC$), the number of connected component is invariant for any switching. 
		Thus, without loss of generality, we can take a switching $\varphi$ of $\mathcal{G}$~such~that~$\varphi(n)$~is the output premise of $n$ for every $\pout$ node $n$ of $\mathcal{G}$.\footnote{Such a switching is called \emph{intuitionistic} in~\cite[Appendix A.4]{didonna2026roleconnectivitylinearlogic}.} 
 		
 		Let $G$ be the graph we obtain from $\iograph{\mathcal{G}}$ by erasing, for each $\pin$ or $\ctpn$ node $n$ of $\mathcal{G}$, every premise of $n$ except for $\varphi(n)$. Then $\mathcal{G}^\varphi$ is obtained from $G$ by reversing the orientation of every input arc. In particular, we have $\card{\cc(\mathcal{G}^\varphi)} = \card{\cc(G)}$.
 		
 		We now observe that, for every $n \in \nd(\mathcal{G})$, if $n$ is a $\pin$ or $\ctpn$ node, then we have $\smash{\outdeg_{\iograph{\mathcal{G}}}(n) \geq 2}$ and $\outdeg_G(n) = 1$, otherwise $\smash{\outdeg_{\iograph{\mathcal{G}}}(n) = \outdeg_G(n)}$. Therefore, the sinks of $G$ are the same as the~sinks of $\iograph{\mathcal{G}}$ and, by \Cref{itm:io-graph-outdegree} of \Cref{rmk:io-graph}, we know that $G$ is a \pfg. Moreover,~by \Cref{prop:acyclicity}, we have $\mathcal{G} \models \ACio$, meaning that $\iograph{\mathcal{G}}$ is directed acyclic. We deduce that every $C \in \cc(G)$ is a directed acyclic, connected and \pfg: by \Cref{prop:pfg-sink},~it has exactly one sink. By summing on the connected components of $G$, we deduce that $\card{\cc(G)}$ is the number of sinks of $G$ or, equivalently, of $\iograph{\mathcal{G}}$. By~\Cref{itm:io-graph-sinks}~of \Cref{rmk:io-graph}, we conclude that $\card{\cc(\mathcal{G}^\varphi)} = \card{\cc(G)} = \card{\w(\mathcal{G})} + \card{\drcl(\mathcal{G})} + \card{\out(\mathcal{G})}$.
 		
 		\item As $\mathcal{G} \models \AC$, the following equivalences hold:
 		\begin{align*}
 			& &&\mathcal{G} \models \Cw
 			\\
 			 & \iff \textup{(\Cref{def:correctness-conditions})} && \card{\cc(\mathcal{G}^\varphi)} = \card{\w(\mathcal{G})} + 1 \textup{ for every switching $\varphi$ of $\mathcal{G}$}
 			\\
 			& \iff \textup{(\Cref{lemma:connected-components}.\ref{p:connected-components-sinks})} && \card{\drcl(\mathcal{G})} + \card{\out(\mathcal{G})} =  1
 			\\
 			& \iff \textup{(\Cref{def:accio})} && \mathcal{G} \models \Cio \qedhere
 		\end{align*}
 	\end{enumerate}
 \end{proof}
 
 \begin{notation}
  \label{not:io-order}
  Let $\mathcal{G}$ be a \pol \gs. When there is no ambiguity, we use the notation ${\precio} \coloneq {\precg{\iograph{\mathcal{G}}}}$ (recall \Cref{def:directed-acyclic-order}). We also write $\succio$ for the converse~relation~of~$\precio$: for every $n_1, n_2 \in \nd(\mathcal{G})$, we have $n_1 \succio n_2$ if and only if $n_2 \precio n_1$.
 \end{notation}
 
 \begin{remark}
  \label{rmk:io-order}
  Let $\mathcal{G}$ be a \pol \gs such that $\mathcal{G} \models \ACio$. Then, by \Cref{rmk:directed-acyclic-order}, the binary relation $\precio$ is a strict order relation. Thus, since $\nd(\mathcal{G})$ is finite, both $\precio$~and~$\succio$ are well-founded.
 \end{remark}

 \begin{lemma}
  \label{lemma:poutter-terminal-pout}
  Let $\mathcal{G}$ be a \poutter \gs such that $\mathcal{G} \models \ACio$. If there exists~a $\pout$ node of $\mathcal{G}$, then there exists a terminal $\pout$ node of $\mathcal{G}$.
 \end{lemma}
 
 \begin{proof}
  By induction on $\succio$ (\Cref{rmk:io-order}). By hypothesis, we can consider a $\pout$ node $n$ of $\mathcal{G}$. If $n$ is terminal, we are done. Otherwise, by the hypothesis that $\mathcal{G}$ is \poutter,~the~conclusion of $n$ is a premise of a $\pout$ node $n'$ of $\mathcal{G}$. Since $\mathcal{G} \models \ACio$, we have $n' \neq n$, hence $n' \succio n$. We can then apply the induction hypothesis on $n'$ to conclude that $\mathcal{G}$ has a terminal~$\pout$~node.
 \end{proof}
 
 We define the analogous notion of $\pout$ node and $\pin$ node for $\otimes$ nodes.
  
 \begin{definition}
  Let $\mathcal{G}$ be a \pol \gs. A $\otimes$ node of $\mathcal{G}$ is a $\tout$~(resp.~$\tin$)~node if its conclusion is output (resp.~input).
 \end{definition}
 
 \begin{lemma}
  \label{lemma:input-conclusion-terminal}
  Let $\mathcal{G}$ be a \pol \gs with no $\pout$ node and such that $\mathcal{G} \models \ACio$, and let $n \in \nd(\mathcal{G})$. Then at least one of the following holds:
  \begin{itemize}
   \item
    There exists a terminal node $n_0 \precio n$;
   \item
    Every input conclusion of $n$ is a conclusion of $\mathcal{G}$.
  \end{itemize}
 \end{lemma}
  
 \begin{proof}
  By induction on $\precio$ (\Cref{rmk:io-order}). We suppose that there exists an input conclusion $a$ of $n$ which is not a conclusion of $\mathcal{G}$, and we prove that there exists a terminal node~$n_0 \precio n$. Since $a$ is input and not a conclusion of $\mathcal{G}$, we have that $a$ is a premise of a $\cutpn$, $\tin$, $\pin$, $\drpn$ or $\wn$ node $n'$ of $\mathcal{G}$. And since $\mathcal{G} \models \ACio$, we know that $n' \neq n$. We thus obtain that $n' \precio n$. Now, if $n'$ is terminal, then we are done by picking $n_0 \coloneq n'$. Otherwise, we have that $n'$ is not a $\cutpn$ node (because $\cutpn$ nodes are always terminal), so $n'$ has exactly~one~input~conclusion~$a'$,~and~$a'$ is not a conclusion of $\mathcal{G}$. Since $n' \precio n$, we can~apply~the induction hypothesis on $n'$ to~deduce the existence of a terminal node $n_0 \precio n'$. Since $\precio$ is a strict order relation~(\Cref{rmk:io-order}),~we can conclude that $n_0 \precio n$.
 \end{proof}
  
 \begin{lemma}
  \label{lemma:existence-terminal-node}
  Let $\mathcal{G}$ be a \pol \gs with no $\pout$ node and such that $\mathcal{G} \models \ACio$. Then there exists a terminal node of $\mathcal{G}$.
 \end{lemma}
  
 \begin{proof}
  If there exists an input arc of $\mathcal{G}$ that is not a conclusion of $\mathcal{G}$, then we are done by taking $n$ as the node of $\mathcal{G}$ having $a$ among its conclusions in \Cref{lemma:input-conclusion-terminal}. On the other hand, if every input arc of $\mathcal{G}$ is a conclusion~of~$\mathcal{G}$, then $\mathcal{G}$ has no $\pin$ or $\ctpn$ node: the premises of a $\pin$ or $\ctpn$ node of $\mathcal{G}$ would be input arcs of $\mathcal{G}$ and not conclusions of $\mathcal{G}$. By \Cref{itm:io-graph-outdegree} of \Cref{rmk:io-graph}, we obtain that $\iograph{\mathcal{G}}$ is a \pfg. In particular,~by~the hypothesis $\mathcal{G} \models \ACio$, any $C \in \cc(\iograph{\mathcal{G}})$ is a directed acyclic and connected \pfg. From \Cref{prop:pfg-sink} we deduce that $C$ has exactly one sink $n$, and we know, by \Cref{itm:io-graph-sinks} of \Cref{rmk:io-graph}, that $n$ is a $\bot$, $\wkpn$, $\drpn$ or output $\bullet$ node of $\mathcal{G}$. If $n$ is a $\bot$, $\wkpn$ or $\drpn$ node, then $n$ is terminal, because we are supposing that every input arc of $\mathcal{G}$ is a conclusion of $\mathcal{G}$. Otherwise, the premise of $n$ is an $\axpn$, $\onepn$, $\tout$ or $\oc$ node $n_0$ of $\mathcal{G}$, and~then~$n_0$~is~terminal,~again~by~the~assumption~that~every~input~arc~of~$\mathcal{G}$~is~a conclusion of $\mathcal{G}$.
 \end{proof}
 
 \lemmaIoter*
 
 \begin{proof}
  By induction on $\card{\nd(\mathcal{G})}$. Let $a$ be an output premise of $n$. Since $n$ is a terminal node of $\mathcal{G}$ with an output premise, and since $\mathcal{G}$ has no $\pout$ node, the node $n$ is a $\cutpn$, $\tout$ or $\drpn$ node. We consider the node $n'$ of $\mathcal{G}$ which has $a$ among its conclusions, and we~equip $C$ with an arbitrary total ordering on its conclusions, so that $C$ becomes a \gs. From the hypothesis $\mathcal{G} \models \ACio$ we deduce that $n' \neq n$, thus $n' \precg{\iograph{\mathcal{G}}} n$, and that $C \models \ACio$, because a directed cycle of $\iograph{C}$ would also be a directed cycle of $\iograph{\mathcal{G}}$. Since $n'$ has an~output~conclusion~(the~arc $a$), and since $\mathcal{G}$ has no $\pout$ node, we know that $n'$ is an $\axpn$, $\onepn$, $\tout$ or $\oc$ node.
  
  We claim that every node $n''$ of $C$ such that $n'' \precg{\iograph{C}} n'$ is not a terminal node of $C$. Indeed, if that were not the case, then there would be a terminal node $n''$ of $C$ such that $n'' \precg{\iograph{C}} n'$. In particular, since $\precg{\iograph{C}}$ is a strict order relation (\Cref{rmk:io-order}), we would have $n'' \neq n'$. Since $n''$ is terminal, and since $\bullet$ nodes are never terminal, we would have $n'' \neq n_0$, thus $n''$ would be a terminal node of $\mathcal{G}$. On the other hand, by definition of $\precg{\iograph{C}}$, there would exist a directed path $\delta$ of $\iograph{C}$ from $n''$ to $n'$. Since $n' \neq n_0 \neq n''$, the node $n_0$ would not be a node of $\delta$, thus $\delta$ would also be a directed path of $\iograph{\mathcal{G}}$. Then $n'' \precg{\iograph{\mathcal{G}}} n'$. Finally, since $\precg{\iograph{\mathcal{G}}}$ is a strict order relation, we would obtain $n'' \precg{\iograph{\mathcal{G}}} n$, contradicting~the~hypothesis~that~$n$~is~a~\ioter~node of $\mathcal{G}$. We have thus proven our claim.
       
  We now apply \Cref{lemma:input-conclusion-terminal}: since every node $n''$ of $C$ such that $n'' \precg{\iograph{C}} n'$ is not a terminal node of $C$, every input conclusion of $n'$ is a conclusion of $C$. It follows that $n'$ is a terminal node of $C$. We then observe that $n'$ is actually a \ioter node of~$C$,~because~every~node~$n''$~of~$C$ such that $n'' \precg{\iograph{C}} n'$ is not a terminal node of $C$. And since $n$ is a $\cutpn$, $\tout$ or $\drpn$ node, we have that $\card{\nd(C)} < \card{\nd(C)} + 1 \leq \card{\nd(\mathcal{G})}$. We can then apply the induction hypothesis~on~$C$~and~$n'$~in the rest of the proof. We prove \Cref{itm:ioter-connected-component,itm:ioter-directed-path} for $\mathcal{G}$, $n$ and $a$:
  \begin{enumerate}[i.]
   \item
    If $n'$ is an $\axpn$, $\onepn$ or $\oc$ node, we are done because $n'$ is a terminal node of $C$. 
    Otherwise $n'$ is a $\tout$ node: let $a_1$ and $a_2$ be the two premises of $n'$. For each $i \in \{1, 2\}$, let $G_i$ be the graph obtained from $C$ by adding a fresh $\bullet$ node $n_i$ and by setting the head of $a_i$ to be $n_i$, and let $C_i$ be the connected component of $G_i$ that contains $n_i$. By \Cref{itm:ioter-connected-component} for $C$, $n'$ and $a_i$, we know that $C_i$ has no $\drpn$ node, and that $n_i$ is the only $\bullet$ node of $C_i$~whose premise is output. We notice that, for every $n'' \in \nd(C)$, either $n'' = n'$, or $n'' = n_0$, or there exists $i \in \{1, 2\}$ such that $n'' \in \nd(C_i)$. Since $n'$ is a $\tout$ node and $n_0$ is a $\bullet$ node, we conclude~that~$C$~has~no $\drpn$ node, and~that~$n_0$~is~the~only~$\bullet$~node~of~$C$~whose~premise~is~output;
   \item
    Let $\gamma$ be a path of $\mathcal{G}$ to $n$, and suppose that the last arc of $\gamma$ is $a$. Then $\gamma = \gamma \, n' a \, n$. If $\gamma \, n'$ is empty then, since $a$ is output, we can immediately conclude that $\gamma$ is a directed path of $\iograph{\mathcal{G}}$. We then suppose $\gamma \, n'$ is non-empty. Then $n'$ is an $\axpn$, $\tout$ or $\oc$ node. If $n'$ is an $\axpn$ or $\oc$ node, then a conclusion $a' \neq a$ of $n'$ (hence, an input arc) is the only arc of $\gamma \, n'$, because $n'$ is a terminal node of $C$, so we can conclude that $\gamma$ is a directed path of $\iograph{\mathcal{G}}$. Finally, if $n'$ is a $\tout$ node, the last arc of $\gamma \, n'$ is a premise $a'$ of $n'$. Since $a$ is not an arc of $\gamma \, n'$, we know that $\gamma \, n'$ is also a path of $C$. Thus, by \Cref{itm:ioter-directed-path} for $C$, $n'$ and $a'$, we obtain that $\gamma \, n'$ is a directed path of $\iograph{C}$. In particular,~the~path~$\gamma \, n'$~is~a~directed~path~of~$\iograph{\mathcal{G}}$,~from~which we~immediately~deduce~that $\gamma$ is a directed path of $\iograph{\mathcal{G}}$. \qedhere
  \end{enumerate}
 \end{proof}
 
 \begin{lemma}[Splitting]
  \label{prop:splitting}
  Let $\mathcal{G}$ be a \pol \gs with no $\pout$ node and such that $\mathcal{G} \models \ACio$, and let $n$ be a \ioter $\cutpn$ or $\otimes$ node of $\mathcal{G}$. For every cycle $\gamma$ of $\mathcal{G}$,~we~have~that $n$ is not a node of $\gamma$.
 \end{lemma}
  
 \begin{proof}
  Since $n$ is a $\cutpn$ or $\otimes$ node of $\mathcal{G}$, there exists at least one output premise $a$ of $n$. If there were a cycle $\gamma$ of $\mathcal{G}$ having $n$ among its nodes, then $\gamma$ would have $a$ among its arcs, because $n$ is a terminal node of $\mathcal{G}$. We could then suppose, without loss of generality, that the last arc of $\gamma$ is $a$ and, by \Cref{itm:ioter-directed-path} of \Cref{lemma:ioter}, we would obtain that~$\gamma$~is~a~directed~cycle~of~$\iograph{\mathcal{G}}$,~contradicting the hypothesis $\mathcal{G} \models \ACio$.
 \end{proof}
  
 \begin{remark}
  \label{rmk:acio-connected-component}
  Let $\mathcal{G}$ be a \pol \gs, and let $C \in \cc(\mathcal{G})$. Then $C$ inherits from $\mathcal{G}$ a total ordering on the premises of its $\bullet$ nodes, so $C$ is a \pol \gs. We also notice that, if $\mathcal{G} \models \ACio$, then $C \models \ACio$ (any directed cycle of $\iograph{C}$ is a directed~cycle~of~$\iograph{\mathcal{G}}$).
 \end{remark}
  
 \begin{lemma}
  \label{lemma:ioter-connected-component}
  Let $\mathcal{G}$ be a \pol \gs with no $\pout$ node, such that $\mathcal{G} \models \ACio$,~and let $C \in \cc(\mathcal{G})$. Then every \ioter node of $C$ is a \ioter node of $\mathcal{G}$.
 \end{lemma}
  
 \begin{proof}
  Let $n$ be a \ioter node of $C$. If $n$ were not a \ioter node of $\mathcal{G}$, then~we~would have a terminal node $n'$ of $\mathcal{G}$ such that $n' \precg{\iograph{\mathcal{G}}} n$ and so, by definition of $\precg{\iograph{\mathcal{G}}}$, we would have a directed path $\delta$ of $\iograph{\mathcal{G}}$ from $n'$ to $n$. But then $n'$ would be a terminal node of $C$ and $\delta$ would be a directed path of $\iograph{C}$, because $C$ is the connected component of $\mathcal{G}$ which contains $n$. Hence, we would have $n' \precg{\iograph{C}} n$, contradicting the fact that $n$ is a \ioter~node~of~$C$.
 \end{proof}
  
 \begin{lemma}
  \label{lemma:ioter-cut-tensor}
  Let $\mathcal{G}$ be a \poutter \gs having no terminal $\bot$, $\parr$, $\drpn$ or $\wn$ node, and such that $\mathcal{G} \models \ACio$. If $\mathcal{G}$ has a terminal $\cutpn$ or $\otimes$ node, then $\mathcal{G}$ has~a~\ioter~$\cutpn$~or $\otimes$ node.
 \end{lemma}
  
 \begin{proof}
  Let $n$ be a terminal $\cutpn$ or $\otimes$ node of $\mathcal{G}$. We consider the connected component $C$ of $\mathcal{G}$ which contains $n$. Since $\mathcal{G}$ is \poutter and has no terminal $\pout$ node, by \Cref{lemma:poutter-terminal-pout}~we~know that $\mathcal{G}$ (and in particular $C$) has no $\pout$ node. By \Cref{lemma:existence-terminal-node} (and by \Cref{rmk:acio-connected-component}),~there~exists~a terminal node of $C$, so we can consider a \ioter node $n_0$ of $C$. By \Cref{lemma:ioter-connected-component}, we~know~that $n_0$ is also a \ioter node of $\mathcal{G}$. We prove that $n_0$ is a $\cutpn$ or $\otimes$ node. By the hypothesis~that $\mathcal{G}$ has no terminal $\bot$, $\parr$, $\drpn$ or $\wn$ node, we know that $n_0$ is an $\axpn$, $\cutpn$, $\onepn$, $\otimes$ or $\oc$ node. It~suffices to exclude that $n_0$ is an $\axpn$, $\onepn$ or $\oc$ node: if that were the case then, besides $n_0$, every node~of~$C$ would be a $\bullet$ node, and thus the $\cutpn$ or $\otimes$ node $n$ would not be a node of $C$, contradiction.
 \end{proof}
 
 \begin{remark}
  \label{rmk:acio-cio}
  Let $R \in \psset$, let $n$ be a $\cutpn$, $\bot$, $\otimes$, $\parr$, $\drpn$ or $\wn$ node of $\undgs{R}$, let $R_1, R_2 \in \psset$ such that $R$ has one of the shapes in \Cref{fig:sequential}, depending on the label of $n$, and let $k \coloneq 2$ if $n$ is a $\cutpn$ or $\otimes$ node, $k \coloneq 1$ otherwise. For every $R' \in \psset$, we have $R' \sqsubset R$ if and only if $R' \sqsubset R_i$~for~some $i \in \{1, \dots, k\}$. Thus, if $R$ is \pol,~the~following~conditions~are~equivalent:
  \begin{itemize}
   \item
    For every $R' \sqsubset R$, we have $R' \models \ACio$ (resp.~$\Cio$);
   \item
    For every $i \in \{1, \dots, k\}$ and for every $R' \sqsubset R_i$, we have $R' \models \ACio$ (resp.~$\Cio$).
  \end{itemize}
  
  \begin{figure}
   \centering
   \subcaptionbox{$l = \axpn$}{
    \centering
    \input{Figures/sequential-axiom.tikz}
   }
   \hskip 0.5em
   \subcaptionbox{\label{subfig:sequential-cut}$l = \cutpn$}{
    \centering
    \input{Figures/sequential-cut.tikz}
   }
   \hskip 0.5em
   \subcaptionbox{$l = \onepn$}[3.25em][c]{
    \centering
    \input{Figures/sequential-one.tikz}
   }
   \hskip 0.5em
   \subcaptionbox{$l = \bot$}{
    \centering
    \input{Figures/sequential-bottom.tikz}
   }
   \vskip 1em
   \subcaptionbox{\label{subfig:sequential-tensor}$l = \otimes$}{
    \centering
    \input{Figures/sequential-tensor.tikz}
   }
   \hskip 0.5em
   \subcaptionbox{\label{subfig:sequential-par}$l = \parr$}{
    \centering
    \input{Figures/sequential-par.tikz}
   }
   \hskip 0.5em
   \subcaptionbox{\label{subfig:sequential-bang}$l = \oc$}{
    \centering
    \input{Figures/sequential-bang.tikz}
   }
   \hskip 0.5em
   \subcaptionbox{$l = \drpn$}{
    \centering
    \input{Figures/sequential-dereliction.tikz}
   }
   \hskip 0.5em
   \subcaptionbox{$l = \wn$}{
    \centering
    \input{Figures/sequential-why-not.tikz}
   }
   \caption{\label{fig:sequential}Possible shapes of an $n$-\seq \ps, according to the label $l$ of $n$.}
  \end{figure}
 \end{remark}
 
 \begin{lemma}
  \label{lemma:ground-acio-cio}
  Let $R \in \psset$, let $n$ be a $\cutpn$, $\bot$, $\otimes$, $\parr$, $\drpn$ or $\wn$ node of $\undgs{R}$, let $R_1, R_2 \in \psset$ such that $R$ has one of the shapes in \Cref{fig:sequential}, depending on the label of $n$, and let $k \coloneq 2$ if $n$~is~a $\cutpn$ or $\otimes$ node, $k \coloneq 1$ otherwise. If $R$ is \pol, then the following statements hold:
  \begin{enumerate}[(i)]
   \item \label{itm:ground-acio}
    We have $\undgs{R} \models \ACio$ if and only if, for every $i \in \{1, \dots, k\}$, we have $\undgs{R_i} \models \ACio$;
   \item \label{itm:ground-cio}
    If $k = 1$, then $\undgs{R} \models \Cio$ if and only if $\undgs{R_1} \models \Cio$.
  \end{enumerate}
 \end{lemma}
  
 \begin{proof}
  \Cref{itm:ground-acio} is an immediate consequence of the fact that, for every directed cycle $\delta$, we have that $\delta$ is in $\iograph{\undgs{R}}$ if and only if there exists $i \in \{1, \dots, k\}$ such that $\delta$ is in $\iograph{\undgs{R_i}}$. As~for~\Cref{itm:ground-cio},~it is enough to prove the equality:
  \[
   \card{\drcl(\undgs{R})} + \card{\out(\undgs{R})} = \card{\drcl(\undgs{R_1})} + \card{\out(\undgs{R_1})}
  \]
  We consider two cases for $n$:
  \begin{itemize}
   \item
    \emph{$\bot$, $\parr$, intuitionistic $\drpn$ or $\wn$ node.} Then we know that $\drcl(\undgs{R}) = \drcl(\undgs{R_1})$. We observe that, if $n$ is a $\pout$ node, then $n$ has exactly one output conclusion and exactly one output premise, otherwise $n$ has~no~output conclusion and no output premise. Since $n$ has the same number of output conclusions and output premises, we have $\card{\out(\undgs{R})} = \card{\out(\undgs{R_1})}$,~and~so~we~are~done;
   \item
    \emph{Classical $\drpn$ node.} Then $\card{\drcl(\undgs{R})} = \card{\drcl(\undgs{R_1})} + 1$, because we have $\drcl(\undgs{R}) = \drcl(\undgs{R_1}) \cup \{n\}$ and $n \notin \drcl(\undgs{R_1})$. By definition of classical $\drpn$ node (\Cref{def:polarized}), we know that $n$ has no output conclusion and exactly one output premise, and thus $\card{\out(\undgs{R})} = \card{\out(\undgs{R_1})} - 1$. \qedhere
  \end{itemize}
 \end{proof}
  
 \begin{lemma}
  \label{lemma:acio-cio}
  Let $R \in \psset$, let $n$ be a $\cutpn$, $\bot$, $\otimes$, $\parr$, $\oc$, $\drpn$ or $\wn$ node of $\undgs{R}$, let $R_1, R_2 \in \psset$ such that $R$ has one of the shapes in \Cref{fig:sequential}, depending on the label of $n$, and let $k \coloneq 2$~if~$n$~is~a $\cutpn$ or $\otimes$ node, $k \coloneq 1$ otherwise. If $R$ is \pol, then the following statements hold:
  \begin{enumerate}[(i)]
   \item
    We have $R \models \ACio$ if and only if, for every $i \in \{1, \dots, k\}$, we have $R_i \models \ACio$;
   \item
    If $k = 1$, then $R \models \Cio$ if and only if $R_1 \models \Cio$.
  \end{enumerate}
 \end{lemma}
  
 \begin{proof}
  If $n$ is a $\oc$ node of $\undgs{R}$, then $k = 1$ and we are done, because $\undgs{R} \models \ACCio$ and the only $R' \in \psset$ such that $R' \sqsubset R$ is $R_1$. Otherwise, it suffices to apply \Cref{rmk:acio-cio}~and~\Cref{lemma:ground-acio-cio}.
 \end{proof}
 
 \poutterAccioSequential*
 
 \begin{proof}
  By induction on $\card{\ar(R)}$. First of all, suppose there exists a terminal $\bot$, $\parr$, $\drpn$ or $\wn$ node $n$ of  $\undgs{R}$. Then there exists $R_1 \in \psset$ such that $R$ has one of the shapes in \Cref{fig:sequential}, depending on the label of $n$. Clearly, the \ps $R_1$ is \poutter and, by \Cref{lemma:acio-cio}, we have $R_1 \models \ACCio$. Since $\card{\ar(R_1)} < \card{\ar(R_1)} + 1 \leq \card{\ar(R)}$, we can apply the induction hypothesis on $R_1$, from which we deduce that $R_1$ is \seq. We can then conclude that $R$ is $n$-\seq.
       
  We now suppose that $\undgs{R}$ has no terminal $\bot$, $\parr$, $\drpn$ or $\wn$ node, and we consider two cases:
  \begin{itemize}
   \item
    \emph{Every terminal node of $\undgs{R}$ is an $\axpn$, $\onepn$ or $\oc$ node.} By \Cref{lemma:existence-terminal-node} (and by \Cref{rmk:acio-connected-component}), every $C \in \cc(\undgs{R})$ is made up of exactly one terminal $\axpn$, $\onepn$ or $\oc$ node. In particular, we obtain that $\card{\drcl(\undgs{R})} = 0$. Since $R \models \ACCio$, we have $\card{\out(\undgs{R})} = \card{\drcl(\undgs{R})} + \card{\out(\undgs{R})} = 1$. We deduce that $\undgs{R}$ is a connected graph made up of exactly one terminal $\axpn$, $\onepn$ or $\oc$ node $n$. If $n$ is an $\axpn$ or $\onepn$ node, then we are done because $R$ is $n$-\seq. Otherwise, there exists $R_1 \in \psset$ such that $R$ has the shape in \Cref{subfig:sequential-bang}. Since $R$ is \poutter~and~$R_1 \sqsubset R$,~we~know~that~$R_1$~is \poutter and that $R_1 \models \ACCio$. Since $\card{\ar(R_1)} < \card{\ar(R_1)} + 1 \leq \card{\ar(R)}$, we can apply~the induction hypothesis on $R_1$. We conclude that $R_1$ is \seq, and thus $R$ is $n$-\seq;
   \item
    \emph{There exists a terminal $\cutpn$ or $\otimes$ node of $\undgs{R}$.} By \Cref{lemma:ioter-cut-tensor}, there exists a \ioter $\cutpn$ or $\otimes$ node $n$ of $\undgs{R}$. Since $n$ is a $\cutpn$ or $\otimes$ node, it has at least one output premise $a_1$. Let $a_2$ be the other premise of $n$, and let $G$ be the graph obtained from $\undgs{R}$ by adding two fresh $\bullet$ nodes $n_1$ and $n_2$, by setting the head of $a_i$ to be $n_i$ for each $i \in \{1, 2\}$,~and~by~erasing the connected component which contains $n$ from the resulting graph. Let $G_1$ be the~connected component of $G$ that contains $a_1$, and let $G_2$ be the graph obtained from $G$ by erasing $G_1$. We equip $G_1$ and $G_2$ with arbitrary total orderings on the premises of their $\bullet$ nodes, so that they become \gss. For each $i \in \{1, 2\}$, we consider the \ps $R_i$ obtained by taking $\undgs{R_i} \coloneq G_i$ and by defining $\undbox{R_i}$ as the restriction of $\undbox{R}$ to~$\bset(G_i)$. Since, by \Cref{prop:splitting}, no cycle of $\undgs{R}$ has $n$ among its nodes, the \ps $R$ has one of the shapes in \Cref{subfig:sequential-cut,subfig:sequential-tensor}, depending on the label of $n$. We obviously~have~that $R_1$ and $R_2$ are \poutter.
             
    We claim that $R_i \models \ACCio$ for each $i \in \{1, 2\}$. By \Cref{rmk:acio-cio} and \Cref{lemma:acio-cio}, it is~sufficient to prove that $G_i \models \Cio$. Since $a_1$ is output, by \Cref{itm:ioter-connected-component} of \Cref{lemma:ioter} the set $\drcl(G_1)$ is~empty and $\out(G_1) = \{a_1\}$, thus:
    \[
     \card{\drcl(G_1)} + \card{\out(G_1)} = 0 + 1 = 1
    \]
    This proves that $G_1 \models \Cio$. We now observe that, for every $n' \in \nd(\undgs{R})$, either $n' = n$, or $n'$ is the $\bullet$ node having the conclusion of $n$ as its premise, or there exists $i \in \{1, 2\}$ such that $n' \in \nd(G_i)$. It follows that $\drcl(G_2) = \drcl(\undgs{R})$. Also notice that $a_2$ is output if and only if $n$ is a $\tout$ node, so $\card{\out(G_2)} = \card{\out(\undgs{R})}$. From $\undgs{R} \models \Cio$ we then deduce~that~$G_2 \models \Cio$,~because:
    \[
     \card{\drcl(G_2)} + \card{\out(G_2)} = \card{\drcl(\undgs{R})} + \card{\out(\undgs{R})} = 1
    \]
             
    Finally, since $\card{\ar(R_i)} < \card{\ar(R_i)} + 1 \leq \card{\ar(R)}$ for each $i \in \{1, 2\}$, we can apply the induction hypothesis on $R_i$, from which we deduce that $R_i$ is \seq, thus $R$ is $n$-\seq. \qedhere
  \end{itemize}
 \end{proof}
     
\subsection{Complements of \texorpdfstring{\Cref{sec:typed-setting}}{Section 4}}
 \label{subsec:complements-typed-setting}
 
 \subsubsection{Definition of the desequentialization function}
  \label{subsec:desequentialization}
 
  \begin{definition}
   Let $\Sigma = A_1, \dots, A_k$ be a \sequent{\MELL}. For every \seqcptWithNameAndConclusion{$\pi$}{\MELL}{$\Sigma$}, we define $\Deseq{\pi} \in \psf{\MELL}$ with conclusions of types $A_1, \dots, A_k$,~by induction on $\pi$. Consider the last rule of $\pi$. Cases:
   \begin{itemize}
    \item
     \axsc. Then $\pi$ is only made of an \axsc rule 
     with $\Sigma = A, A^\bot$. We define $\Deseq{\pi}$ as in \Cref{subfig:desequentialization-axiom}; 
    \item
     \cutsc. Then $\Sigma = \Gamma, \Delta$ and $\pi$ is obtained by applying a \cutsc rule to \seqcpstWithNameAndConclusion{$\pi_1$ and $\pi_2$}{\MELL}{$\Gamma, A$ and $A^\bot, \Delta$} respectively. By induction hypothesis on $\pi_1$ and $\pi_2$, we have that $\Deseq{\pi_1}, \Deseq{\pi_2} \in \psf{\MELL}$ are defined. We define $\Deseq{\pi}$ as depicted in~\Cref{subfig:desequentialization-cut} with $R_i \coloneq \Deseq{\pi_i}$ for each $i \in \{1, 2\}$;
    \item
     $\one$. Then $\pi$ is only made of a $\one$ rule 
     with $\Sigma = \one$. We define $\Deseq{\pi}$ as in \Cref{subfig:desequentialization-one};
    \item
     $\bot$. Then $\Sigma = \Gamma, \bot$ and $\pi$ is obtained by applying a $\bot$ rule to a \seqcptWithNameAndConclusion{$\pi_1$}{\MELL}{$\Gamma$}. By induction hypothesis on $\pi_1$, we obtain that $\Deseq{\pi_1} \in \psf{\MELL}$~is defined. We define $\Deseq{\pi}$ as in \Cref{subfig:desequentialization-bottom} with $R_1 \coloneq \Deseq{\pi_1}$;
    \item
     $\otimes$. Then $\Sigma = \Gamma, A \otimes B, \Delta$ and $\pi$ is obtained by applying a $\otimes$ rule to \seqcpstWithNameAndConclusion{$\pi_1$ and $\pi_2$}{\MELL}{$\Gamma, A$ and $B, \Delta$} respectively. By induction~hypothesis on $\pi_1$ and on $\pi_2$, we obtain that $\Deseq{\pi_1}, \Deseq{\pi_2} \in \psf{\MELL}$ are defined. We define $\Deseq{\pi}$ as in \Cref{subfig:desequentialization-tensor} with $R_i \coloneq \Deseq{\pi_i}$ for each $i \in \{1, 2\}$;
    \item
     $\parr$. Then $\Sigma = \Gamma, A \parr B$ and $\pi$ is obtained by applying a $\parr$ rule to a \seqcptWithNameAndConclusion{$\pi_1$}{\MELL}{$\Gamma, A, B$}. We can apply the induction hypothesis on~$\pi_1$,~from~which it follows that $\Deseq{\pi_1} \in \psf{\MELL}$ is defined. We define~$\Deseq{\pi}$ as depicted in \Cref{subfig:desequentialization-par} with $R_1 \coloneq \Deseq{\pi_1}$;
    \item
     $\oc$. Then $\Sigma = \wn \Gamma, \oc A$ and $\pi$ is obtained by applying a $\oc$ rule to a \seqcptWithNameAndConclusion{$\pi_1$}{\MELL}{$\wn \Gamma, A$}. By induction hypothesis on $\pi_1$, we know that $\Deseq{\pi_1} \in \psf{\MELL}$~is defined. We define $\Deseq{\pi}$ as in \Cref{subfig:desequentialization-bang} with $R_1 \coloneq \Deseq{\pi_1}$;
    \item
     \emph{\drsc (resp.~$\wn$)}. Then $\Sigma = \Gamma, \wn A$ and $\pi$ is obtained by applying a \drsc (resp.~$\wn$) rule to a \seqcptWithName{$\pi_1$}{\MELL} with conclusion $\Gamma, A$ (resp.~$\Gamma, \wn A, \dots, \wn A$). We~can~apply the induction hypothesis on $\pi_1$, from which we deduce that $\Deseq{\pi_1} \in \psf{\MELL}$ is defined. We~define~$\Deseq{\pi}$ as in \Cref{subfig:desequentialization-dereliction} (resp.~\ref{subfig:desequentialization-why-not}) with $R_1 \coloneq \Deseq{\pi_1}$;
    \item
     \exch. Then $\Sigma = \Gamma, B, A, \Delta$ and $\pi$ is obtained by applying an \exch rule to a \seqcptWithNameAndConclusion{$\pi_1$}{\MELL}{$\Gamma, A, B, \Delta$}. By induction hypothesis on $\pi_1$, we obtain that $\Deseq{\pi_1} \in \psf{\MELL}$ is defined, with ordered conclusions $c_1, \dots, c_h, a, b, d_1, \dots, d_l$ such that $\Gamma$ is the sequence of the types of $c_1, \dots, c_h$, $A$ and $B$ are the types of $a$ and $b$ respectively, and $\Delta$ is the sequence of the types of $d_1, \dots, d_l$. We define $\Deseq{\pi}$ as the \pst{\MELL} obtained from $\Deseq{\pi_1}$ by just changing the total ordering of its conclusions in the~following~way: $c_1, \dots, c_h, b, a, d_1, \dots, d_l$.
   \end{itemize}
  \end{definition}
  
 \subsubsection{Embedding \texorpdfstring{\IMELL}{IMELL} into \texorpdfstring{\VMELL}{VMELL}}
  \label{subsec:embedding-imell-vmell}
  
  The Venn diagram in \Cref{fig:venn-diagram} graphically represents the set relations between the fragments $\LLpol$, \IMELL and \VMELL. Notice that \IMELL is not a fragment of \VMELL. In~this section, we present an embedding of \seqcpst{\IMELL} into \VMELL. In~particular,~the translation of the formulae preserves provability (\Cref{prop:IMELLtoVMELL}).
  
  \begin{figure}
   \centering
   \input{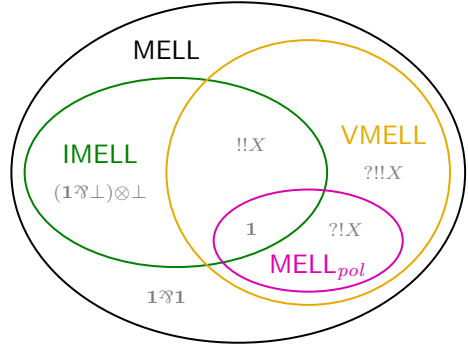}
   \caption{\label{fig:venn-diagram}Venn diagram of some fragments of \MELL we considered.}
  \end{figure}
  
  \begin{notation}
   We denote by $\outfragm{\IMELL}$ (resp.~$\inpfragm{\IMELL}$) the set of output (resp.~input) \formulae{\IMELL}, and by $\outfragm{\VMELL}$ (resp.~$\inpfragm{\VMELL}$) the set of output (resp.~input) \formulae{\VMELL}.
  \end{notation}
  
  \begin{definition}
   The function $\VMELLt{(\cdot)}$ from \IMELL to \MELL is defined inductively on $A \in \IMELL$:
   \begin{itemize}
    \item
     \emph{$A = X$.} Then $\VMELLt{A} \coloneq \oc X$;
    \item
     \emph{$A = X^\bot$.} Then $\VMELLt{A} \coloneq \wn X^\bot$;
    \item
     \emph{$A = \one$.} Then $\VMELLt{A} \coloneq \oc \one$;
    \item
     \emph{$A = \bot$.} Then $\VMELLt{A} \coloneq \wn \bot$;
    \item
     \emph{$A = O_1 \otimes O_2$ for some $O_1, O_2 \in \outfragm{\IMELL}$.} Then $\VMELLt{A} \coloneq \oc (\VMELLt{O_1} \otimes \VMELLt{O_2})$;
    \item
     \emph{$A = I \otimes O$ for some $I \in \inpfragm{\IMELL}$ and $O \in \outfragm{\IMELL}$.} Then $\VMELLt{A} \coloneq \wn (\VMELLt{I} \otimes \VMELLt{O})$;
    \item
     \emph{$A = O \otimes I$ for some $O \in \outfragm{\IMELL}$ and $I \in \inpfragm{\IMELL}$.} Then $\VMELLt{A} \coloneq \wn (\VMELLt{O} \otimes \VMELLt{I})$;
    \item
     \emph{$A = I_1 \parr I_2$ for some $I_1, I_2 \in \inpfragm{\IMELL}$.} Then $\VMELLt{A} \coloneq \wn (\VMELLt{I_1} \parr \VMELLt{I_2})$;
    \item
     \emph{$A = O \parr I$ for some $O \in \outfragm{\IMELL}$ and $I \in \inpfragm{\IMELL}$.} Then $\VMELLt{A} \coloneq \oc (\VMELLt{O} \parr \VMELLt{I})$;
    \item
     \emph{$A = I \parr O$ for some $I \in \inpfragm{\IMELL}$ and $O \in \outfragm{\IMELL}$.} Then $\VMELLt{A} \coloneq \oc (\VMELLt{I} \parr \VMELLt{O})$;
    \item
     \emph{$A = \oc O$ for some $O \in \outfragm{\IMELL}$.} Then $\VMELLt{A} \coloneq \VMELLt{O}$;
    \item
     \emph{$A = \wn I$ for some $I \in \inpfragm{\IMELL}$.} Then $\VMELLt{A} \coloneq \VMELLt{I}$.
   \end{itemize}
   We extend $\VMELLt{(\cdot)}$ to \sequents{\IMELL} by defining $\VMELLt{\Gamma} \coloneqq \VMELLt{A_1}, \dots, \VMELLt{A_k}$ for every $\Gamma = A_1, \dots, A_k$.
  \end{definition}
  
  The following remarks express simple properties of the translation that can be checked~by~a straightforward proof by induction on $A$.
  
  \begin{remark}
   \label{rmk:vmell-pos-neg}
   Let $A \in \IMELL$. Then:
   \begin{enumerate}[(i)]
    \item 
     If $A$ is output, we have $\VMELLt{A} = \oc O$ for some $O \in \outfragm{\VMELL}$;
    \item
     If $A$ is input, we have $\VMELLt{A} = \wn I$ for some $I \in \inpfragm{\VMELL}$.
   \end{enumerate}
   In particular, we have $\VMELLt{A} \in \VMELL$.
  \end{remark}
  
  \begin{remark}
   \label{rmk:vmell-not}
   For every formula $A \in \IMELL$, we have $\VMELLt{(A^\bot)} = (\VMELLt{A})^\bot$.
  \end{remark}
  
  We recall the following result about provability in \IMELL.
  
  \begin{proposition}[\cite{lamarche2008proof}]
   \label{prop:imell-one-output}
   Let $\Gamma$ be a \sequent{\IMELL}. If $\Gamma$ is provable in \IMELL, then~exactly one formula of $\Gamma$ is output.
  \end{proposition}

  \begin{definition}
   If $\pi$ is a \seqcptWithConclusion{\IMELL}{$\Sigma$}, then we define a \seqcptWithNameAndConclusion{$\VMELLt{\pi}$}{\VMELL}{$\VMELLt{\Sigma}$}, by induction on $\pi$. Consider~the~last rule of $\pi$. Cases:
   \begin{itemize}
    \item
     \axsc. Then $\pi$ is only made of an \axsc rule, and we have that $\Sigma = A, A^\bot$. We define $\VMELLt{\pi}$ as the \seqcp made of an \axsc rule with conclusion $\VMELLt{\Sigma}$ (\Cref{rmk:vmell-not}), in \Cref{subfig:vmell-axiom};
    \item
     \cutsc. Then $\Sigma = \Gamma, \Delta$ and $\pi$ is obtained by applying a \cutsc rule to two \seqcpstWithNameAndConclusion{$\pi_1$ and $\pi_2$}{\IMELL}{$\Gamma, A$ and $A^\bot, \Delta$ respectively}. By induction hypothesis on $\pi_1$ and on $\pi_2$, the \seqcpstWithNameAndConclusion{$\VMELLt{\pi_1}$ and $\VMELLt{\pi_2}$}{\VMELL}{$\VMELLt{\Gamma}, \VMELLt{A}$ and $(\VMELLt{A})^\bot, \VMELLt{\Delta}$ (\Cref{rmk:vmell-not})} respectively are defined. We define $\VMELLt{\pi}$~by~applying a \cutsc~rule~to $\VMELLt{\pi_1}$ and $\VMELLt{\pi_2}$, as shown in \Cref{subfig:vmell-cut};
    \item
     $\one$. Then $\pi$ is only made of a $\one$ rule, and we have that $\Sigma = \one$. We define $\VMELLt{\pi}$ as the \seqcp made of a $\one$ rule and a $\oc$ rule, in \Cref{subfig:vmell-one};
    \item
     $\bot$. Then $\Sigma = \Gamma, \bot$ and $\pi$ is obtained by applying a $\bot$ rule to a \seqcptWithNameAndConclusion{$\pi_1$}{\IMELL}{$\Gamma$}. By induction hypothesis on $\pi_1$, the \seqcptWithNameAndConclusion{$\VMELLt{\pi_1}$}{\VMELL}{$\VMELLt{\Gamma}$} is defined. We define $\VMELLt{\pi}$ by applying a $\bot$ rule and~a~\drsc~rule~to $\VMELLt{\pi_1}$, as shown in \Cref{subfig:vmell-bottom};
    \item
     $\otimes$. Then $\Sigma = \Gamma, A \otimes B, \Delta$ and $\pi$ is obtained by applying a $\otimes$ rule to \seqcpstWithNameAndConclusion{$\pi_1$ and $\pi_2$}{\IMELL}{$\Gamma, A$ and $B, \Delta$ respectively}. By induction hypothesis on $\pi_1$ and on $\pi_2$, the \seqcpstWithNameAndConclusion{$\VMELLt{\pi_1}$ and $\VMELLt{\pi_2}$}{\VMELL}{$\VMELLt{\Gamma}, \VMELLt{A}$ and $\VMELLt{B}, \VMELLt{\Delta}$ respectively} are defined. Since $A \otimes B \in \IMELL$, at least one of the formulae $A$ and $B$ is output. We can suppose $A$ is output without loss of generality. By \Cref{prop:imell-one-output}, every formula of $\Gamma$ is input. By \Cref{rmk:vmell-pos-neg}, every formula of $\VMELLt{\Gamma}$ has the shape $\wn I$ for some $I \in \inpfragm{\VMELL}$. Now suppose that $\Delta = A_1, \dots, A_k$ and~distinguish~the~two~following~subcases:
     \begin{itemize}
      \item
       \emph{$B$ is output.} Then, for every $i \in \{1, \dots, k\}$, the formula $A_i$ is input, thus, by \Cref{rmk:vmell-pos-neg}, there exists $I_i \in \inpfragm{\VMELL}$ such that $\VMELLt{A_i} = \wn I_i$. We define $\VMELLt{\pi}$ by applying a $\otimes$ rule to $\VMELLt{\pi_1}$ and $\VMELLt{\pi_2}$, then an $\exch$ rule $k$ times and a $\oc$ rule, as shown in \Cref{subfig:vmell-tensor-output};
      \item
       \emph{$B$ is input.} We define $\VMELLt{\pi}$ by applying a $\otimes$ rule to $\VMELLt{\pi_1}$ and $\VMELLt{\pi_2}$, then an $\exch$ rule $k$~times~and a \drsc rule, as shown in \Cref{subfig:vmell-tensor-input};
     \end{itemize}
    \item
     $\parr$. Then $\Sigma = \Gamma, A \parr B$ and $\pi$ is obtained by applying a $\parr$ rule to a \seqcptWithNameAndConclusion{$\pi_1$}{\IMELL}{$\Gamma, A, B$}. By induction hypothesis on $\pi_1$, the \seqcptWithNameAndConclusion{$\VMELLt{\pi_1}$}{\VMELL}{$\VMELLt{\Gamma}, \VMELLt{A}, \VMELLt{B}$} is defined. Since $A \parr B \in \IMELL$, at least one of the formulae $A$ and $B$ is input. We can suppose $A$ is input without loss~of~generality. We distinguish the two following subcases:
     \begin{itemize}
      \item
       \emph{$B$ is output.} By \Cref{prop:imell-one-output}, every formula of $\Gamma$ is input. Thus, by \Cref{rmk:vmell-pos-neg}, every formula of $\VMELLt{\Gamma}$ has the shape $\wn I$ for some $I \in \inpfragm{\VMELL}$. We define $\VMELLt{\pi}$ by applying~a~$\parr$~rule and a $\oc$ rule to $\VMELLt{\pi_1}$, as shown in \Cref{subfig:vmell-par-output};
      \item
       \emph{$B$ is input.} We define $\VMELLt{\pi}$ by applying a $\parr$ rule and a \drsc rule to $\VMELLt{\pi_1}$, as in \Cref{subfig:vmell-par-input};
     \end{itemize}
    \item
     $\oc$. Then $\Sigma = \wn \Gamma, \oc A$ and $\pi$ is obtained by applying a $\oc$ rule to a \seqcptWithNameAndConclusion{$\pi_1$}{\IMELL}{$\wn \Gamma, A$}. By induction hypothesis on $\pi_1$, the \seqcptWithNameAndConclusion{$\VMELLt{\pi_1}$}{\VMELL}{$\VMELLt{( \wn \Gamma)}, \VMELLt{A}$} is defined. Since $\VMELLt{(\oc A)} = \VMELLt{A}$ (and $\VMELLt{(\wn \Gamma)} = \VMELLt{\Gamma}$),~we just define $\VMELLt{\pi} \coloneq \VMELLt{\pi_1}$, as depicted in \Cref{subfig:vmell-bang-dereliction};
    \item
     \emph{\drsc (resp.~$\wn$).} Then $\Sigma = \Gamma, \wn A$ and $\pi$ is obtained by applying a \drsc (resp.~$\wn$) rule to a \seqcptWithName{$\pi_1$}{\IMELL} with conclusion $\Gamma, A$ (resp.~$\Gamma, \wn A, \dots, \wn A$). By~induction hypothesis on $\pi_1$, we can consider the \seqcptWithName{$\VMELLt{\pi_1}$}{\VMELL} with conclusion $\VMELLt{\Gamma}, \VMELLt{A}$ (resp.~$\VMELLt{\Gamma}, \VMELLt{(\wn A)}, \dots, \VMELLt{(\wn A)}$). We have that $\VMELLt{(\wn A)} = \VMELLt{A}$ and, since $\wn A \in \IMELL$,~the formula $A$ is input, thus $\VMELLt{A} = \wn I$ for some $I \in \inpfragm{\VMELL}$, by \Cref{rmk:vmell-pos-neg}. We define $\VMELLt{\pi}$~as~in \Cref{subfig:vmell-bang-dereliction} (resp.~\ref{subfig:vmell-why-not});
    \item
     \exch. Then $\Sigma = \Gamma, B, A, \Delta$ and $\pi$ is obtained by applying an \exch rule to a \seqcptWithNameAndConclusion{$\pi_1$}{\IMELL}{$\Gamma, A, B, \Delta$}. By induction hypothesis on $\pi_1$, the \seqcptWithNameAndConclusion{$\VMELLt{\pi_1}$}{\VMELL}{$\VMELLt{\Gamma}, \VMELLt{A}, \VMELLt{B}, \VMELLt{\Delta}$} is defined. We define $\VMELLt{\pi}$ by applying an \exch rule to $\VMELLt{\pi_1}$, as shown in \Cref{subfig:vmell-exchange}.
   \end{itemize}
      
   \begin{figure}
    \centering
    \subcaptionbox{\label{subfig:vmell-axiom}\axsc rule.}{
     \centering
     \begin{prooftree}
      \hypo{\infcstrut}
      \infer1[\scriptsize\axsc{}]{\infcstrut \vdash \VMELLt{A},(\VMELLt{A})^\bot}
     \end{prooftree}
    }
    \quad
    \subcaptionbox{\label{subfig:vmell-cut}\cutsc rule.}{
     \centering
     \begin{prooftree}
      \hypo{}
      \ellipsis{$\VMELLt{\pi_1}$}{\infcstrut \vdash \VMELLt{\Gamma}, \VMELLt{A}}
      \hypo{}
      \ellipsis{$\VMELLt{\pi_2}$}{\infcstrut \vdash (\VMELLt{A})^\bot, \VMELLt{\Delta}}
      \infer2[\scriptsize\cutsc{}]{\infcstrut \vdash \VMELLt{\Gamma}, \VMELLt{\Delta}}
     \end{prooftree}
    }
    \quad
    \subcaptionbox{\label{subfig:vmell-exchange}\exch rule.}{
     \centering
     \begin{prooftree}
      \hypo{}
      \ellipsis{$\VMELLt{\pi_1}$}{\infcstrut \vdash \VMELLt{\Gamma}, \VMELLt{A}, \VMELLt{B}, \VMELLt{\Delta}}
      \infer1[\scriptsize\exch{}]{\infcstrut \vdash \VMELLt{\Gamma}, \VMELLt{B}, \VMELLt{A}, \VMELLt{\Delta}}
     \end{prooftree}
    }
    \vskip 1em
    \subcaptionbox{\label{subfig:vmell-one}$\one$ rule.}[3.875em][c]{
     \centering
     \begin{prooftree}
      \hypo{}
      \infer1[\scriptsize$\one$]{\infcstrut \vdash \one}
      \infer1[\scriptsize$\oc$]{\infcstrut \vdash \oc \one}
     \end{prooftree}
    }
    \quad
    \subcaptionbox{\label{subfig:vmell-bottom}$\bot$ rule.}{
     \centering
     \begin{prooftree}
      \hypo{}
      \ellipsis{$\VMELLt{\pi_1}$}{\infcstrut \vdash \VMELLt{\Gamma}}
      \infer1[\scriptsize$\bot$]{\infcstrut \vdash \VMELLt{\Gamma}, \bot}
      \infer1[\scriptsize\drsc{}]{\infcstrut \vdash \VMELLt{\Gamma}, \wn \bot}
     \end{prooftree}
    }
    \quad
    \subcaptionbox{$\otimes$ rule, with $A$ output.}{
     \setcounter{subsubfigure}{0}
     \subsubcaptionbox{\label{subfig:vmell-tensor-output}$B$ output.}{
      \begin{prooftree}
       \hypo{}
       \ellipsis{$\VMELLt{\pi_1}$}{\infcstrut \vdash \VMELLt{\Gamma}, \VMELLt{A}}
       \hypo{}
       \ellipsis{$\VMELLt{\pi_2}$}{\infcstrut \vdash \VMELLt{B}, \VMELLt{\Delta}}
       \infer2[\scriptsize$\otimes$]{\infcstrut \vdash \VMELLt{\Gamma}, \VMELLt{A} \otimes \VMELLt{B}, \VMELLt{\Delta}}
       \infer[double]1[\scriptsize \exch]{\infcstrut \vdash \VMELLt{\Gamma}, \VMELLt{\Delta}, \VMELLt{A} \otimes \VMELLt{B}}
       \infer1[\scriptsize$\oc$]{\vdash \VMELLt{\Gamma}, \VMELLt{\Delta}, \oc (\VMELLt{A} \otimes \VMELLt{B})}
      \end{prooftree}
     }
     \hskip 0.5em
     \subsubcaptionbox{\label{subfig:vmell-tensor-input}$B$ input.}{
      \begin{prooftree}
       \hypo{}
       \ellipsis{$\VMELLt{\pi_1}$}{\infcstrut \vdash \VMELLt{\Gamma}, \VMELLt{A}}
       \hypo{}
       \ellipsis{$\VMELLt{\pi_2}$}{\infcstrut \vdash \VMELLt{B}, \VMELLt{\Delta}}
       \infer2[\scriptsize$\otimes$]{\infcstrut \vdash \VMELLt{\Gamma}, \VMELLt{A} \otimes \VMELLt{B}, \VMELLt{\Delta}}
       \infer[double]1[\scriptsize \exch]{\infcstrut \vdash \VMELLt{\Gamma}, \VMELLt{\Delta}, \VMELLt{A} \otimes \VMELLt{B}}
       \infer1[\scriptsize\drsc{}]{\vdash \VMELLt{\Gamma}, \VMELLt{\Delta}, \wn (\VMELLt{A} \otimes \VMELLt{B})}
      \end{prooftree}
     }
    }
    \vskip 1em
    \subcaptionbox{$\parr$ rule, with $A$ input.}{
     \setcounter{subsubfigure}{0}
     \subsubcaptionbox{\label{subfig:vmell-par-output}$B$ output.}{
      \begin{prooftree}
       \hypo{}
       \ellipsis{$\VMELLt{\pi_1}$}{\infcstrut \vdash \VMELLt{\Gamma}, \VMELLt{A}, \VMELLt{B}}
       \infer1[\scriptsize$\parr$]{\infcstrut \vdash \VMELLt{\Gamma}, \VMELLt{A} \parr \VMELLt{B}}
       \infer1[\scriptsize$\oc$]{\infcstrut \vdash \VMELLt{\Gamma}, \oc (\VMELLt{A} \parr \VMELLt{B})}
      \end{prooftree}
     }
     \hskip 0.5em
     \subsubcaptionbox{\label{subfig:vmell-par-input}$B$ input.}{
      \begin{prooftree}
       \hypo{}
       \ellipsis{$\VMELLt{\pi_1}$}{\infcstrut \vdash \VMELLt{\Gamma}, \VMELLt{A}, \VMELLt{B}}
       \infer1[\scriptsize$\parr$]{\infcstrut \vdash \VMELLt{\Gamma}, \VMELLt{A} \parr \VMELLt{B}}
       \infer1[\scriptsize\drsc{}]{\infcstrut \vdash \VMELLt{\Gamma}, \wn (\VMELLt{A} \parr \VMELLt{B})}
      \end{prooftree}
     }
    }
    \quad
    \subcaptionbox{\label{subfig:vmell-bang-dereliction}$\oc$ and \drsc rules.}[6.75em][c]{
     \centering
     \begin{prooftree}
      \hypo{}
      \ellipsis{$\VMELLt{\pi_1}$}{\infcstrut \vdash \VMELLt{\Gamma}, \VMELLt{A}}
     \end{prooftree}
    }
    \quad
    \subcaptionbox{\label{subfig:vmell-why-not}$\wn$ rule.}{
     \centering
     \begin{prooftree}
      \hypo{}
      \ellipsis{$\VMELLt{\pi_1}$}{\infcstrut \vdash \VMELLt{\Gamma}, \VMELLt{A}, \dots, \VMELLt{A}}
      \infer1[\scriptsize$\wn$]{\infcstrut \vdash \VMELLt{\Gamma}, \VMELLt{A}}
     \end{prooftree}
    }
    \caption{Translation $\VMELLt{(\cdot)}$ of \seqcpst{\IMELL} into \VMELL.}
   \end{figure}
  \end{definition}
  
  The following result expresses a straightforward consequence of the previous definition.
  
  \begin{proposition}
   \label{prop:IMELLtoVMELL}
   Let $\Gamma$ be a \sequent{\IMELL}. If $\Gamma$ is \provable{\IMELL}, then $\VMELLt{\Gamma}$~is~\provable{\VMELL}.
  \end{proposition}
  
  \begin{remark}
   The converse of \Cref{prop:IMELLtoVMELL} is false. Indeed, we may consider, for example, the sequent $\Gamma = \one, X^\bot$, which is not \provable{\IMELL} even though $\VMELLt{\Gamma} = \oc \one, \wn X^\bot$ is \provable{\VMELL}: \Cref{fig:counterexample-converse-provability} shows a \seqcpt{\VMELL} with conclusion $\VMELLt{\Gamma}$.
   
   \begin{figure}
    \centering
    \begin{prooftree}
     \hypo{}
     \infer1[\scriptsize$\one$]{\infcstrut \vdash \one}
     \infer1[\scriptsize$\oc$]{\infcstrut \vdash \oc \one}
     \infer1[\scriptsize$\wn$]{\infcstrut \vdash \oc \one, \wn X^\bot}
    \end{prooftree}
    \caption{\label{fig:counterexample-converse-provability}A proof in \VMELL with conclusion a sequent $\VMELLt{\Gamma}$ such that $\Gamma$ is not \provable{\IMELL}.}
   \end{figure}
  \end{remark}
  
\subsection{Complements of \texorpdfstring{\Cref{sec:bang-calculus}}{Section 5}}

 \subsubsection{Proof of \texorpdfstring{\Cref{thm:factorization}}{Theorem 45} (Factorization)}
 \label{subsec:factorization}
 
  \begin{notation}
   The set of \lcontexts (resp.~\bcontexts) is denoted by $\sctxset$ (resp.~$\bctxset$).
  \end{notation}
 
  \begin{definition}
   We define a partial function $\ltyp$ from $\sctxset \times \anlset$ to $\STLC$ by induction on~$M \in \anlset$. For every $\Gamma \in \sctxset$:
   \begin{itemize}
    \item
     \emph{Case $M = x$ for some $x \in \vset$.}
          
     If $x \in \dom(\Gamma)$, then $\ltyp(\Gamma, M) \coloneq \Gamma(x)$, otherwise $(\Gamma, M) \notin \dom(\ltyp)$;
    \item
     \emph{Case $M = \labs{x}{\tau}{N}$ for some $x \in \vset$, $\tau \in \stset$ and $N \in \anlset$.\footnote{\label{foot:variable-convention}We can suppose $x \notin \dom(\Gamma)$ by Barendregt's variable convention (\cite[Convention~2.1.13]{barendregt1984lambda}).}}
          
     If $((\Gamma, \oftype{x}{\tau}), N) \in \dom(\ltyp)$ then $\ltyp(\Gamma, M) \coloneq \intuArrow{\tau}{\ltyp((\Gamma, \oftype{x}{\tau}), N)}$, else $(\Gamma, M) \notin \dom(\ltyp)$;
    \item
     \emph{Case $M = \lapp{M_1}{M_2}$ for some $M_1, M_2 \in \anlset$.}
          
     If $(\Gamma, M_1), (\Gamma, M_2) \in \dom(\ltyp)$ and $\ltyp(\Gamma, M_1) = \intuArrow{\ltyp(\Gamma, M_2)}{\tau}$, then we take $\ltyp(\Gamma, M) \coloneq \tau$, otherwise $(\Gamma, M) \notin \dom(\ltyp)$.
   \end{itemize}
  \end{definition}
  
  \begin{definition}
   We define a partial function $\bctyp$ from $\bctxset \times \anbcset$ to $\bangMELL$ by induction~on~$t \in \anbcset$. For every $\Gamma \in \bctxset$:
   \begin{itemize}
    \item
     \emph{Case $t = x$ for some $x \in \vset$.}
          
     If $x \in \dom(\Gamma)$ and $\Gamma(x) = \oc O$, then $\bctyp(\Gamma, t) \coloneq O$, otherwise $(\Gamma, t) \notin \dom(\bctyp)$;
    \item
     \emph{Case $t = \bcabs{x}{\oc O}{s}$ for some $x \in \vset$, $O$ output \formula{\bangMELL} and $s \in \anbcset$.\textsuperscript{\normalfont\ref{foot:variable-convention}}}
          
     If $((\Gamma, \oftype{x}{\oc O}), s) \in \dom(\bctyp)$, we set $\bctyp(\Gamma, t) \coloneq \linArrow{\oc O}{\bctyp((\Gamma, \oftype{x}{\oc O}), s)}$, else $(\Gamma, t) \notin \dom(\bctyp)$;
    \item
     \emph{Case $t = \bcapp{t_1}{t_2}$ for some $t_1, t_2 \in \anbcset$.}
          
     If $(\Gamma, t_1), (\Gamma, t_2) \in \dom(\bctyp)$ and $\bctyp(\Gamma, t_1) = \linArrow{\bctyp(\Gamma, t_2)}{O}$, we take $\bctyp(\Gamma, t) \coloneq O$, otherwise $(\Gamma, t) \notin \dom(\bctyp)$;
    \item
     \emph{Case $t = \bcder{s}$ for some $s \in \anbcset$.}
          
     If $(\Gamma, s) \in \dom(\bctyp)$ and $\bctyp(\Gamma, s) = \oc O$, then $\bctyp(\Gamma, t) \coloneq O$, otherwise $(\Gamma, t) \notin \dom(\bctyp)$;
    \item
     \emph{Case $t = \bcbang{s}$ for some $s \in \anbcset$.}
          
     If $(\Gamma, s) \in \dom(\bctyp)$, then $\bctyp(\Gamma, t) \coloneq \oc \bctyp(\Gamma, s)$, otherwise $(\Gamma, t) \notin \dom(\bctyp)$.
   \end{itemize}
  \end{definition}
 
  The following result expresses the compatibility of our Church-style typing functions with~the \cbn and \cbv translations. Its proof is a very easy induction on $M$.
  
  \begin{proposition}
   \label{prop:girard-translations}
   Let $M \in \anlset$. For every $\Gamma \in \sctxset$, the following are equivalent:
   \begin{enumerate}[(a)]
    \item \label{itm:lambda-typed}
     $(\Gamma, M) \in \dom(\ltyp)$;
    \item \label{itm:cbn-typed}
     $(\CBNT{\Gamma}, \CBNt{M}) \in \dom(\bctyp)$;
    \item \label{itm:cbv-typed}
     $(\CBVT{\Gamma}, \CBVt{M}) \in \dom(\bctyp)$.
   \end{enumerate}
   If these equivalent conditions hold, then $\bctyp(\CBNT{\Gamma}, \CBNt{M}) = \CBNT{\ltyp(\Gamma, M)}$ and $\bctyp(\CBVT{\Gamma}, \CBVt{M}) = \oc \CBVT{\ltyp(\Gamma, M)}$.
  \end{proposition}
  
  To give a rigorous proof of \Cref{thm:factorization}, we redefine the translations into \pns~by~using modular graph transformations as in \cite{danos1990logique}.
  


  \begin{definition}
   Let $R$ be a \nam \pn.
   \begin{itemize}
    \item
     If $a$ is an input conclusion of $R$, we denote by $\pabs{a}{R}$ the \nam \pn in \Cref{subfig:sequential-par},~where the premises of the $\parr$ node $n$ are $a$ and the output conclusion of $R$. If $R$~is~injectively~named, we may write $\pabs{\nmfun{R}(a)}{R} \coloneq \pabs{a}{R}$;
    \item
     We denote by $\deradd{R}$ the \nam \pn obtained by adding, for every input conclusion $a$ of $R$ that is not a conclusion of a $\oc$, $\drpn$ or $\wn$ node, a fresh terminal $\drpn$ node having~$a$~as~its~premise, and by erasing the $\bullet$ node which previously had $a$ as its premise.
   \end{itemize}
  \end{definition}

  \begin{definition}
   Let $R_1, R_2$ be \nam \pns, and let $O$ be the type of~the~output~conclusion $a_0$ of $R_2$.
   \begin{itemize}
    \item
     For every input conclusion $a$ of $R_1$ having type $O^\bot$, we denote by $\cut{a}{R_1}{R_2}$ the~\nam \pn in \Cref{subfig:sequential-cut}, where the premises of the $\cutpn$ node $n$ are $a$ and $a_0$. Moreover,~if~$R_1$ is \injnam, we may write $\cut{\nmfun{R_1}(a)}{R_1}{R_2} \coloneq \cut{a}{R_1}{R_2}$;
    \item
     If the output conclusion of $R_1$ is a conclusion of an $\axpn$ node of type $O^\bot \parr O_0$, we denote~by $\graft{R_1}{R_2}$ the \nam \pn in \Cref{subfig:named-graft}.
   \end{itemize}
      
   \begin{figure}
    \centering
    \subcaptionbox{$R_1$}{
     \centering
     \input{Figures/named-graft-1.tikz}
    }    
    \quad
    \subcaptionbox{$R_2$}{
     \centering
     \input{Figures/named-graft-2.tikz}
    }
    \quad
    \subcaptionbox{\label{subfig:named-graft}$\graft{R_1}{R_2}$}{
     \centering
     \input{Figures/named-graft.tikz}
    }
    \caption{A graph transformation of \nam \pns.}
   \end{figure}
  \end{definition}
%
  \begin{remark}
   \label{rmk:gua}
   Let $R_1$ be a \gua \pn. By \Cref{rmk:bangMELL-output-conclusion}, there exists \emph{exactly} one~\pntWithName{$R$}{\bangMELL} having the shape in \Cref{subfig:sequential-bang}, and $R$ is \gua too.
  \end{remark}
  
  \begin{notation}
   Let $R \in \psset$. We denote by $\norm{\wn}{R}$ the \rnf{\wn} of $R$.
  \end{notation}
  
  \begin{definition}
   Let $R_1$ be a \gua \pn. If $R$ is the \gua \pn in \Cref{subfig:sequential-bang},~we write $\enbox{R_1} \coloneq \norm{\wn}{R}$.
  \end{definition}

  \begin{remark}
   Let $R$ and $R'$ be \pnst{\bangMELL}. If $R \town R'$ and $R$ is \gua,~then~so~is~$R'$.
  \end{remark}

  \begin{remark}
   Let $R$ be a \gua \pn. Then $\enbox{R}$ is \gua. Moreover, if $R$ is~(injectively) \nam, then $\enbox{R}$ inherits from $R$ a structure of (injectively) \nam \pn.
  \end{remark}
 
  \begin{definition}
   Let $R$ be a \nam \gua \pn. We set $\share{R} \coloneq \norm{\rwn}{R'}$,~where $R'$ is the \nam \pn that is obtained from $R$ by applying the following local~transformation to $\undgs{R}$: for every $x \in \vset$, if the preimage $\nmfun{R}^{-1}(x)$ is non-empty, we add a fresh terminal $\wn$ node $n$ having the arcs in $\nmfun{R}^{-1}(x)$ as premises, we erase the $\bullet$ nodes which previously had~the~arcs~in $\nmfun{R}^{-1}(x)$ as premises and, if $a$ is the conclusion of $n$, we define $\nmfun{R'}(a) \coloneq x$.
  \end{definition}

  \begin{remark}
   Let $R$ be a \nam \gua \pn. Then $\share{R} \in \bpnset$.
  \end{remark}

  We now associate with every \bctx a simple \pn that is used as an elementary building block in the translations of the term calculi.

  \begin{definition}
   Let $\Gamma \coloneq \oftype{x_1}{\oc O_1}, \dots, \oftype{x_k}{\oc O_k}, \oftype{x}{\oc O}$ be a \bctx. We write $\initVar{x}{\Gamma}$ for the \injnam \pn depicted in \Cref{fig:init}. If $k = 0$, we can write $\init{\Gamma} \coloneq \initVar{x}{\Gamma}$.
      
   \begin{figure}
    \centering
    \input{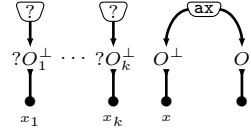}
    \caption{\label{fig:init}The elementary building block of the translations.}
   \end{figure}
  \end{definition}

  \begin{remark}
   Let $\Gamma \in \bctxset$, and let $x \in \dom(\Gamma)$. Then $\deradd{\initVar{x}{\Gamma}} \in \bpnset$.
  \end{remark}

  The following definition expresses in this language the translations of the \cbn~and \cbv \lci into \pns (\Cref{fig:cbn-pns,fig:cbv-pns}).
  
  \begin{figure}
   \centering
   \subcaptionbox{$\CBNp{((\Gamma, \oftype{x}{\tau}), x)}$ with~$O \coloneq \CBNT{\tau}$.}[12em][c]{
    \centering
    \input{Figures/cbn-pn-variable.tikz}
   }
   \quad
   \subcaptionbox{$\CBNp{(\Gamma, \labs{x}{\tau}{N})}$ with $O \coloneq \CBNT{\tau}$.}[11em][c]{
    \centering
    \input{Figures/cbn-pn-abstraction.tikz}
   }
   \vskip 1em
   \subcaptionbox{$\CBNp{(\Gamma, \lapp{M_1}{M_2})}$ is the \rnf{\rwn} of the \pn in the picture.}[26em][c]{
    \centering
    \input{Figures/cbn-pn-application.tikz}
   }
   \caption{\label{fig:cbn-pns}The translation $\CBNp{(\cdot)}$ of the \cbn \lc into \pns.}
  \end{figure}
  
  \begin{figure}
   \centering
   \subcaptionbox{$\CBVp{((\Gamma, \oftype{x}{\tau}), x)}$ with $O \coloneq \CBVT{\tau}$.}{
    \centering
    \input{Figures/cbv-pn-variable.tikz}
   }
   \quad
   \subcaptionbox{$\CBVp{(\Gamma, \labs{x}{\tau}{N})}$ is the \rnf{\wn} of the \pn in the picture with $O \coloneq \CBVT{\tau}$.}[15.5em][c]{
    \centering
    \input{Figures/cbv-pn-abstraction.tikz}
   }
   \vskip 1em
   \subcaptionbox{$\CBVp{(\Gamma, \lapp{M_1}{M_2})}$ is the \rnf{\rwn} of the \pn in the picture.}[26em][c]{
    \centering
    \input{Figures/cbv-pn-application.tikz}
   }
   \caption{\label{fig:cbv-pns}The translation $\CBVp{(\cdot)}$ of the \cbv \lc into \pns.}
  \end{figure}
  
  \begin{definition}
   We define two functions $\CBNp{(\cdot)}$ and $\CBVp{(\cdot)}$ from $\dom(\ltyp)$ to $\bpnset$ such that, for~every $(\Gamma, M) \in \dom(\ltyp)$ with $\Gamma = \oftype{x_1}{\tau_1}, \dots, \oftype{x_k}{\tau_k}$, we have that $\CBNp{(\Gamma, M)}$ (resp.~$\CBVp{(\Gamma, M)}$) has~$k + 1$ conclusions $a_0, \dots, a_k$, the type of $a_0$ is $\CBNT{\ltyp(\Gamma, M)}$ (resp.~$\oc \CBVT{\ltyp(\Gamma, M)}$) and, for all~$i \in \{1, \dots, k\}$, the type of $a_i$ is $\wn O_i^\bot$ with $O_i \coloneq \CBNT{\tau_i}$ (resp.~$O_i \coloneq \CBVT{\tau_i}$), and we have $\nmfun{R}(a_i) = x_i$. The definition~is by induction on $M \in \anlset$. For every $\Gamma \in \sctxset$ such that $(\Gamma, M) \in \dom(\ltyp)$:
   \begin{itemize}
    \item
     \emph{Case $M = x$ for some $x \in \vset$.} We define:
     \[
      \CBNp{(\Gamma, M)} \coloneq \deradd{\initVar{x}{\CBNT{\Gamma}}} \qquad \CBVp{(\Gamma, M)} \coloneq \enbox{(\deradd{\initVar{x}{\CBVT{\Gamma}}})}
     \]
    \item
     \emph{Case $M = \labs{x}{\tau}{N}$ for some $x \in \vset$, $\tau \in \stset$ and $N \in \anlset$.\textsuperscript{\normalfont\ref{foot:variable-convention}}} We define:
     \[
      \CBNp{(\Gamma, M)} \coloneq \pabs{x}{\CBNp{((\Gamma, \oftype{x}{\tau}), N)}} \qquad \CBVp{(\Gamma, M)} \coloneq \enbox{(\pabs{x}{\CBVp{((\Gamma, \oftype{x}{\tau}), N)}})}
     \]
    \item
     \emph{Case $M = \lapp{M_1}{M_2}$ for some $M_1, M_2 \in \anlset$.} We take $x \in \vset \setminus \dom(\Gamma)$, and we define:
     \begin{gather*}
      \CBNp{(\Gamma, M)} \coloneq \share{\cut{x}{\graft{\init{\oftype{x}{\oc \CBNT{\ltyp(\Gamma, M_1)}}}}{\enbox{(\CBNp{(\Gamma, M_2)})}}}{\CBNp{(\Gamma, M_1)}}} \\
      \CBVp{(\Gamma, M)} \coloneq \share{\cut{x}{\graft{\deradd{\init{\oftype{x}{\oc \CBVT{\ltyp(\Gamma, M_1)}}}}}{\CBVp{(\Gamma, M_2)}}}{\CBVp{(\Gamma, M_1)}}}
     \end{gather*}
   \end{itemize}
  \end{definition}

  We give an analogous definition to express the translation of the typed \bc into \pns by graph transformations (\Cref{fig:bang-pns}).
  
  \begin{definition}
   We give a function $\WH{(\cdot)}$ from $\dom(\bctyp)$ to $\bpnset$ such that, for all $(\Gamma, t) \in \dom(\bctyp)$, if $\Gamma = \oftype{x_1}{\oc O_1}, \dots, \oftype{x_k}{\oc O_k}$, then $\WH{(\Gamma, t)}$ has exactly $k + 1$ conclusions $a_0, \dots, a_k$, the type of $a_0$ is $\bctyp(\Gamma, t)$ and, for all $i \in \{1, \dots, k\}$, the type of $a_i$ is $\wn O_i^\bot$, and we have $\nmfun{R}(a_i) = x_i$. The definition is by induction on $t \in \anbcset$. For every $\Gamma \in \bctxset$ such that $(\Gamma, t) \in \dom(\bctyp)$:
   \begin{itemize}
    \item
     \emph{Case $t = x$ for some $x \in \vset$.} We define:
     \[
      \WH{(\Gamma, t)} \coloneq \deradd{\initVar{x}{\Gamma}}
     \]
    \item
     \emph{Case $t = \bcabs{x}{\oc O}{s}$ for some $x \in \vset$, $O$ output \formula{\bangMELL} and $s \in \anbcset$.\textsuperscript{\normalfont\ref{foot:variable-convention}}} We define:
     \[
      \WH{(\Gamma, t)} \coloneq \pabs{x}{\WH{((\Gamma, \oftype{x}{\oc O}), s)}}
     \]
    \item
     \emph{Case $t = \bcapp{t_1}{t_2}$ for some $t_1, t_2 \in \anbcset$.} We take $x \in \vset \setminus \dom(\Gamma)$, and we define:
     \[
      \WH{(\Gamma, t)} \coloneq \share{\cut{x}{\graft{\init{\oftype{x}{\oc \bctyp(\Gamma, t_1)}}}{\WH{(\Gamma, t_2)}}}{\WH{(\Gamma, t_1)}}}
     \]
    \item
     \emph{Case $t = \bcder{s}$ for some $s \in \anbcset$.} We take $x \in \vset \setminus \dom(\Gamma)$, and we define:
     \[
      \WH{(\Gamma, t)} \coloneq \cut{x}{\deradd{\init{\oftype{x}{\bctyp(\Gamma, s)}}}}{\WH{(\Gamma, s)}}
     \]
    \item
     \emph{Case $t = \bcbang{s}$ for some $s \in \anbcset$.} We define:
     \[
      \WH{(\Gamma, t)} \coloneq \enbox{(\WH{(\Gamma, s)})}
     \]
   \end{itemize}
  \end{definition}

  \begin{remark}
   \label{rmk:cbv-translation-reduction}
   Let $R_1$ and $R_2$ be \injnam \pns, and let $O$ be the type of the output conclusion of $R_2$. If the output conclusion of $R_1$ has type $\oc (\linArrow{O}{O_0})$ then,~for~every~fresh $x, y \in \vset$, we have the following reduction (\Cref{fig:cbv-translation-reduction}):
   \begin{gather*}
    \cut{x}{\graft{\init{\oftype{x}{\oc (\linArrow{O}{O_0})}}}{R_2}}{\cut{y}{\deradd{\init{\oftype{y}{\oc (\linArrow{O}{O_0})}}}}{R_1}} \\
    \mapstor{\rax} \cut{x}{\graft{\deradd{\init{\oftype{x}{\oc (\linArrow{O}{O_0})}}}}{R_2}}{R_1}
   \end{gather*}
      
   \begin{figure}
    \centering
    \input{Figures/cbv-translation-reduction-1.tikz}
    \caption{\label{fig:cbv-translation-reduction}The key step in the proof of \Cref{thm:factorization}.}
   \end{figure}
  \end{remark}
  
  We now prove a technical result about the interaction between the reductions $\tor{\rax}$~and~$\town$.
  
  \begin{definition}[by induction on $R$]
   Let $R \in \psset$. We say that $R$ is \emph{\wnsafe} if no premise of a $\cutpn$ node of $\undgs{R}$ is the conclusion of a $\wn$ node and, for every $R' \sqsubset R$, we have that~$R'$~is~\wnsafe.
  \end{definition}
  
  \begin{lemma}
   \label{lemma:ax-wn}
   Let $R, R' \in \psset$ such that $R$ is \wnsafe and $R \tor{\rax} R'$. Then:
   \begin{enumerate}[(i)]
    \item
     We have that $R'$ is \wnsafe;
    \item
     For every $R_0' \in \psset$, if $R' \town R_0'$, there exists $R_0 \in \psset$ such that $R \town R_0 \tor{\rax} R_0'$;\label{itm:ax-wn-swap}
    \item
     We have that $\norm{\rwn}{R} \tor{\rax} \norm{\rwn}{R'}$.\label{itm:ax-wn-normal}
   \end{enumerate}
  \end{lemma}
  
  
  \begin{remark}
   A counterexample to \Cref{itm:ax-wn-swap,itm:ax-wn-normal} of \Cref{lemma:ax-wn} in the absence of the~\wnsafe hypothesis is illustrated in \Cref{fig:counterexample-ax-wn}.
   
   \begin{figure}
    \centering
    \input{Figures/counterexample-ax-wn.tikz}
    \caption{\label{fig:counterexample-ax-wn}For the \ps $R$ on the left, we have that $\norm{\rwn}{R} = R \mapstor{\rax} R' \neq \norm{\rwn}{R'}$.}
   \end{figure}
  \end{remark}
  
  \factorization*
  
  \begin{proof}
   By induction on $M$. We consider the three possible cases for $M$:
   \begin{itemize}
    \item
     \emph{Case $M = x$ for some $x \in \vset$.} We have:
     \begin{gather*}
      \WH{(\CBNT{\Gamma}, \CBNt{M})} = \WH{(\CBNT{\Gamma}, x)} = \deradd{\initVar{x}{\CBNT{\Gamma}}} = \CBNp{(\Gamma, M)} \\ \WH{(\CBVT{\Gamma}, \CBVt{M})} = \WH{(\CBVT{\Gamma}, \bcbang{x})} = \enbox{(\WH{(\CBVT{\Gamma}, x)})} = \enbox{(\deradd{\initVar{x}{\CBVT{\Gamma}}})} = \CBVp{(\Gamma, M)}
     \end{gather*}
    \item
     \emph{Case $M = \labs{x}{\tau}{N}$ for some $x \in \vset$, $\tau \in \stset$ and $N \in \anlset$.\textsuperscript{\normalfont\ref{foot:variable-convention}}} By induction hypothesis on $N$,~and by \Cref{itm:ax-wn-normal} of \Cref{lemma:ax-wn}, we have:
     \begin{align*}
      \WH{(\CBNT{\Gamma}, \CBNt{M})} & = \WH{(\CBNT{\Gamma}, \bcabs{x}{\oc \CBNT{\tau}}{\CBNt{N}})} & \WH{(\CBVT{\Gamma}, \CBVt{M})} & = \WH{(\CBVT{\Gamma}, \bcbang{(\bcabs{x}{\oc \CBVT{\tau}}{\CBVt{N}})})} \\
      & = \pabs{x}{\WH{((\CBNT{\Gamma}, \oftype{x}{\oc \CBNT{\tau}}), \CBNt{N})}} && = \enbox{(\WH{(\CBVT{\Gamma}, \smash{\bcabs{x}{\oc \CBVT{\tau}}{\CBVt{N}}})})} \\
      & = \pabs{x}{\CBNp{((\Gamma, \oftype{x}{\tau}), N)}} && = \enbox{(\pabs{x}{\WH{((\CBVT{\Gamma}, \oftype{x}{\oc \CBVT{\tau}}), \CBVt{N})}})} \\
      & = \CBNp{(\Gamma, M)} && \tor{\rax}^* \enbox{(\pabs{x}{\CBVp{((\Gamma, \oftype{x}{\tau}), N)}})} \\
      &&& = \CBVp{(\Gamma, M)}
     \end{align*}
    \item
     \emph{Case $M = \lapp{M_1}{M_2}$ for some $M_1, M_2 \in \anlset$.} Let $x \in \vset \setminus \dom(\Gamma)$. By \Cref{prop:girard-translations},~and~by induction hypothesis on $M_1$ and on $M_2$, we have:
     \begin{align*}
      \WH{(\CBNT{\Gamma}, \CBNt{M})} & = \WH{(\CBNT{\Gamma}, \bcapp{\CBNt{M_1}}{\bcbang{(\CBNt{M_2})}})} \\
      & = \share{\cut{x}{\graft{\init{\oftype{x}{\oc \bctyp(\CBNT{\Gamma}, \CBNt{M_1})}}}{\WH{(\CBNT{\Gamma}, \bcbang{(\CBNt{M_2})})}}}{\WH{(\CBNT{\Gamma}, \CBNt{M_1})}}} \\
      & = \share{\cut{x}{\graft{\init{\oftype{x}{\oc \bctyp(\CBNT{\Gamma}, \CBNt{M_1})}}}{\enbox{(\WH{(\CBNT{\Gamma}, \CBNt{M_2})})}}}{\WH{(\CBNT{\Gamma}, \CBNt{M_1})}}} \\
      & = \share{\cut{x}{\graft{\init{\oftype{x}{\oc \CBNT{\ltyp(\Gamma, M_1)}}}}{\enbox{(\CBNp{(\Gamma, M_2)})}}}{\CBNp{(\Gamma, M_1)}}} \\
      & = \CBNp{(\Gamma, M)}
     \end{align*}
     Now let $y \in \vset \setminus \dom(\Gamma)$ with $y \neq x$. By \Cref{prop:girard-translations}, \Cref{rmk:cbv-translation-reduction} (with $O \coloneq \bctyp(\CBVT{\Gamma}, \CBVt{M_2})$, $O_0 \coloneq \bctyp(\CBVT{\Gamma}, \CBVt{M})$ and $R_i \coloneq \WH{(\CBVT{\Gamma}, \CBVt{M_i})}$ for each $i \in \{1, 2\}$), and by induction hypothesis~on $M_1$ and on $M_2$, we obtain:
     \begin{align*}
      & \hspace{-0.0625em} \cut{x}{\graft{\init{\oftype{x}{\oc \bctyp(\CBVT{\Gamma}, \bcder{\CBVt{M_1}})}}}{\WH{(\CBVT{\Gamma}, \CBVt{M_2})}}}{\WH{(\CBVT{\Gamma}, \bcder{\CBVt{M_1}})}} \\
      & \hspace{-0.0625em} = \cut{x}{\graft{\init{\oftype{x}{\bctyp(\CBVT{\Gamma}, \CBVt{M_1})}}}{\WH{(\CBVT{\Gamma}, \CBVt{M_2})}}}{\cut{y}{\deradd{\init{\oftype{y}{\bctyp(\CBVT{\Gamma}, \CBVt{M_1})}}}}{\WH{(\CBVT{\Gamma}, \CBVt{M_1})}}} \\
      & \hspace{-0.0625em} \mapstor{\rax} \cut{x}{\graft{\deradd{\init{\oftype{x}{\bctyp(\CBVT{\Gamma}, \CBVt{M_1})}}}}{\WH{(\CBVT{\Gamma}, \CBVt{M_2})}}}{\WH{(\CBVT{\Gamma}, \CBVt{M_1})}} \\
      & \hspace{-0.0625em} \tor{\rax}^* \cut{x}{\graft{\deradd{\init{\oftype{x}{\oc \CBVT{\ltyp(\Gamma, M_1)}}}}}{\CBVp{(\Gamma, M_2)}}}{\CBVp{(\Gamma, M_1)}}
     \end{align*}
     To conclude, we apply \Cref{itm:ax-wn-normal} of \Cref{lemma:ax-wn}:
     \begin{align*}
      \WH{(\CBVT{\Gamma}, \CBVt{M})} & = \WH{(\CBVT{\Gamma}, \bcapp{\bcder{\CBVt{M_1}}}{\CBVt{M_2}})} \\
      & = \share{\cut{x}{\graft{\init{\oftype{x}{\oc \bctyp(\CBVT{\Gamma}, \bcder{\CBVt{M_1}})}}}{\WH{(\CBVT{\Gamma}, \CBVt{M_2})}}}{\WH{(\CBVT{\Gamma}, \bcder{\CBVt{M_1}})}}} \\
      & \tor{\rax}^* \share{\cut{x}{\graft{\deradd{\init{\oftype{x}{\oc \CBVT{\ltyp(\Gamma, M_1)}}}}}{\CBVp{(\Gamma, M_2)}}}{\CBVp{(\Gamma, M_1)}}} \\
      & = \CBVp{(\Gamma, M)} \qedhere
     \end{align*}
   \end{itemize}
  \end{proof}
  
%
%
%
%

 \subsubsection{Proof of \texorpdfstring{\Cref{thm:characterization}}{Theorem 44} (Characterization of the \texorpdfstring{\bc}{bang calculus})}
 \label{subsec:characterization}
 
  We start this section by proving some purely geometric properties. The following key result~is essentially a variant of~\cite[Proposition 2(1)]{maieli2022proof}.
  
  \begin{lemma}
   \label{lemma:sequential}
   Let $R \in \psset$ be $n$-\seq with $n$ a $\cutpn$, $\bot$, $\otimes$, $\parr$, $\drpn$ or $\wn$ node. Let $k \coloneq 2$ if $n$ is a $\cutpn$ or $\otimes$ node, $k \coloneq 1$ otherwise, and let $R_1, \dots, R_k \in \psset$ such that $R$ has one of the shapes in \Cref{fig:sequential}, depending on the label of $n$. Let $m$ be a terminal node of $\undgs{R}$ and a node of $\undgs{R_i}$ for some $i \in \{1, \dots, k\}$. Finally, assume that either $n$ is~not~a~$\bot$~or~$\wkpn$~node,~or~$m$~is~not~an~$\axpn$,~$\onepn$ or $\oc$ node, and suppose either $n$ is not a $\parr$ or $\ctpn$ node, or $m$ is not a $\cutpn$ or $\otimes$ node. If $R_i$ is $m$-\seq, then $R$ is $m$-\seq.
  \end{lemma}
  
  \begin{proof}
   \begin{figure}
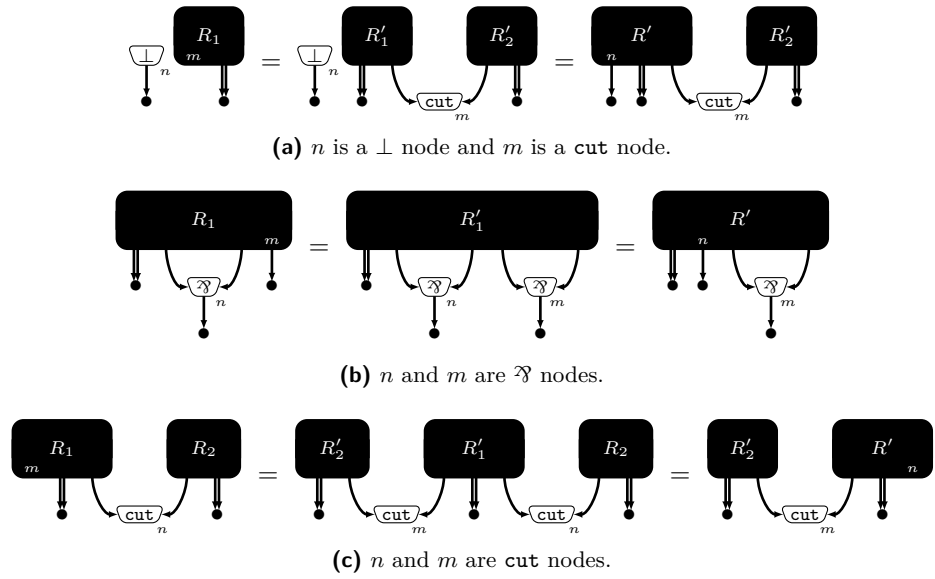

    \centering
    \begin{subfigure}{0.9875\textwidth}
     \centering
     \input{Figures/sequential-proof-bot.tikz}
     \caption{$n$ is a $\bot$ node and $m$ is a $\cutpn$ node.}
    \end{subfigure}
    \vskip 1em
    \begin{subfigure}{0.9875\textwidth}
     \centering
     \input{Figures/sequential-proof-par.tikz}
     \caption{$n$ and $m$ are $\parr$ nodes.}
    \end{subfigure}
    \vskip 1em
    \begin{subfigure}{0.9875\textwidth}
     \centering
     \input{Figures/sequential-proof-cut.tikz}
     \caption{$n$ and $m$ are $\cutpn$ nodes.}
    \end{subfigure}
    \caption{\label{fig:sequential-proof}A few particular and interesting cases of the proof of \Cref{lemma:sequential}.}
   \end{figure}
      
   We can suppose $i \coloneq 1$ without loss of generality. Since $m$ is a terminal node of~$\undgs{R}$,~we know that $m$ is not a $\bullet$ node, and thus $m$ has no premise of $n$ among its conclusions.
      
   We claim that $m$ is not an $\axpn$, $\onepn$ or $\oc$ node. If $n$ is a $\bot$ or $\wkpn$ node, this is~known~by~hypothesis. Otherwise, we have that $n$ is a $\cutpn$, $\otimes$, $\parr$, $\drpn$ or $\wn$ node with at least one premise, and~then~$\undgs{R_1}$ contains not only $m$, but also a node having a premise of $n$ among its conclusions. Since~$R_1$~is $m$-\seq, we conclude that $m$ cannot be an $\axpn$, $\onepn$ or $\oc$ node.
      
   We have then established that $m$ is a $\cutpn$, $\bot$, $\otimes$, $\parr$, $\drpn$ or $\wn$ node. Let $h \coloneq 2$ if $m$ is a $\cutpn$ or $\otimes$ node, $h \coloneq 1$ otherwise. Since $R_1$ is $m$-\seq, there exist sequential $R_1', \dots R_h' \in \psset$ such that $R_1$ has one of the shapes in \Cref{fig:sequential} (in the figure, $R_j'$ replaces~$R_j$~for~each~$j \in \{1, \dots, h\}$, and $m$ replaces $n$), depending on the label of $m$. Since no premise of $n$ is a conclusion of $m$, we deduce that, for every premise $a$ of $n$ in $\undgs{R_1}$, there exists $j \in \{1, \dots, h\}$ such that~$a$~is~in~$\undgs{R_j'}$.
      
   We claim that we can swap these quantifiers: there exists $j \in \{1, \dots, h\}$ such that every premise of $n$ in $\undgs{R_1}$ is in $\undgs{R_j'}$. This is trivial if $n$ is a $\bot$ or $\wkpn$ node, or if $h = 1$. We then assume that $n$ is a $\cutpn$, $\otimes$, $\parr$, $\drpn$ or $\wn$ node with at least one premise, and that $h = 2$. Then $m$ is a $\cutpn$ or $\otimes$ node. By hypothesis, the node $n$ is not a $\parr$ or $\ctpn$ node. In other words, we have that $n$ is either a $\cutpn$, $\otimes$ or $\drpn$ node, or a $\wn$ node with exactly one premise. In any of these cases, exactly one premise $a$ of $n$ is in $\undgs{R_1}$, and we~already~know~that~there~exists~$j \in \{1, \dots, h\}$~such~that~$a$~is in $\undgs{R_j'}$.
      
   We can now suppose, without loss of generality, that every premise of $n$ in $\undgs{R_1}$ is in $\undgs{R_1'}$. Consider $R' \in \psset$ obtained by replacing $R_1$ with $R_1'$ in $R$. Since $R_1', R_2, \dots, R_k$ are \seq, we obtain that $R'$ is $n$-\seq.  Finally, since $R', R_2', \dots, R_h'$ are \seq,~we~can conclude that $R$ is $m$-\seq (see \Cref{fig:sequential-proof}).
  \end{proof}
  
  \begin{example}
   Consider the $n$-\seq \pss in \Cref{fig:sequential-examples}. \Cref{lemma:sequential} can be applied to $R_\cutpn$. In this case, we have that $n$ is a $\cutpn$ node and $m$ is a $\bot$ node, and~\Cref{lemma:sequential} states that $R_\cutpn$ is $m$-\seq. On the other hand, the \pss $R_\rob$ and $R_\rtp$ are not $m$-\seq. These key examples (and obvious variants) show that \Cref{lemma:sequential} cannot be extended to the cases in which $n$ is a $\bot$ or $\wkpn$ node and $m$~is~an~$\axpn$,~$\onepn$~or~$\oc$~node,~or~$n$~is~a~$\parr$~or~$\ctpn$ node and $m$ is a $\cutpn$ or $\otimes$ node.
      
   \begin{figure}
    \centering
    \subcaptionbox{$R_\cutpn$}{
     \centering
     \input{Figures/sequential-uniqueness.tikz}
    }
    \hskip 0.5em
    \subcaptionbox{$R_\rob$}{
     \centering
     \input{Figures/eta-unit-untyped.tikz}
    }
    \hskip 0.5em
    \subcaptionbox{$R_\rtp$}{
     \centering
     \input{Figures/eta-multiplicative-untyped.tikz}
    }
    \caption{\label{fig:sequential-examples}Three $n$-\seq \pss. The only $m$-\seq \ps is $R_\cutpn$.}
   \end{figure}
  \end{example}
 
  The following result intuitively means that a $\bot$, $\parr$, \drsc or $\wn$ rule in a \seqcp $\pi$ can always be applied last if the formula it introduces appears in the conclusion~sequent~of~$\pi$.
 
  \begin{lemma}
   \label{lemma:sequential-last-rule}
   Let $R \in \psset$ and let $n$ be a terminal $\bot$, $\parr$, $\drpn$ or $\wn$ node of $\undgs{R}$. If $R$ is~\seq, then $R$ is $n$-\seq.
  \end{lemma}
 
  \begin{proof}
   By induction on $\card{\ar(\undgs{R})}$. Since $R$ is \seq, there exists $m \in \nd(\undgs{R})$ such that $R$ is $m$-\seq. If $m = n$, we are done, so we suppose $m \neq n$. Since $m$ and $n$ are two~distinct nodes of $\undgs{R}$, we know that $m$ is not an $\axpn$, $\onepn$ or $\oc$ node. In other words, we have that $m$ is a $\cutpn$, $\bot$, $\otimes$, $\parr$, $\drpn$ or $\wn$ node. From the fact that $R$ is $m$-\seq we deduce that $R$ has one of the shapes in \Cref{fig:sequential} (in the picture, $m$ replaces $n$), depending on the label~of~$m$,~for~some \seq $R_1, R_2 \in \psset$. Since $m \neq n$, we can suppose without loss of generality that $n$ is a node of $\undgs{R_1}$. Since $\card{\ar(\undgs{R_1})} < \card{\ar(\undgs{R_1})} + 1 \leq \card{\ar(\undgs{R})}$, we can apply the induction hypothesis on $R_1$, from which we deduce that $R_1$ is $n$-\seq. By~\Cref{lemma:sequential}~(with~the~roles~of~$m$~and~$n$ swapped), we conclude that $R$ is $n$-\seq.
  \end{proof}
  
  We extend \Cref{not:io-order} straightforwardly to \pol \pss.
  
  
  \begin{notation}
   Let $R \in \ppset$. We use the notation ${\precio} \coloneq {\precg{\iograph{\undgs{R}}}}$ when there is no ambiguity.
  \end{notation}
  
  \begin{remark}
   \label{rmk:io-order-pss}
   If $R \in \ppset$ and $\undgs{R} \models \ACio$ (in particular, by \Cref{thm:poutter-accio-sequential}, if $R$ is \poutter and \seq), then $\precio$ is a strict order relation (by \Cref{rmk:io-order}), and thus it is well-founded.
  \end{remark}
  
  \begin{lemma}
   \label{lemma:bang-minimal}
   Let $R \in \psset$ be \poutter and \seq, let $k \geq 0$ and $n_0, \dots, n_k \in  \nd(\undgs{R})$. Let $\{\parr, \wn\} \subseteq L \subseteq \{\axpn, \bot, \parr, \oc, \drpn, \wn\}$, suppose that $\undgs{R}$ has no $\pout$ node and that:
   \begin{enumerate}[(i)]
    \item \label{itm:bang-minimal}
     The node $n_0$ is minimal in $\{n \colon \text{$n$ is a node of $\undgs{R}$ with label in $L$}\}$ with respect to~$\precio$;
    \item \label{itm:bang-minimal-input}
     For every $i \in \{1, \dots, k\}$, an input conclusion of $n_{i - 1}$ is a premise of $n_i$;
    \item \label{itm:bang-minimal-tensor}
     For every $i \in \{1, \dots, k - 1\}$, we have that $n_i$ is a $\tin$ or $\drpn$ node;
    \item \label{itm:bang-minimal-cut}
     Either $n_k$ is a $\cutpn$ node, or a conclusion of $n_k$ is a conclusion of $R$.
   \end{enumerate}
   Then there exist \seq $R_0, \dots, R_k \in \psset$ such that $\undgs{R_i}$ has no node with label in $L$ for every $i \in \{1, \dots, k\}$, and $R$ has one of the shapes in \Cref{fig:bang-minimal}~(in~the~figures,~none~among~$n_1, \dots, n_k$ is a $\drpn$ node for simplicity), depending on the label of $n_0$, with $R_0$ defined by erasing $n_1, \dots, n_k$ and $R_1, \dots, R_k$ from $R$, and by setting the conclusion of $n_0$ to be a premise of a fresh $\bullet$~node~if $k \geq 1$. Moreover, if $\drpn \in L$, then $n_i$ is a $\tin$ node for every $i \in \{1, \dots, k - 1\}$.
   
   \begin{figure}
    \centering
    \subcaptionbox{$l = \axpn$}{
    \centering
    \input{Figures/bang-minimal-ax.tikz}
    }
    \quad
    \subcaptionbox{\label{subfig:bang-minimal}$l \in \{\bot, \parr, \drpn, \wn\}$}{
     \centering
     \input{Figures/bang-minimal.tikz}
    }
    \vskip 1em
    \subcaptionbox{$l = \oc$}{
     \centering
     \input{Figures/bang-minimal-box.tikz}
    }
    \caption{\label{fig:bang-minimal}Decomposing a \poutter \seq \ps $R$ such that $\undgs{R}$ has no $\pout$ node, using a node $n_0$ with label $l$ and minimal with respect to $\precio$. In the pictures, we have $l' \in \{\cutpn, \otimes\}$. For simplicity, the output premise of every $\cutpn$ or $\otimes$ node is its left premise, and~$\drpn$~nodes~are~omitted.}
   \end{figure}
  \end{lemma}
  
  \begin{proof}
   By induction on $k$. The base case ($k = 0$) is obvious, so we assume that $k \geq 1$. For each $i \in \{1, \dots, k\}$, let $a_i$ be the input conclusion of $n_{i - 1}$ which is a premise of $n_i$. By \Cref{thm:poutter-accio-sequential}, we have $R \models \ACCio$. Since $\undgs{R} \models \ACio$, the nodes $n_0, \dots, n_k$ are pairwise distinct,~and we can consider the path $\delta \coloneq n_k a_k n_{k - 1} a_{k - 1} \dots n_1 a_1 n_0$,~that~is~a~directed~path~of~$\smash{\iograph{\undgs{R}}}$. By \Cref{itm:bang-minimal}, the node $n_k$ is not a $\parr$ or $\wn$ node. Since the premise $a_k$ of $n_k$ is input (and by \Cref{itm:bang-minimal-cut}), we know that $n_k$ is a $\cutpn$, $\tin$ or $\drpn$ node. Now let $R' \in \psset$ obtained from~$R$~by~setting~$a_k$~to~be~the~premise~of a fresh $\bullet$ node, and the input premise of $n_k$ to be the~conclusion of a fresh $\bot$ node $n_k'$. We~have $R' \models \ACCio$, and thus $R'$ is \seq by \Cref{thm:poutter-accio-sequential}. We observe that \Cref{itm:bang-minimal} holds for $R'$, because every (directed) path of $\smash{\iograph{\undgs{R'}}}$ between two nodes of $\undgs{R}$ is actually a (directed)~path~of $\iograph{\undgs{R}}$, and the only directed path of $\iograph{\undgs{R}}$ from the $\bot$ node $n_k'$ is the empty one. By applying the induction hypothesis on $R'$, we then obtain sequential $R_0', \dots, R_{k-1}' \in \psset$ such that $\undgs{R_i'}$ has no node with label in $L$ for all $i \in \{1, \dots, k - 1\}$, and $R'$ has one of~the~shapes~in~\Cref{fig:bang-minimal} (in the figures, $k - 1$ replaces $k$, $R_i'$ replaces $R_i$ for all $i \in \{1, \dots, k - 1\}$, and none among~$n_1, \dots, n_{k - 1}$ is a $\drpn$ node for simplicity), depending on the label of $n_0$, with $R_0'$ obtained from $R$ by erasing $n_1, \dots, n_{k - 1}$ and $R_1', \dots, R_{k - 1}'$, and by setting the conclusion of $n_0$ to be a premise of a fresh $\bullet$ node if $k \geq 2$. Since $n_k \in \nd(\undgs{R'})$, and since we have $n_k \neq n_i$ for each $i \in \{0, \dots, k - 1\}$,~there exists $h \in \{0, \dots, k - 1\}$ such that $n_k \in \nd(\undgs{R_h'})$, and either $h = 0$ or $n_h$ is a $\tin$ node. For~each $i \in \{0, \dots k - 1\}$ with $i \neq h$, let $R_i \coloneq R_i'$. If $n_k$ is a $\drpn$ node, we are done by taking $R_h \coloneq R_h'$ and $R_k \in \psset$ without nodes. From now on, we can then suppose that $n_k$ is a $\cutpn$ or $\otimes$ node. Let $a_k'$ be the output premise of $n_k$, and let $G_k'$ be the connected component of~$\undgs{R'}$~containing~$n_k$.
   
   We claim that $G_k'$ has no node with label in $L$, except possibly $n_k'$. To prove this,~we~assume that there exists a path $\gamma$ of $\undgs{R'}$ from a node $n$ with label in $L$ to $n_k$ and show that this~leads us to a contradiction. First of all, we observe that $n \neq n_k$, because $L \cap \{\cutpn, \otimes\} = \emptyset$,~and~thus $\gamma$ is non-empty. Moreover, by the assumption that $n \neq n_k'$, the last arc of $\gamma$ is $a_k'$. We~can~also assume, without loss of generality, that $n$ is the only node of $\gamma$ with label in $L$. Since $\undgs{R'}$~has no $\pout$ node, we obtain that $\gamma$ is a path of $\smash{\iograph{\undgs{R'}}}$. Moreover, since $\{\parr, \wn\} \subseteq L$, and~since~at~most~one premise $a$ of $n$ is in $\gamma$, we have that $\gamma$ is also a path of the graph~$G$~obtained~from~$\smash{\iograph{\undgs{R'}}}$~by~erasing every premise $a' \neq a$ of every $\pin$ or $\wn$ node of $\undgs{R'}$. By \Cref{itm:io-graph-outdegree} of \Cref{rmk:io-graph}, we know that $G$ is a \pfg. And since the last arc of $\gamma$ is $a_k'$, by \Cref{prop:pfg-directed-path}~we~obtain~that~$\gamma$~is a directed path of $\smash{\iograph{\undgs{R'}}}$. We now distinguish the two following cases:
   \begin{itemize}
    \item
     If there exists $i \in \{0, \dots, k - 1\}$ such that $n_i$ is a node of $\gamma$, we can consider the maximum $i$ having that property, thus $n_i \gamma \, n_k \delta \, n_i$ is a directed cycle of $\iograph{\undgs{R}}$, contradicting~$\undgs{R} \models \ACio$;
    \item
     If, for every $i \in \{0, \dots, k - 1\}$, we have that $n_i$ is not a node of $\gamma$, then $n \, \gamma \, n_k \delta$ is a~directed path of $\iograph{\undgs{R}}$, contradicting \Cref{itm:bang-minimal} for $R$.
   \end{itemize}
   
   We now observe that every non-empty path $\gamma$ of $G_k'$ from a node $n \neq n_k'$ to $n_k$ is a~directed path of $\smash{\iograph{\undgs{R'}}}$. Indeed, since $n \neq n_k'$, the last arc of $\gamma$ is $a_k'$. And since no node of $\gamma$ has label in $L$ and $\{\parr, \wn\} \subseteq L$, we deduce that $\gamma$ is a path of the graph $G$ we obtain from $\smash{\iograph{\undgs{R'}}}$ by erasing every premise of every $\pin$ or $\wn$ node of $\undgs{R'}$. By \Cref{prop:pfg-directed-path}~and~\Cref{itm:io-graph-outdegree}~of~\Cref{rmk:io-graph}~we~have~that~$\gamma$ is a directed path of $\smash{\iograph{\undgs{R'}}}$.
   
   We deduce that $n_h$ is not a node of $G_k'$: if there were a path $\gamma$ of $\undgs{R'}$ from $n_h$ to $n_k$, then $\gamma$ would be a directed path of $\smash{\iograph{\undgs{R'}}}$, hence $\gamma \, n_k \delta \, n_h$ would be a directed cycle~of~$\smash{\iograph{\undgs{R}}}$,~contradicting $\undgs{R} \models \ACio$. We also have that $n_k'$ is the only sink of $\smash{\iograph{\undgs{R'}}}$ in $G_k'$. Thus, by \Cref{itm:io-graph-sinks} of \Cref{rmk:io-graph}, we obtain that $\drcl(G_k') = \out(G_k') = \emptyset$.
   
   Let $G_h$ be the graph we obtain from $\undgs{R_h'}$ by erasing $G_k'$, let $G_k$ be the graph we obtain from $G_k'$ by erasing $n_k$ and $n_k'$, and by setting $a_k'$ to be a premise of a fresh $\bullet$ node~and,~for~$i \in \{h, k\}$, let $R_i \in \psset$ obtained by setting $\undgs{R_i} \coloneq G_i$ and by defining $\undbox{R_i}$ as the restriction of $\undbox{R_h'}$~to $\bset(G_i)$. Then $R$ has one of the shapes in \Cref{fig:bang-minimal} (in the pictures, none among $n_1, \dots, n_{k - 1}$ is a $\drpn$ node for simplicity), depending on the label of $n_0$, with $R_0$ defined by erasing~$n_1, \dots, n_k$ and $R_1, \dots, R_k$ from $R$, and by setting the conclusion of $n_0$ to be a premise of a fresh $\bullet$~node. Moreover, the graph $G_k$ has no node with label in $L$ and, if $h \neq 0$, neither does $G_h$. We~also observe that $\drcl(G_k) = \drcl(G_k') = \emptyset$~and~$\out(G_k) = \out(G_k') \cup \{a_k'\} = \{a_k'\}$,~thus~$G_k \models \Cio$. From this and from $R_h' \models \ACCio$ (which we have by \Cref{thm:poutter-accio-sequential}) we deduce that $R_k \models \ACCio$~and $R_h \models \ACCio$, hence $R_k$ and $R_h$ are \seq by \Cref{thm:poutter-accio-sequential}.
   
   The only thing left to prove is that, if $\drpn \in L$, then $n_i$ is a $\tin$ node for all $i \in \{1, \dots, k - 1\}$, but this is just a straightforward consequence of \Cref{itm:bang-minimal,itm:bang-minimal-tensor} for $R$.
  \end{proof}
  
  The following result is a simpler variant of \Cref{lemma:bang-minimal} which does not use the order $\precio$. The proof, omitted, is an easy induction on $k$ that uses \Cref{thm:poutter-accio-sequential}.
  
  \begin{lemma}
   \label{lemma:bang-initial}
   Let $R \in \ppset$ be \seq, let $k \geq 0$, and let $n_0, \dots, n_k \in \nd(\undgs{R})$. We~suppose that $\undgs{R}$ has no $\parr$ or $\ctpn$ node, and that:
   \begin{enumerate}[(i)]
    \item
     The node $n_0$ is an $\axpn$ or $\oc$ node having the output conclusion of $R$ among its conclusions;
    \item
     For every $i \in \{1, \dots, k\}$, an input conclusion of $n_{i - 1}$ is a premise of $n_i$;
    \item
     For every $i \in \{1, \dots, k - 1\}$, we have that $n_i$ is a $\tin$, $\drpn$ or $\wn$ node;
    \item
     Either $n_k$ is a $\cutpn$ node, or a conclusion of $n_k$ is a conclusion of $R$.
   \end{enumerate}
   Then there exist \seq $R_0, \dots, R_k \in \psset$ such that $R$ has one of the shapes in \Cref{fig:bang-initial} (in the pictures, none among $n_1, \dots, n_k$ is a $\drpn$ or $\wn$ node for simplicity), depending on the label of $n_0$, with $R_0$ defined by erasing $n_1, \dots, n_k$ and $R_1, \dots, R_k$ from $R$, and by setting~the~conclusion of $n_0$ to be a premise of a fresh $\bullet$ node if $k \geq 1$.
   
   \begin{figure}
    \centering
    \subcaptionbox{$l = \axpn$}{
     \centering
     \input{Figures/bang-initial.tikz}
    }
    \subcaptionbox{$l = \oc$}{
     \centering
     \input{Figures/bang-initial-box.tikz}
    }
    \caption{\label{fig:bang-initial}Decomposing a \seq \ps $R$ such that $\undgs{R}$ has no $\parr$ or $\ctpn$ node, using a node $n_0$ with label $l$ having the output conclusion of $R$ among its conclusions. In the pictures, we have $l' \in \{\cutpn, \otimes\}$, the output premise of any $\cutpn$ or $\otimes$ node is the left one, and~$\drpn$~and~$\wn$~nodes~are~omitted.}
   \end{figure}
  \end{lemma}
  
  We now consider the typed setting of $\bangMELL$.
  
  \begin{definition}
   Let $(\Gamma, t) \in \tbcset$. We define $\BL{(\Gamma, t)}$ as the \rnf{\rmuw} of $\WH{(\Gamma, t)}$.
  \end{definition}
  
%
  
  \begin{lemma}
   \label{lemma:why-not}
   Let $t$ be a bang term, let $R \in \bpnset$ with input conclusions $c_1, \dots, c_k, c_{0, 1}, \dots, c_{0, h}$ of types $\wn O_1^\bot, \dots, \wn O_k^\bot, \wn O^\bot, \dots, \wn O^\bot$ respectively. For every $i \in \{1, \dots, h\}$, let $x_i \coloneq \nmfun{R}(c_{0,i})$. Let $\Gamma \coloneq \oftype{\nmfun{R}(c_1)}{\oc O_1}, \dots, \oftype{\nmfun{R}(c_k)}{\oc O_k}$, $\Gamma_0 \coloneq \Gamma, \oftype{x_1}{\oc O}, \dots, \oftype{x_h}{\oc O}$, let $R' \in \bpnset$ defined from $R$ by adding a terminal $\wn$ node with premises $c_{0, 1}, \dots, c_{0, h}$ and conclusion $c$ of type $\wn O^\bot$, let $x \coloneq \nmfun{R'}(c)$, and $\Gamma' \coloneq \Gamma, \oftype{x}{\oc O}$. Then $R = \BL{(\Gamma_0, t)}$ iff $\norm{\wn}{R'} = \BL{(\Gamma', t\{x/x_1, \dots, x/x_h\})}$.
  \end{lemma}
  
  \begin{proof}
   Straightforward induction on $t$.
  \end{proof}
  
  \begin{remark}
   \label{rmk:bang-poutter}
   If $\mathcal{G}$ is a \gst{\bangMELL}, the following holds (recall \Cref{def:polarized}):
   \begin{enumerate}[(i)]
    \item
     Every $\otimes$ node of $\mathcal{G}$ is a $\tin$ node, and its output premise is its left premise;
    \item \label{itm:bang-pout}
     Every $\parr$ node of $\mathcal{G}$ is a $\pout$ node;
    \item \label{itm:bang-poutter}
     If no premise of a $\cutpn$ node of $\mathcal{G}$ is the conclusion of a $\parr$ node of $\mathcal{G}$, then $\mathcal{G}$ is \poutter.
   \end{enumerate}
  \end{remark}
  
  \begin{proposition}
   \label{prop:bang-sequentialization}
   Let $R \in \psf{\bangMELL}$. If $R$ is \rn{\rmu}, then $R$ is a \pn if and only~if~we have $R \models \ACio$ and $\card{\out(\undgs{R})} = 1$.
  \end{proposition}
  
  \begin{proof}
   Straightforward consequence of \Cref{thm:poutter-accio-sequential,rmk:sequential,rmk:imell-output-conclusion,rmk:bang-poutter}.
  \end{proof}
  
  \begin{definition}
   \label{def:derrem}
   Let $R \in \psf{\bangMELL}$. We denote by $\derrem{R}$ the \ps which is obtained from $R$ by removing every terminal $\drpn$ node which does not have a $\oc$, $\drpn$ or $\wn$ node as~its~premise.
  \end{definition}
  
  \begin{remark}
   \label{rmk:derrem}
   Let $R$ be a \nam \pn. By \Cref{lemma:sequential-last-rule}, we know that $\derrem{R}$ is a \pn, and $\derrem{R}$ inherits from $R$ a structure of \nam \pn. Moreover, we have $R^{+-} = R^{-+} = R$.
  \end{remark}
  
  \characterization*
  
  \begin{proof}
   We prove, by induction on $\card{\ar(R)}$, that, if $R$ is \rn{\rmuw}, then there exists a typed \bterm $(\Gamma, t)$ such that $R = \BL{(\Gamma, t)}$. If $\undgs{R}$ has a terminal $\parr$ node $n$, then there~is~$R_1 \in \psf{\bangMELL}$ such that $R$ has the shape in \Cref{subfig:sequential-par}. By \Cref{lemma:sequential-last-rule} (and by \Cref{rmk:sequential}), we know that $R_1$ is a \pn. By definition of \gua \pn, the input premise of $n$ is the~conclusion of a $\oc$, $\drpn$ or $\wn$ node, thus $R_1$ is \gua. We observe that $R_1$ is \injnam if we~extend~$\nmfun{R}$ to the input premise $a$ of $n$ by taking $\nmfun{R_1}(a) \coloneq x$, where $x \in \vset$ is a fresh variable. We~deduce that $R_1 \in \bpnset$, and obviously $R_1$ is \rn{\rmuw}. By \Cref{rmk:bangMELL-output-conclusion}, the type of $a$ is $\wn O^\bot$ for some output \formulaWithName{$O$}{\bangMELL}. Let $\Gamma_1 \coloneq \Gamma, \oftype{x}{\oc O}$. Since $\card{\ar(R_1)} < \card{\ar(R_1)} + 1 = \card{\ar(R)}$,~we~can~apply the induction hypothesis on $R_1$, from which we deduce that there exists a \bt $t_1$~such~that $R_1 = \BL{(\Gamma_1, t_1)}$, and then we are done by taking $t \coloneq \bcabs{x}{\oc O}{t_1}$.
   
   We can now suppose that $\undgs{R}$ has no terminal $\parr$ node. By \Cref{lemma:poutter-terminal-pout}, \Cref{prop:bang-sequentialization}~and \Cref{itm:bang-pout,itm:bang-poutter} of \Cref{rmk:bang-poutter}, we know that $\undgs{R}$ has no $\parr$ node. We define $L \coloneq \{\drpn, \wn\}$. If $\undgs{\derrem{R}}$ has a node with label in $L$ then, since $\undgs{R}$ has no $\parr$ node, we can consider $n_0, \dots, n_k \in \nd(\undgs{\derrem{R}})$ satisfying the hypotheses of \Cref{lemma:bang-minimal}, and thus there exist \pnstWithName{$R_0, \dots, R_k$}{\bangMELL} such that $\derrem{R}$ has the shape in \Cref{subfig:bang-minimal}, where $R_0$ is obtained from $\derrem{R}$ by erasing~$n_1, \dots, n_k$ and $R_1, \dots, R_k$, and by setting the conclusion of the node $n_0$ to be a premise of a fresh~$\bullet$~node~if $k \geq 1$. By \Cref{rmk:derrem}, we know that $R$ has the same shape as $\derrem{R}$: we obtain $R$ from $\derrem{R}$~by replacing $R_0, \dots, R_k$ with $\deradd{R_0}, \dots, \deradd{R_k}$.
   
   Suppose $k = 0$. We have $\deradd{R_0} = R$, and we can suppose without loss of generality that the conclusion of $n_0$ is $c_1$. Let $c_{1, 1}', \dots, c_{1, h}'$ be the premises of $n_0$, and let $R' \in \bpnset$ obtained~from $R$ by erasing $n_0$, by setting $c_{1, i}'$ to be a premise of a fresh $\bullet$ node and $\nmfun{R'}(c_{1, i}) \coloneq x_i$,~where~$x_i \in \vset$ is a fresh variable, for all $i \in \{1, \dots, h\}$. Since $\card{\ar(R')} < \card{\ar(R')} + 1 = \card{\ar(R)}$, we can~apply~the induction hypothesis on $R'$. Let $x \coloneq \nmfun{R}(c_1)$. We consider the two following cases:
   \begin{itemize}
    \item
     \emph{$n_0$ is a $\drpn$ node.} Then $h = 1$. Since $n_0 \in \nd(\undgs{\derrem{R}})$, by definition of $\derrem{R}$ we know that $c_{1, 1}'$ is  a conclusion of a $\oc$, $\drpn$ or $\wn$ node. Moreover, since $R \in \psf{\bangMELL}$, we have that~$c_{1, 1}'$~is~input,~thus the type of $c_{1, 1}'$ has the shape $\wn O^\bot$ for some output \formulaWithName{$O$}{\bangMELL}. We now consider~the \bctx $\Gamma_0 \coloneq \oftype{x_1}{\oc O}, \oftype{\nmfun{R}(c_2)}{\oc O_2}, \dots, \oftype{\nmfun{R}(c_k)}{\oc O_k}$, and then we apply the induction hypothesis on $R'$ to obtain a \bt $t_0$ such that $R' = \BL{(\Gamma_0, t_0)}$. We are done by~defining $t \coloneq \bcapp{\bcabs{x_1}{\oc O}{t_0}}{x}$;
    \item
     \emph{$n_0$ is a $\wn$ node.} Let $\Gamma' \coloneq \oftype{\nmfun{R}(c_2)}{\oc O_2}, \dots, \oftype{\nmfun{R}(c_k)}{\oc O_k}$ and $\Gamma_0 \coloneq \Gamma', \oftype{x_1}{\oc O_1}, \dots, \oftype{x_h}{\oc O_1}$. Then $\Gamma = \Gamma', \oftype{x}{\oc O_1}$. By induction hypothesis on $R'$, there exists a \bt $t_0$ such~that $R' = \BL{(\Gamma_0, t_0)}$. By \Cref{lemma:why-not}, we are done by taking $t \coloneq t_0\{x/x_1, \dots, x/x_h\}$.
   \end{itemize}
   
   Now suppose $k \geq 1$. Observe that, for every $i \in \{1, \dots, k\}$, the output premise of $n_i$ is the output conclusion of $\deradd{R_i}$, and recall that $n_0$ was chosen as a $\drpn$ or $\wn$ node. For~all~$i \in \{0, \dots, k\}$, we thus obtain that $\deradd{R_i}$ is \gua. Let $c_{0, 0}$ be the conclusion of $n_0$. By \Cref{lemma:why-not}, for all $i \in \{0, \dots, k\}$, we may add terminal $\wkpn$ nodes to $\deradd{R_i}$ without loss of~generality,~in~such~a~way~that:
   \begin{itemize}
    \item
     We have $\deradd{R_0} \in \bpnset$ with input conclusions $c_{0, 0}, c_{0, 1}, \dots, c_{0, k}$;
    \item
     For every $i \in \{1, \dots, k\}$, we have $\deradd{R_i} \in \bpnset$ with input conclusions $c_{i, 1}, \dots, c_{i, k}$;
    \item
     For every $i \in \{0, \dots, k\}$, $j \in \{1, \dots, k\}$, the type of $c_{i, j}$ is $\wn O_j^\bot$ and $\smash{\nmfun{\deradd{R_i}}(c_{i, j}) = \nmfun{R}(c_j)}$;
    \item
     The type of $c_{0, 0}$ is $\wn O^\bot$ for some \formulaWithName{$O$}{\bangMELL}, and $\nmfun{\deradd{R_0}}(c_{0, 0}) = x$, where $x \in \vset$ is a fresh variable.
   \end{itemize}
   Let $\Gamma_0 \coloneq \Gamma, \oftype{x}{\oc O}$. For each $i \in \{0, \dots, k\}$, we have $\card{\ar(\deradd{R_i})} < \card{\ar(\deradd{R_i})} + 1 \leq \card{\ar(R)}$, and thus we can apply the induction hypothesis on $\deradd{R_i}$. We deduce that there exists a \bt $t_0$ such that $\deradd{R_0} = \BL{(\Gamma_0, t_0)}$ and, for each $i \in \{1, \dots, k\}$, there is a \bt $t_i$~such~that~$\deradd{R_i} = \BL{(\Gamma, t_i)}$. We distinguish two cases:
   \begin{itemize}
    \item
     \emph{$n_k$ is a $\cutpn$ node.} By the hypothesis that $R$ is \rn{\rmuw}, we obtain $k = 1$, and thus~we~are done by taking $t \coloneq \bcapp{\bcabs{x}{\oc O}{t_0}}{t_1}$;
    \item
     \emph{$n_k$ is a $\otimes$ node.} Let $c$ be the conclusion of the $\drpn$ node of $\undgs{R}$ having the conclusion of $n_k$ as its premise, and let $y \coloneq \nmfun{R}(c_1)$. Then we are done by defining $t \coloneq \bcapp{\bcabs{x}{\oc O}{t_0}}{\bcapp{\cdots \bcapp{y}{t_k} \cdots}{t_1}}$.
   \end{itemize}
   
   Finally, we suppose $\undgs{\derrem{R}}$ has no $\parr$, $\drpn$ or $\wn$ node. We can then consider $n_0, \dots, n_k \in \nd(\undgs{\derrem{R}})$ satisfying the hypotheses of \Cref{lemma:bang-initial} and such that $k$ is the maximum non-negative integer for which that choice of nodes is possible. Thus, there are \pnstWithName{$R_0, \dots, R_k$}{\bangMELL} such that $\derrem{R}$ has one of the shapes in \Cref{fig:bang-initial}, depending on the label of $n_0$, with $R_0$ defined by erasing $n_1, \dots, n_k$ and $R_1, \dots, R_k$~from~$\derrem{R}$,~and~by~setting~the~conclusion~of~$R_0$~to~be~a~premise of a fresh $\bullet$ node if $k \geq 1$. Observe that $R$ has the same shape as $\derrem{R}$ by \Cref{rmk:derrem}:~we~obtain $R$ from $\derrem{R}$ by replacing $R_0, \dots, R_k$ with $\deradd{R_0}, \dots, \deradd{R_k}$.
   
   Suppose that $k = 0$ and that $n_0$ is a $\oc$ node. We have $\deradd{R_0} = R$ and, by maximality of $k$, the node $n_0$ is terminal. Define $R' \coloneq \undbox{R}(n_0)$. Since $\card{\ar(R')} < \card{\ar(R')} + 1 \leq \card{\ar(R)}$,~we~can~apply the induction hypothesis on $R'$, from which it follows that there exists a \bt $t_0$ such~that $R' = \BL{(\Gamma, t_0)}$, and then we are done by taking $t \coloneq \bcbang{t_0}$.
   
   We now suppose either $n_0$ is an $\axpn$ node, or $k \geq 1$ and $n_0$ is a $\oc$ node. For all $i \in \{1, \dots, k\}$, since the output premise of $n_i$ is the output conclusion of $\deradd{R_i}$, we know that $\deradd{R_i}$ is \gua. By \Cref{lemma:why-not}, we may add terminal $\wkpn$ nodes to $\deradd{R_i}$ without loss of generality, in such~a~way~that:
   \begin{itemize}
    \item
     We have $\deradd{R_i} \in \bpnset$ with input conclusions $c_{i, 1}, \dots, c_{i, k}$;
    \item
     For every $j \in \{1, \dots, k\}$, the type of $c_{i, j}$ is $\wn O_j^\bot$ and $\nmfun{\deradd{R_i}}(c_{i, j}) = \nmfun{R}(c_j)$.
   \end{itemize}
   Since $\card{\ar(\deradd{R_i})} < \card{\ar(\deradd{R_i})} + 1 \leq \card{\ar(R)}$, we can apply the induction hypothesis on $\deradd{R_i}$ to obtain a \bt $t_i$ such that $\deradd{R_i} = \BL{(\Gamma, t_i)}$.
   
   If $n_0$ is an $\axpn$ node, then $n_k$ is not a $\cutpn$ node by the hypothesis that $R$ is \rn{\rmuw}. We can then consider the conclusion $c$ of the $\drpn$ node of $\undgs{R}$ having a conclusion of $n_k$~as~its~premise, and set $x \coloneq \nmfun{R}(c)$. We are done by taking $t \coloneq \bcapp{\cdots \bcapp{x}{t_k} \cdots}{t_1}$.
   
   We can now suppose that $k \geq 1$ and that $n_0$ is a $\oc$ node. Let $c_{0, 0}$ be the conclusion of $n_0$ that is a premise of $n_1$. We know that $\deradd{R_0}$ is \gua and, by \Cref{lemma:why-not},~we~may~add~terminal~$\wkpn$ nodes to $\deradd{R_0}$ without loss of generality, in such a way that:
   \begin{itemize}
    \item
     We have $\deradd{R_0} \in \bpnset$ with input conclusions $c_{0, 0}, c_{0, 1}, \dots, c_{0, k}$;
    \item
     For every $j \in \{1, \dots k\}$, the type of $c_{0, j}$ is $\wn O_j^\bot$ and we have $\smash{\nmfun{\deradd{R_0}}(c_{0, j}) = \nmfun{R}(c_j)}$;
    \item
     The type of $c_{0, 0}$ is $\wn O^\bot$ for some \formulaWithName{$O$}{\bangMELL}, and $\nmfun{\deradd{R_0}}(c_{0, 0}) = x$, where $x \in \vset$ is a fresh variable.
   \end{itemize}
   Let $\Gamma_0 \coloneq \Gamma, \oftype{x}{\oc O}$. Since $k \geq 1$, we have $\card{\ar(\deradd{R_0})} < \card{\ar(\deradd{R_0})} + 1 \leq \card{\ar(R)}$. Thus, we can apply the induction hypothesis on $\deradd{R_0}$ to obtain a \bt $t_0$ such that $\deradd{R_0} = \BL{(\Gamma_0, t_0)}$. As before, we conclude by separating the case in which $n_k$ is a $\cutpn$ node from the one in~which~$n_k$~is~a~$\otimes$ node.
  \end{proof}